\documentclass[12pt]{thesis}  

\usepackage[utf8]{inputenc}
\DeclareUnicodeCharacter{03C3}{$\sigma$}

\usepackage{titlesec}
   \titleformat{\chapter}
      {\normalfont\large}{Chapter \thechapter:}{1em}{}
\usepackage{amsmath,amsfonts}
\usepackage{graphicx}
\usepackage{cite}
\usepackage{slashed}
\usepackage{lscape}
\usepackage{indentfirst}
\usepackage{latexsym}
\usepackage{multirow}
\usepackage{tabls}
\usepackage{wrapfig}
\usepackage{longtable}
\usepackage{supertabular}
\usepackage{subfigure}
\usepackage{qcircuit}
\usepackage{amsmath}
\usepackage{amsthm}
\usepackage{comment}
\usepackage{tikz}
\usepackage[onehalfspacing]{setspace} 

\usepackage[colorlinks=true,urlcolor=black,linkcolor=blue,citecolor=blue]{hyperref}

\renewcommand{\baselinestretch}{2}
\renewcommand{\d}{\mathrm{d}}
\newcommand{\ket}[1]{\left|#1\right>}

\newtheorem{theorem}{Theorem}

\DeclareMathOperator{\Tr}{Tr}

\renewcommand{\Re}{\mathrm{Re}\,}
\renewcommand{\Im}{\mathrm{Im}\,}
\begin{document}
\pagestyle{empty}

\hbox{\ }

\renewcommand{\baselinestretch}{1}
\small \normalsize

\begin{center}
\large{{ABSTRACT}} 

\vspace{3em} 

\end{center}
\hspace{-.15in}
\begin{tabular}{ll}
Title of dissertation:    & {\large SIGN PROBLEMS IN QUANTUM}\\ &{\large FIELD THEORY: CLASSICAL}\\ &{\large AND QUANTUM APPROACHES }\\
\ \\
&                          {\large  Scott Lawrence}\\
&                       {\large Doctor of Philosophy, 2020} \\
\ \\
Dissertation directed by: & {\large  Professor Paulo F.\ Bedaque} \\
&  				{\large	 Department of Physics } \\
\end{tabular}

\vspace{3em}

\renewcommand{\baselinestretch}{2}
\large \normalsize

Monte Carlo calculations in the framework of lattice field theory provide nonperturbative access to the equilibrium physics of quantum fields. When applied to certain fermionic systems, or to the calculation of out-of-equilibrium physics, these methods encounter the so-called sign problem, and computational resource requirements become impractically large. These difficulties prevent the calculation from first principles of the equation of state of quantum chromodynamics, as well as the computation of transport coefficients in quantum field theories, among other things.

This thesis details two methods for mitigating or avoiding the sign problem. First, via the complexification of the field variables and the application of Cauchy's integral theorem, the difficulty of the sign problem can be changed. This requires searching for a suitable contour of integration. Several methods of finding such a contour are discussed, as well as the procedure for integrating on it. Two notable examples are highlighted: in one case, a contour exists which entirely removes the sign problem, and in another, there is provably no contour available to improve the sign problem by more than a (parametrically) small amount.

As an alternative, physical simulations can be performed with the aid of a quantum computer. The formal elements underlying a quantum computation --- that is, a Hilbert space, unitary operators acting on it, and Hermitian observables to be measured --- can be matched to those of a quantum field theory. In this way an error-corrected quantum computer may be made to serve as a well controlled laboratory. Precise algorithms for this task are presented, specifically in the context of quantum chromodynamics.


\thispagestyle{empty}
\hbox{\ }
\vspace{1in}
\renewcommand{\baselinestretch}{1}
\small\normalsize
\begin{center}

\large{SIGN PROBLEMS IN QUANTUM FIELD THEORY:}\\\large{CLASSICAL AND QUANTUM APPROACHES}\\
\ \\
\ \\
\large{by} \\
\ \\
\large{Scott Lawrence}
\ \\
\ \\
\ \\
\ \\
\normalsize
Dissertation submitted to the Faculty of the Graduate School of the \\
University of Maryland, College Park in partial fulfillment \\
of the requirements for the degree of \\
Doctor of Philosophy \\
2020
\end{center}

\vspace{7.5em}

\noindent Advisory Committee: \\
Professor Andrew Baden\\
Professor Paulo F.\ Bedaque, Chair/Advisor \\
Professor Zackaria Chacko\\
Professor Massimo Ricotti\\
Professor Raman Sundrum

\pagestyle{plain} \pagenumbering{roman} \setcounter{page}{2}
\addcontentsline{toc}{chapter}{Acknowledgements}

\renewcommand{\baselinestretch}{2}
\small\normalsize
\hbox{\ }
 
\vspace{-.65in}

\begin{center}
\large{Acknowledgments} 
\end{center} 

\vspace{1ex}

The work described in this thesis was performed in collaboration with others at the University of Maryland and the George Washington University: Andrei Alexandru, Paulo Bedaque, Siddhartha Harmalkar, Neill Warrington, and Yukari Yamauchi. Most especially, credit is due to Henry Lamm, who aside from being an excellent collaborator, is unfailing generous with his time.

That this thesis exists at all is down to the aid and encouragement of a long series of mentors, who helped me along throughout undergraduate and graduate school. In chronological order, I am deeply indebted to: Drew Baden, Raman Sundrum, Victor Yakovenko, Cole Miller, Tom Cohen, and of course my Ph.D.\ advisor Paulo Bedaque.

Lastly, for their ceaseless encouragement throughout this decade, I am grateful to Brian McPeak, Ophir Lifshitz, and Yukari Yamauchi.

    \cleardoublepage
    \phantomsection
    \addcontentsline{toc}{chapter}{Table of Contents}
    \renewcommand{\contentsname}{Table of Contents}
\renewcommand{\baselinestretch}{0.98}
\small\normalsize
\tableofcontents 
\renewcommand{\baselinestretch}{1}
\newpage
\listoffigures 
\newpage

\newpage
\setlength{\parskip}{0em}
\renewcommand{\baselinestretch}{2}
\small\normalsize

\setcounter{page}{1}
\pagenumbering{arabic}
\chapter{Introduction}

The majority of observed non-gravitational phenomena in laboratories and the universe are believed to be described by the quantum field theory of the standard model. The standard model can be crudely divided into two sectors, one governed by the electroweak force and the other by the strong force. The electroweak force posesses a small expansion parameter. As a result, the physics of electroweak phenomena can be computed in perturbation theory. In the other sector, governed by the strong force, there is no small expansion parameter available.

The physics of the strong force governs phenomena ranging from heavy ion collisions to the structure of neutron stars. Furthermore, various extensions of the standard model (of relevance, for instance, in the search for dark matter~\cite{degrand2018composite}) also involve no small expansion parameter; indeed there is no \emph{a priori} reason why we should expand physics beyond the standard model to be perturbative. Computational methods capable of probing non-perturbative physics are therefore critical not only for understanding the physics of the standard model, but also for designing tests for beyond-standard-model physics. 

The main nonperturbative tool for field theory is the lattice. Whereas a quantum field theory, as usually understood, involves arbitrarily small distance scales and therefore arbitrarily high momenta (and with them, a host of technical complications), the lattice regularization removes momenta above a certain cutoff, reducing the physics of field theory to the task of evaluating a finite-dimensional (albeit the number of dimensions is large) integral. The lattice regularization, in addition to being a central subject of the formal mathematical study of field theories~\cite{douglas2004report}, is therefore a prime candidate for \emph{numerical} study, the subject of this thesis.

The problems of lattice field theory are problems of evaluating high-dimensional integrals. Monte Carlo simulations of lattice field theory provide the main tool by which nonperturbative physics may be accessed in a practical way~\cite{gattringer2009quantum,montvay1997quantum}. In certain regimes, known algorithms for simulating lattice field theory encounter severe computation obstacles. The most prominent of these is the \emph{fermion sign problem}: when simulating field theories with a finite density of fermions, evaluating the integral by usual methods requires resolving fine cancellations of large terms. This prevents calculations, for example, of the equation of state of a neutron star. Another sign problem occurs when applying lattice methods to determine transport coefficients of quantum fluids, restricting the ability to compare the experimentally observed behavior of heavy ion collisions to theoretical calculations.

This thesis concerns the task of performing nonperturbative computations in regimes where Monte Carlo methods encounter a sign problem. One approach is to re-structure the integral in such a way as to alleviate the sign problem. We will see in Chapter~\ref{ch:complexification} that the methods of complex analysis --- in particular, an application of a multidimensional generalization of Cauchy's integral theorem --- may be used to construct such algorithms. The rapid progress being made on the construction of practical quantum computers~\cite{monroe2014large,preskill2018quantum} suggests another promising avenue of attack. As a quantum computer is, in fact, a quantum-mechanical system, it is an ideal device to simulate other quantum-mechanical systems~\cite{lloyd1996universal,feynman1999simulating}. This suggests the possibility of using a quantum computer as a tool for simulating lattice quantum field theory~\cite{Jordan:2011ci}.

This thesis is structured in three broad parts. Chapters~\ref{ch:lattice} and \ref{ch:difficulties} introduce lattice field theory, the standard context in which nonperturbative quantum field theory computations are performed, and describe the ways in which sign problems obstruct calculations performed at finite fermion density or in real time. Chapter~\ref{ch:complexification} describes several closely related methods for thwarting the sign problem on classical computers, based on the complexification of the path integral. These methods are applicable to a wide range of theories, but the Thirring model (in both $1+1$ and $2+1$ dimensions) is investigated in detail. Finally in Chapter~\ref{ch:quantum}, quantum computing is introduced as a tool for the simulation of field theories, which entirely circumvents the sign problem. The necessary algorithms for simulating QCD as a lattice gauge theory are detailed, and the costs (in terms of quantum gates and qubits) are estimated.

\chapter{Two Views of Lattice Field Theory}\label{ch:lattice}

In this chapter we describe lattice-regularized field theory from two perspectives: first the Hamiltonian formulation, in which time is continuous and the field theory is a regular quantum mechanical system, and second from a lattice action, where space and time are both discretized, as is appropriate for a relativistic theory.

\section{Hamiltonian Lattice Field Theory}

We begin our overview of the Hamiltonian formulation of lattice field theory with the example of a noninteracting real scalar field~\cite{Peskin:1995ev}. The Hamiltonian of a noninteracting scalar field in the continuum may be written in momentum space as
\begin{equation}\label{eq:scalar-continuum-momentum}
H = \sum_k \frac 1 2 \Pi_k^2 + \frac{\sqrt{m^2 + k^2}}{2}\Phi_k^2
\text,
\end{equation}
where $(\Phi_k, \Pi_k)$ are the position and conjugate momentum coordinates of harmonic oscillators indexed by momentum $k$.

This form makes clear the noninteracting structure of the theory. Each momentum mode can contain any non-negative number of particles, and in a mode with momentum $k$, each particle contributes an energy of $\sqrt{m^2 + k^2}$, where $m$ is the mass of the field. The possible momenta in a box of side length $L$ are
$k = \frac{2 \pi n}{L}$ where $n = 0,1,2,\ldots$. The same theory may also be rewritten in position space, yielding the Hamiltonian
\begin{equation}\label{eq:scalar-continuum}
H = \int \d^d x\;
\frac 1 2 \pi(x)^2 + \frac 1 2 (\partial_i \phi(x))(\partial_i \phi(x))+ \frac{m^2}{2} \phi(x)^2
\text,
\end{equation}
where the spatial index is summed over $i=1,2,3$.

This field theory is `free' in the sense that the partition function factorizes into a product of factors, each involving only one momentum mode. Introducing an interaction directly in the continuum field theory is technically difficult. In the example discussed here of a real scalar field, it is believed that no interacting theory can be constructed in three spatial dimensions~\cite{Wolff:2009ke}. Moreover, the continuum formulation, with an infinite number of degrees of freedom even in a finite-sized box, isn't amenable to numerical simulation.

For these reasons, we introduce the lattice regularization of the field theory. In the case of scalar field theory, this regularization is obtained by placing a cutoff on the momenta $k$ included in the sum of (\ref{eq:scalar-continuum-momentum}). Equivalently, and more conveniently, the integral in (\ref{eq:scalar-continuum}) is changed to a discrete sum over lattice sites, and the derivative is changed to a finite difference. The resulting lattice hamiltonian is
\begin{equation}\label{eq:scalar-lattice}
H = \sum_x
\left[
\frac 1 2 \pi_x^2
+
\frac{m^2}{2} \phi_x^2
\right]
+ 
\sum_{\left<xy\right>}
\frac 1 2 \left(\phi_x - \phi_y\right)^2
\text,
\end{equation}
where the second sum denotes the sum over all pairs of neighboring lattice sites.
It is this object (and others like it) that will serve as the starting point for numerical work. At this point, an interaction is introduced by adding a term proportional to $\phi(x)^4$ to the Hamiltonian\footnote{Any polynomial of even degree greater than $2$ will do.}.

The original field theory consisted of non-interacting harmonic oscillators, each associated to a separate momentum mode. Before introducing an interaction, the same decomposition could be done for the Hamiltonian (\ref{eq:scalar-lattice}), now with a finite number of momentum modes. However, it will be useful to take an alternate view of the system, where each \emph{position} is associated to a harmonic oscillator. Even in the absence of an interaction, these oscillators are coupled by the finite difference term. When the interaction is introduced, the oscillators become anharmonic. Allowing for an arbitrary potential $V(\phi)$, the final lattice hamiltonian is
\begin{equation}\label{eq:h-scalar-interacting}
H = \sum_x
\left[
\frac 1 2 \pi_x^2
+
V(\phi_x)
\right]
+ 
\sum_{\left<xy\right>}
\frac 1 2 \left(\phi_x - \phi_y\right)^2
\text.
\end{equation}
On a lattice with $N$ sites, the Hilbert space is the tensor product of $N$ copies of that of the harmonic oscillator. The operators $\phi_x$ and $\pi_x$ act on the portion of the Hilbert space associated to site $x$, and commute with $\phi_y$ and $\pi_y$ at any site $y \ne x$. Moreover, these operators can be written in terms of the raising and lowering operators associated to a site:
\begin{equation}
\phi_x = \sqrt{\frac{1}{2m}}\left(a^\dagger_x + a_x\right)
\;\text{ and }\;
\pi_x = i\sqrt{2m}\left(a^\dagger_x - a_x\right)
\text.
\end{equation}

\subsection{Fermionic Theories}

For a theory of lattice fermions, due to Pauli exclusion, at most one fermion will be permitted per degree of freedom. (There may be multiple degrees of freedom per lattice site; for example, a lattice site may be occupied by both a spin-up and a spin-down fermion.) The Hilbert space of one degree of freedom, therefore, is two-dimensional, and the Hilbert space of the full lattice is again the tensor product of $N$ copies of that local Hilbert space.

The two-dimensional fermionic Hilbert space associated to degree of freedom $i$ is acted on by creation and annihilation operators $a^\dagger_i$ and $a_i$, respectively. These operators obey the anticommutation relations
\begin{equation}
\{a_i, a_j\} = 0
\,\text{ and }\,
\{a^\dagger_i, a_i\} = \delta_{ij}
\text.
\end{equation}

Most theories of physical interest have multiple fermionic degrees of freedom per lattice site. Typically these degrees of freedom are related by some global symmetry.

A typical theory of lattice fermions is given by the Hamiltonian
\begin{equation}\label{eq:h-hubbard}
H =
- t \sum_{\left<xy\right>,s}
\left(
a^\dagger_{xs} a_{ys} + a^\dagger_{ys} a_{xs}
\right)
- U
\sum_x
a^\dagger_{x\downarrow} a_{x\downarrow}
a^\dagger_{x\uparrow} a_{x\uparrow}
\text,
\end{equation}
where each lattice site (indexed by $x$ or $y$) contains a spin-up and spin-down degree of freedom, and the first sum is taken over all pairs of neighboring lattice sites, and $s=\downarrow,\uparrow$.
This Hamiltonian describes the nonrelativistic fermions of the Hubbard model~\cite{hubbard1963electron}, a common target of numerical work~\cite{scalapino2007numerical}. The Hubbard model is often used as a test problem for sign-problem-mitigating techniques~\cite{Ulybyshev:2019hfm,Ulybyshev:2019fte}, although we will not discuss it in that capacity here.

Due to the nature of the dispersion relation of lattice fermions, the number of low-energy modes of such a Hamiltonian will generically be larger than naive counting would suggest~\cite{nielsen1981no}. Various methods for removing these modes exist~\cite{gattringer2009quantum,montvay1997quantum}; alternatively, one may simply accept that the theory being simulated has more particles than originally intended.

\section{Lattice Actions}

Quantum field theories can be described by a path integral; in this section we derive the lattice path integral. We begin by considering the thermal properties of a lattice field theory at temperature $T \equiv 1/\beta$, defined by the hermitian operator $\rho = e^{-\beta H}$, termed the \emph{density matrix}. The thermal partition function is given by $Z = \Tr\rho$, and thermal expectation values are given by various derivatives of $\log Z$. 

To obtain the path integral, we expand the trace by summing over all possible intermediate states:
\begin{equation}
Z = \Tr \left(e^{-\Delta t H}\right)^{N_t}
=
\sum_{\Psi_1,\Psi_2,\cdots}
\left<\Psi_1\right|e^{-\Delta t H}\left|\Psi_2\right>
\cdots
\left<\Psi_N\right|e^{-\Delta t H}\left|\Psi_1\right>
\text,
\end{equation}
where $\Delta_t N_t = \beta$, and use has been made of the completeness relation
$I = \sum_\Psi \left|\Psi\right>$. The operator $e^{-\Delta t H}$ is teremd the \emph{transfer matrix}.

In the case of the bosonic theory defined by Hamiltonian (\ref{eq:h-scalar-interacting}), an appropriate set of states is given by the simultaneous eigenstates of the field operators $\psi(x)$, denotes $\left|\psi\right>$. The resulting resolution of the identity is
\begin{equation}
I = \int \d^V\phi\; \left|\phi\right>\left<\phi\right|\text.
\end{equation}
It is natural to take $\Delta_t$ to be equal to $1$, and after approximation by the Suzuki-Trotter decomposition~\cite{suzuki1976generalized,trotter1959product}, the resulting lattice path integral is
\begin{equation}\label{eq:scalar-path-integral}
Z
=
\int \d^{\beta V}\phi\;
\exp 
\left\{
- \sum_{\left<xy\right>}
\frac{\left(\phi_x - \phi_y\right)^2}{2}
-
\sum_x
V(\phi_x)
\right\}
\text.
\end{equation}

The continuum limit of (\ref{eq:scalar-path-integral}) is approached by tuning the potential $V(\phi)$ such that correlation functions (e.g. $\left<\phi(x)\phi(0)\right>$) decay slowly with $\left|x\right|$. As mentioned previously, the process of taking a continuum limit will not be of much interest here; however it must be noted that the continuum limit obtained in this way need not be the same as that obtained by working with the Hamiltonian (\ref{eq:h-scalar-interacting}) directly. In order for the two methods to be equivalent, we must first take $\Delta t\rightarrow 0$ in the path integral (termed the Hamiltonian limit) and only then performing the tuning of the potential.

There is always some ambiguity in constructing the path integral: for instance, the Suzuki-Trotter decomposition is not unique, and there are many possible resolutions of the idenitty to be used. If we are interested in the Hamiltonian limit of the lattice theory, then this is not an issue at all. In fact, in that limit, the multiple lattice path integrals simply correspond to distinct classical algorithms for computing thermodynamic quantities. Difficulties can potentially arise when taking a continuum limit directly from the lattice theory, as there is generally no guarantee that all of the possible lattice path integrals will yield the same continuum limit.

\subsection{Fermionic Path Integral}
The fermionic path integral is derived in the same way as (\ref{eq:scalar-path-integral}) above, but with a particular choice of completeness relation. Consider first a single fermionic degree of freedom. Define the state $\left|0\right>$ to be the eigenstate of $a^\dagger a$ with eigenvalue $0$ --- this is the occupation number basis. We introduce the coherent state $\left|\psi\right>$ and its dual $\left<\psi\right|$, defined by
\begin{equation}
\left|\psi\right>\equiv
e^{-\psi a^\dagger} \left|0\right>
\,\text{ and }\,
\left<\psi\right|\equiv
\left<0\right| e^{-a \psi^\dagger}
\text,
\end{equation}
where $\psi$ and $\psi^\dagger$ are independent Grassmann variables (see Appendix~\ref{ap:grassmann}). With these definitions, a new completeness relation is available:
\begin{equation}\label{eq:grassmann-completeness}
I =
\int \d \psi^\dagger \d \psi\;
e^{-\psi^\dagger \psi}
\left|\psi\right>\left<\psi\right|
\text.
\end{equation}
For a lattice with multiple fermionic degrees of freedom, we introduce a pair of grassman variables $\psi_i$ and $\psi_i^\dagger$ for each. The completeness relation for the full lattice system is
\begin{equation}\label{eq:grassmann-completeness-2}
I = \int
\prod_i \left(\d \psi_i^\dagger
\d \psi_i
\right)
e^{- \sum_i \psi_i^\dagger \psi_i}
\left|\psi_1\cdots\psi_N\right>\left<\psi_1\cdots\psi_N\right|
\end{equation}
where the products and sums run over all $N$ fermionic degrees of freedom $i$.
At this point the derivation of the path integral proceeds just as for a scalar field theory, using the completeness relation of (\ref{eq:grassmann-completeness-2}).

\section{Gauge Theories}\label{sec:hamiltonian-gauge}

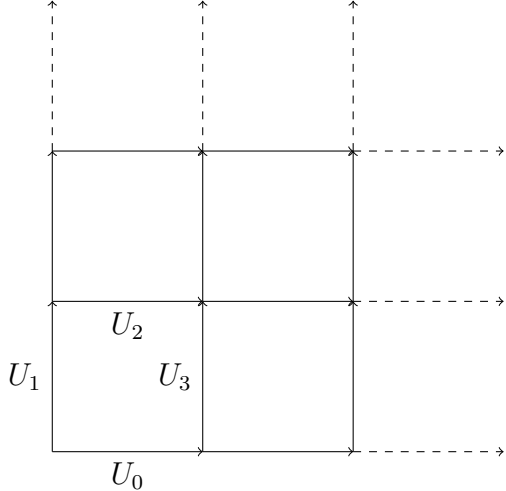
\begin{figure}
\centering
\begin{tikzpicture}
\draw[->] (0,0) -- node[below] {$U_0$} (2,0);
\draw[->] (0,0) -- node[left] {$U_1$} (0,2);
\draw[->] (0,2) -- node[below] {$U_2$} (2,2);
\draw[->] (2,0) -- node[left] {$U_3$} (2,2);
\draw[->] (2,0) -- (4,0);
\draw[dashed,->] (4,0) -- (6,0);
\draw[->] (2,2) -- (4,2);
\draw[dashed,->] (4,2) -- (6,2);
\draw[->] (4,0) -- (4,2);
\draw[->] (0,2) -- (0,4);
\draw[dashed,->] (0,4) -- (0,6);
\draw[->] (2,2) -- (2,4);
\draw[dashed,->] (2,4) -- (2,6);
\draw[->] (0,4) -- (2,4);
\draw[->] (2,4) -- (4,4);
\draw[->] (4,2) -- (4,4);
\draw[dashed,->] (4,4) -- (4,6);
\draw[dashed,->] (4,4) -- (6,4);
\end{tikzpicture}
\caption{\label{fig:gauge-lattice}The lattice associated to a gauge theory in two dimensions. Each link is an independent (commuting) degree of freedom. The lower-leftmost plaquette is given by $\Re \Tr U_3^\dagger U_2 U_1 U_0^\dagger$.}
\end{figure}

We now return to the Hamiltonian picture. Some Hamiltonians posess locally conserved charges. Perhaps the simplest example is $\mathbb Z_2$ gauge theory introduced by Wegner~\cite{Wegner:1984qt}. This theory (like other lattice gauge theories we will discuss) lives on the lattice shown in Figure~\ref{fig:gauge-lattice}; to each link $\ell$ is associated a two-dimensional Hilbert space acted on by local Pauli operators $\sigma_x(\ell)$,$\sigma_y(\ell)$,$\sigma_z(\ell)$. The operators at separate links are mutually commuting. The Hamiltonian of the theory is
\begin{equation}\label{eq:z2-hamiltonian}
H = \sum_{\ell} \sigma_x(\ell)+ \sum_{P} \sigma_z(P_1) \sigma_z(P_2) \sigma_z(P_3)\sigma_z(P_4)
\text,
\end{equation}
where the first sum is taken over all links in the lattice, and the second is taken over all `plaquettes' consisting of four links arranged in a square.

For each site $r$ of the lattice, this Hamiltonian commutes with the operator
\begin{equation}
G(r) = \prod_{\ell\text{ at }r}
\sigma_x(r)
\text,
\end{equation}
where the product is taken over all links which have an endpoint at $r$.
The action of this operator is called a gauge transformation; in the $Z$ basis, it has the effect of flipping any link in contact with $r$ and leaving all other links invariant.

The group of symmetries associated to gauge transformations, for this theory, is $(\mathbb Z_2)^{V}$, where $V$ is the number of sites on the lattice.

A local conservation law is in some sense more restrictive than a global symmetry. Consider a scattering experiment: we begin with a vacuum, introduce particles `at infinity' (at the boundary of a large box), and then observe what particles are measured at later times, again `at infinity'. This is typical of physical experiments, in which the experimentalist can only act on the boundary of the laboratory. Such experiments can introduce charged particles into the theory, and thus explore the sectors of Hilbert space labelled by different global charges; however, locally conserved charges cannot be introduced this way. Thus, only one sector of Hilbert space can be considered physically relevant.

The `physical sector' of Hilbert space is taken to be the space of vectors invariant under gauge transformations; that is, the set of vectors that are simultaneous eigenvectors of all $G(r)$, with eigenvalue $1$. The remainder of the full Hilbert space is physically irrelevant. We consider it to exist only for convenience in writing down the Hamiltonian (\ref{eq:z2-hamiltonian}).

\subsection{General Gauge Groups}
In general, a gauge theory can be defined for any group $G$ chosen to be the gauge group. Lattice gauge theories for continuous groups were initially formulated from the path integral; here we begin with the Hamiltonian formulation~\cite{kogut1975hamiltonian} and obtain the path integral as a consequence.

As with the $\mathbb Z_2$ gauge theory, we define the gauge theory for gauge group $G$ on a rectangular (or cubic) lattice, such as shown in Figure~\ref{fig:gauge-lattice}. To each link $\ell$is associated a local Hilbert space $\mathcal H_\ell = \mathbb C G$, the space of square-integrable complex-valued functions on $G$. A large Hilbert space $\mathcal H = \mathcal H_\ell^{\otimes N}$ is constructed from the tensor product of $N$ copies of $\mathcal H_\ell$, where there are $N$ total links on the lattice.

A local Hilbert space $\mathbb C G$ has a basis\footnote{Strictly speaking, for a continuous group $G$, this set of vectors is overcomplete and does not lie in the Hilbert space, as they are not normalizable. Much as for ordinary quantum mechanics, where eigenstates of the position operator are referred to as $\left|x\right>$, we can disregard this issue without affecting any of the results. It is most visible in the fact that the dimension of the Hilbert space of a lattice is countably infinite, whereas the set of basis vectors we use is uncountably infinite.} consisting of one state $\left|U\right>$ for every element $U \in G$. The basis is the eigenbasis of a $G$-valued operator $\hat U$ defined by $\hat U \left|U\right> = U \left|U\right>$. On a lattice with many links, the generalization of this basis is one state $\left|U(\ell)\right>$ for every function $U$ from the set of links to the gauge group $G$. This set is a basis for $\mathcal H$. To every link is associated an operator $\hat U_\ell$, defined as $\hat U$ above acting on the Hilbert space of that link.

The Hilbert space of the physical theory is in fact a subspace of $\mathcal H$. To construct the physical Hilbert space, we need to consider the group of gauge symmetries. For any site $x$ and $V \in G$, we define an operator $\phi_x(V)$ on $\mathcal H$ which performs a gauge transformation by $V$ at $x$. The operator acts on each link independently. Links going out of site $x$ transform in the right regular representation; links going into site $x$ transform in the adjoint of the left regular representation (other links do not transform, or rather, transform in the trivial representation):
\begin{equation}
\phi_x(V) \left|U_{xy}\right> = \left|U V\right>
\text{ and }
\phi_x(V) \left|U_{yx}\right> = \left|V^\dagger U\right>
\text.
\end{equation}
Here we have introduced the subscript notation $xy$ to denote a link from y to x.
A general gauge transformation is obtained by specifying an element $V$ at every site of the lattice\footnote{Performing a `constant' gauge transformation, where the group element $V$ is the same at every site, has no effect if $g$ is in the center $Z$ of $G$. Therefore, the full group of gauge symmetries is not $G^V$, but rather $G^V / Z$.}.
Under such a gauge transformation, the transformation law for the Hilbert space is
\begin{equation}
\phi(V) \left|\cdots U_{xy}\cdots\right>
=
\left| \cdots \left(V^\dagger_x U_{xy} V_y\right) \cdots\right>
\text.
\end{equation}

The physical Hilbert space $\mathcal H_P$ is the subspace of $\mathcal H$ consisting only of states invariant under gauge transformations. This can be defined constructively with the aid of the gauge projection operator
\begin{equation}
P = \int \d V \phi(V)
\end{equation}
which has the effect of integrating over all possible gauge transformations.

The lattice Hamiltonian for gauge group $G$ is
\begin{equation}\label{eq:gauge-lattice-hamiltonian}
H = \sum_\ell \nabla^2_{\ell} + \sum_P \Re \Tr U_P
\text.
\end{equation}
where, just as for the $\mathbb Z_2$ gauge theory, the first sum is over all links and the second sum goes over all plaquettes $P$. The operator $U_P$ is defined as the product of the $U_\ell$ operators for links going around the plaquette: for the plaquette shown in Figure~\ref{fig:gauge-lattice} we have $U_P$ = $U_0 U_1^\dagger U_2^\dagger U_3$. (Under the trace, the choice of starting link does not matter.)

For a continuous gauge group $G$, the operator $\nabla_\ell^2$ is the Laplace-Beltrami operator; i.e. the kinetic energy of the wavefunction on the curved surface $G$. This is the generalization of the Laplacian for a curved manifold. For discrete gauge groups, an appropriate generalization of this is obtained by noting that the Laplace-Beltrami operator is diagonal in Fourier space, and is proportional to the identity when restricted to any irreducible representation. It remains to pick one real number for each irreducible representation of $G$, and \emph{any} such choice will yield a gauge-invariant Hamiltonian.

Viewing the lattice theory as a regular quantum mechanical system, the first term is a kinetic term (momentum squared on a curved manifold), and the second is a potential term, defining how the degrees of freedom are coupled. In the language of field theory, and specifically by analogy with the $U(1)$ gauge theory of electromagnetism, the first term is the electric term and the second the magnetic term, so that the Hamiltonian can be re-written as $H = \frac 1 2 E^2 + \frac 1 2 B^2$.

Each term in the lattice Hamiltonian (\ref{eq:gauge-lattice-hamiltonian}) is individually gauge-invariant. The kinetic term at link $\ell$ is invariant under any rotation of the group $G$, and therefore under all gauge transformations. The plaquette term in the Hamiltonian is invariant under gauge transformation because each closed path (termed a Wilson loop) must go into and out of every link, and therefore the gauge transformation is immediately undone.

The gauge theories of greatest physical interest are those with continuous Lie groups. The gauge group of quantum chromodynamics is $SU(3)$; that of the standard model is $U(1)\times SU(2) \times SU(3)$. Much of the focus of numerical simulation work is on the calculation of quantities in the $SU(3)$ gauge theory.

\subsection{Path Integral}
We now derive lattice path integral for a gauge theory (following~\cite{creutz1977gauge,Lamm:2019bik}), and discuss the matter of gauge fixing.

The construction of the transfer matrix $T$ presented here, differs slightly from the usual one~\cite{Luscher:1976ms,Mitrjushkin:2002nk,Creutz:1999zy,Caracciolo:2012qw,Creutz:1976ch} in that $T$ is defined on the entire space $\mathcal H$. The usual transfer matrix is defined only on the physical Hilbert space, and may be obtained by projection with $P$.

Fixing a timestep $\Delta t$, the transfer matrix is an approximation to imaginary-time evolution: $T \approx e^{- \Delta t H}$.
We will work in the fiducial basis of $\mathcal H$ of eigenstates $\left|U\right>$ of the field operators. We construct the transfer matrix in this basis from separate kinetic $T_K$ and potential $T_V$ contributions via Suzuki-Trotter approximation~\cite{suzuki1976generalized,trotter1959product}
\begin{equation}\label{eq:transfer}
T = T_V^{1/2} T_K T_V^{1/2}
\end{equation}
where the potential evolution resembles the product of spatial plaquettes that appears in the Wilson action
\begin{equation}
\langle\tilde U_{12}\hdots|T_V|U_{12}\hdots\rangle =
\delta_{U_{12}\cdots}^{\tilde U_{12}\cdots}
\exp\left(\frac {\Delta t} {a g^2}\sum_x W_{\mu\nu}(x)\right)
\text.
\end{equation}
Here $a$ is the spatial lattice spacing, not necessarily equal to $\Delta t$. We have borrowed from lattice field theory the Wilson plaquette
\begin{equation}W_{\mu\nu}(x) = \Re \Tr [U^\dagger_{x,x+\hat\nu} U^\dagger_{x+\hat\nu,x+\hat\mu+\hat\nu} U_{x+\hat\mu,x+\hat\mu+\hat\nu} U_{x,x+\hat\mu}],
\end{equation}
and $\mu,\nu$ are restricted to space-like directions. The kinetic evolution acts on each link independently.
\begin{equation}
\langle\tilde U_{12}\hdots|T_K|U_{12}\hdots\rangle =
\prod_{\left<ij\right>} e^{\frac a {g^2 a_0}\Re \Tr \left[\tilde U_{ij}^\dagger U_{ij}\right]}
\end{equation}
Note that $[T,P]=0$ due to the fact that $T_V$ and $T_K$ individually commute with $P$.

At this point a path integral may be obtained from the approximate partition function $\Tr T^{-\beta/\Delta t}$; however, this partition function includes contributions from nonphysical gauge-variant states. The correct partition function is projected onto the ground state. There is some freedom in how we do this. The most straight-forward approach is to insert a single projection operator, writing $Z = \Tr e^{-\beta H} P$. This yields the correct physics, but the resulting lattice path integral is awkward to work with, because the coupling between the first and last time-slices will be different than the coupling between any other pair. (A calculation performed with this path integral is partly gauge-fixed to the $A_0=0$ gauge.) It is instead conventional to insert multiple projection operators, writing
\begin{equation}
Z = \Tr_P T^\beta = \Tr (P T P)^\beta
\end{equation}
This is possible because the projection operator commutes with the transfer matrix $T$, and in fact with $T_K$ and $T_V$ individually. Projecting the kinetic part of the transfer matrix yields
\begin{equation}
\langle\tilde U_{12}\hdots|P T_K P|U_{12}\hdots\rangle =
\int \d V\;
\prod_{\left<ij\right>} e^{\frac a {g^2 a_0}\Re \Tr \left[\tilde U_{ij}^\dagger V_i^\dagger U_{ij} V_j\right]}
\text.
\end{equation}
The $V_i$ are group elements that live at a single site in the Hamiltonian picture. Visualizing a Euclidean lattice (with a separate time-like dimension), the $V_i$ connect a lattice site on one time-slice to the same site on the next time-slice. They constitute a gauge-transformation performed in going from one time-slice to the next. This is usually visualized in the form of timelike links connecting one spatial slice to the next.

The resulting partition function is
\begin{align}\label{eq:transfer-partition}
Z &=
\left(\int_G \d U_{12}\cdots\right) \exp\left[\frac 1 {g^2} \sum \Re \Tr \prod U_{ij}\right]
\text.
\end{align}
Here $\Tr_P$ denotes the trace over only the physical subspace $\mathcal H_P$, and the sum is taken over both spatial and temporal plaquettes on a $d+1$ lattice.

It remains to show that, for vanishing temporal lattice spacing $\Delta t \rightarrow 0$, the transfer matrix corresponds exactly to imaginary time under the gauge Hamiltonian. This is done in detail by Creutz \cite{Creutz:1976ch} for the gauge-fixed transfer matrix. The result for the gauge-free transfer matrix used here is the same. The potential transfer matrix $T_V$ is exactly $e^{-\Delta t H_V}$. The correspondence between $H_K$ and $T_K$ is not exact, but indeed $\lim_{\Delta t \rightarrow 0} T_K = e^{-\Delta t H_K}$.

\chapter{Computational Difficulties}\label{ch:difficulties}

The most straightforward method for studying the physics of a lattice Hamiltonian is to construct the Hamiltonian as an explicit matrix, and perform computatinal linear algebra on that matrix. For instance, the matrix may be diagonalized to reveal the masses of particles and bound states. Such methods are in practice useless for three-dimensional field theories: linear algebra algorithms scale polynomially with the dimension of the vector space in question, and the Hilbert space of a lattice theory is exponential in the volume of the lattice\footnote{Or, for theories with continuous degrees of freedom, the Hilbert space is infinite.}.

The most widely-used nonperturbative tool for studying lattice QCD in practice is the Markov chain Monte Carlo (MCMC) method. Much of the time, this algorithm scales polynomially with the volume of the system being studied, rather than exponentially. At finite density of relativistic fermions, or when studying real-time evolution, the MCMC method reverts to exponential scaling. In these cases, we are left without a general nonperturbative method.

\section{Monte Carlo Methods}

Monte Carlo methods for lattice field theory are based on the path integral representation of the partition function. For concreteness, we consider here a scalar field theory, but the same ideas generalize to gauge theories and theories of interacting fermions (discussed in more detail in Section~\ref{sec:fermion-density} below).

The lattice partition function for a theory of one real scalar field, as described by the Hamiltonian (\ref{eq:h-scalar-interacting}), is
\begin{align}\label{eq:scalar-partition}
Z[J]
&=
\int \d^{\beta V}\phi\;
e^{-S[\phi]}
e^{\sum_x J_x\phi_x}\nonumber\\
\text{where }&
S[\phi] = 
\sum_{\left<xy\right>}
\frac{\left(\phi_x - \phi_y\right)^2}{2}
+
\sum_x
\frac {m^2}{2} \phi^2
+
\sum_x
\lambda \phi_x^4
\text.
\end{align}
The functional $S$ is referred to as the action.
This expression is only an approximation (\`a la Suzuki-Trotter) to the true partition function of the Hamiltonian, which is itself only an approximation to the continuum theory. After performing a calculation, one must extrapolate to the continuum and infinite volume limits to obtain physically meaningful results. The process of extrapolation is largely independent from the rest of the calculation, and we will ignore it for the remainder of this chapter.

The partition function (\ref{eq:scalar-partition}) couples the fields linearly to a spacetime-dependent source field $J$. Expectation values are obtained by differentiating $\log Z$ with respect to $J$. For an arbitrary observable $\mathcal O$ (typically some polynomial of the fields), the expectation value is given by
\begin{equation}
\left<\mathcal O\right>
=
\frac{\int \d^{\beta V}\phi\;e^{-S[\phi]} \mathcal O[\phi]}{\int \d^{\beta V}\phi\; e^{-S[\phi]}}\text.
\end{equation}

When the action $S[\phi]$ is guaranteed to be real, the ``Boltzmann factor" $e^{-S}$ is non-negative, and this expectation value may be viewed as an expectation value over the probability distribution proportional to $e^{-S}$. It follows that, to calculate arbitrary expectation values, one need only sample from the distribution given by the Boltzmann factor.

\subsection{Markov-Chain Monte Carlo Methods}
A Markov chain is a discrete-time stochastic process in which the state at time $t+1$ depends only on the state at time $t$. The chain is defined by a matrix $P_{ij}$ giving the probability of transitioning to state $i$ at step $t+1$, given that the state was $j$ at step $t$. The matrix $P$ should be thought of as a linear operation on the space of probability distributions.

The long-time behavior of a Markov chain is determined by the eigenvector of $P$ with largest eigenvalue. Markov-Chain Monte Carlo (MCMC) algorithms sample from a distribution by setting up a Markov chain with the target distribution as this eigenvalue. For many distributions encountered in practice, including many lattice field theories, the Markov chain mixes in time polynomial in the volume (and all other physical parameters), making these methods viable for simulations.

The particular method most widely used in lattice field theory is the Metropolis algorithm~\cite{Metropolis:1953am}.

\section{Finite Fermion Density}\label{sec:fermion-density}

When simulating a field theory with fermions, the fermions are typically integrated out analytically before performing the numerical integral~\cite{gattringer2009quantum,montvay1997quantum}. An example of this is a simulation of $SU(3)$ gauge theory with fermions, i.e.\ QCD. The lattice action is
\begin{equation}
S = \sum_P \Re \Tr P
+ \sum_i m \bar\psi_i \psi_i
+ \frac 1 2 \sum_{i,\mu}
\left[
\bar\psi_i U_{i,i+\hat\mu} \gamma^\mu \psi_{i+\hat\mu}
+ \mathrm{h.c.}
\right]
\text.
\end{equation}
Here the first sum is over all plaquettes, the second over all sites, and the third over all sites and spacetime directions $\mu$.
Because this action is only quadratic in the fermion fields, they can be integrated out analytically for any fixed configuration of the gauge fields. This yields the lattice partition function
\begin{equation}
Z = \int \d U e^{-S[U]} \det D[U]
\equiv
\int \d U e^{-S_{\mathrm{eff}}[U]}
\end{equation}
where $S_G$ is the pure-gauge piece of the original action, and $D[U]$ is the gauge-field-dependent matrix that gave the quadratic part of the action.

Although this saves us from having to work directly with anticommuting numbers, it introduces a new complication: the determinant of $D$ may not be a non-negative real number, in which case $e^{-S} \det D$ cannot be interpreted as a probability distribution. This situation is termed a \emph{fermion sign problem}. In the context of relativistic field theories, $\det D$ is usually guaranteed to be positive and real at zero density, but at finite density picks up a complex phase.

A standard approach at this point is to define the \emph{quenched} Boltzmann factor as the absolute value $e^{-S_G} |\det D|$ of the original Boltzmann factor. Expectation values of the physical system can be rewritten as ratios of expectation values of the quenched system:
\begin{equation}\label{eq:reweighting}
\left<\mathcal O\right>
=
\frac{\left<\mathcal O e^{-i \Im S_{\mathrm{eff}}}\right>_Q}
{\left<e^{-i \Im S_{\mathrm{eff}}}\right>_Q}
\end{equation}
where $\left<\cdot\right>_Q$ denotes a quenched expectation value.

The denominator of (\ref{eq:reweighting}) is difficult to evaluate. Each sample will be a number with magnitude $1$; once averaged, these number must cancel out to yield a value of considerably smaller magnitude. In fact, a simple argument shows that the denominator $\left<e^{-i\Im S_{\mathrm{eff}}}\right>_Q$, termed the ``average sign'' $\left<\sigma\right>$, is characteristically exponentially small in the volume. The partition function of a field theory in a large volume should be approximately equal to a product of two partition functions with half the volume: the contribution of the boundary to the free energy is negligible in this limit. The same statement is true of the quenched partition function $Z_Q = \int |e^{-S_{\mathrm{eff}}}|$. The average sign is just the ratio of the physical partition function to the quenched partition function, $\left<\sigma\right> = Z / Z_Q$. It therefore follows that
\begin{equation}
\left<\sigma\right>(V)
=
\frac{Z(V)}{Z_Q(V)}
\approx
\frac{Z(V/2)^2}{Z_Q(V/2)^2}
=
\left[\left<\sigma\right>(V/2)\right]^2
\end{equation}
and therefore the average sign shrinks exponentially with the volume. It follows that an exponentially large number of samples are needed to even resolve the denominator of (\ref{eq:reweighting}) from $0$. This exponential scaling is characteristic of methods that encounter a fermion sign problem.

The fermion sign problem is a major obstacle to nonperturbative calculations in several physical regimes. Prominent in nuclear physics is the problem of determining the low-temperature limit nuclear equation of state~\cite{lattimer2012nuclear}. To a reasonable approximation, the interior of a neutron star is at zero temperature, and so the energy density as a function of number density (or equivalently, pressure as a function of energy density) of zero-temperature nuclear matter determines the mass-radius curve of neutron stars. This function is beginning to be constrained through astronomical observations~\cite{miller2019psr}, but remains largely out of the realm of first-principles calculations.

It has been shown that the most general case of a fermion sign problem is NP-hard~\cite{Troyer:2004ge}. Under standard assumptions of computational complexity, this implies that classical (or even quantum) simulations of such systems cannot be achieved in polynomial time~\cite{arora2009computational}. It is important to bear in mind, however, that this result does not exclude (even heuristically) the polynomial-time simulation of specific instances of systems that suffer from a fermion sign problem. In particular, the system used to prove the NP-hardness of the general case has an inhomogeneous Hamiltonian, and indeed the proof relies heavily on that fact by encoding particular combinatorial problems into the Hamiltonian. The Hamiltonians associated to field theory are homogenous and have little information content.

\section{Real-Time Linear Response}

So far we have discussed the difficulties encountered when trying to determine the equilibrium properties of quantum matter. Out-of-equilibrium physics is also difficult to access with nonperturbative techniques.

A general class of experiments we might perform involve preparing a thermal state of some Hamiltonian $H_0$, and then changing the Hamiltonian to some (possibly time-dependent) $H(t)$, and measuring an expectation value $\left<\mathcal O(T)\right>$ at some later time. The Schwinger-Keldysh formalism~\cite{keldysh1965diagram} presents us with the possibility of performing such calculations with lattice methods~\cite{Alexandru:2016gsd,Alexandru:2017lqr}. The time-dependent expectation value, at inverse temperature $\beta$, is given by
\begin{equation}
\left<\mathcal O(T)\right>
=
\frac{\Tr e^{-\beta H_0} e^{i H T} \mathcal O e^{-i H T}}
{\Tr e^{-\beta H_0} e^{i H T} e^{-i H T}}
\text.
\end{equation}
Note that the denominator is just the partition function. This expression can be transformed into a path integral in the same manner as the pure-imaginary time partition function. The result, for scalar field theory, is a path integral in both imaginary and real time, with lattice action
\begin{equation}
S
=
\sum_{x,\left<tt'\right>}
\frac{\left(\phi_{x,t} - \phi_{x,t'}\right)^2}{2 a_0}
+
\sum_{\left<xy\right>,t}
a_0 \frac {\left(\phi_{x,t} - \phi_{y,t}\right)^2}{2}
+
\sum_{x,t}
a_0 \left(\frac {m^2}{2} \phi_{x,t}^2
+
\lambda \phi_{x,t}^4\right)
\text,
\end{equation}
where $a_0(t)$ depends on the time-slice being considered, being $1$ for the first $\beta$ slices (corresonding to the thermodynamic part of the lattice), $i$ for the next $T$ slices (yielding forward time evolution), and $-i$ for the rest (backward time evolution). Here we have assumed that the lattice spacing is $1$.

A less ambitious version of this task comes from considering the case where $H(t)$ is equal to $H_0$, except for a small, delta-like term added at $t=0$:
\begin{equation}
H_\epsilon(t) = H_0 + \epsilon \delta(t) H'
\text.
\end{equation}
The response of $\mathcal O$ at some later time, to leading order in $\epsilon$, is termed linear response, and is given by a time-separated correlator evaluated in equilibrium.
\begin{equation}
\left<\mathcal O(T)\right>
=
\left<\mathcal O(0)\right>
+
\epsilon
\left<\left[H', \mathcal O(t)\right]\right>
\end{equation}
Transport coefficients, such as bulk and shear viscosity, fall into the category of linear response. The calculation of these time-separated correlators still suffers from a sign problem, and will be our main focus.

An efficient algorithm for classically computing the real-time \emph{non}-linear response of a quantum system, with an arbitrary time-varying Hamiltonian, would imply the ability to efficiently simulate a quantum computer with a classical computer~\cite{jordan2018bqp}. Thus, under common computational complexity assumptions, it is expected that no such algorithm exists. However, just as Troyer and Wiese's result~\cite{Troyer:2004ge} on the hardness of the inhomogeneous fermion sign problem does not forbid a solution to the homogeneous problem, the hardness of nonlinear simulation with a time-varying Hamiltonian does not seem to forbid the efficient computation of two-point correlators.

Unlike in the fermion case, however, more directly relevant results have been recently developed. Two developments are worth highlighting here. First, under stronger (but still widely believed) assumptions about computational complexity\footnote{Specifically, the fact that the polynomial hierarchy does not collapse.}, the simulation of a sequence of \emph{commuting} quantum gates is inaccessible by any polynomial-time classical algorithm~\cite{bremner2011classical}. This problem corresponds to the physical task of computing the nonlinear response of an arbitrary homogeneous state under a time-constant, but spatially inhomogeneous, Hamiltonian. Separately, again under standard assumptions, it was shown in~\cite{aaronson2017computational} that the task of simulating quantum scattering, beginning from an arbitrary initial state, is also inacessible by efficient classical algorithms.

These barriers do not provide evidence that time-separated two-point functions are inaccessible classically. The lattice Schwinger-Keldysh method discussed above is not the only approach to computing these functions on the lattice. A common approach, applied for example to the shear viscosity of lattice Yang-Mills~\cite{astrakhantsev2017temperature} is to compute the two-point function at Euclidean separation, and attempt to analytically continue to timelike Minkowski separation. This approach ultimately suffers from the fact that the analytic continuation is ill-posed, and some modeling assumptions are needed.

\section{Noisy Correlators}

Particle masses may be measured in a lattice calculation by considering the long-time behavior of a correlator separated in imaginary time. Let $\mathcal O$ be an operator that, when applied to the vacuum state $|\Omega\rangle$, has some overlap with the ground state $|P\rangle$ of a single particle whose mass we would like to know: $\langle P|\mathcal O|\Omega\rangle \ne 0$. The Euclidean time-separated correlator has an exponential decay characterized by the mass of the particle:
\begin{equation}
C(\tau)
=
\langle \Omega|e^{\tau H} \mathcal O e^{-\tau H} \mathcal O|\Omega\rangle
=
\sum_i e^{\tau(E_\Omega - E_i)}|\langle i|\mathcal O|\Omega\rangle|^2
\end{equation}
where the sum is taken over all eigenstates of the Hamiltonian, of which $|P\rangle$ is one. When $|P\rangle$ is the lowest-lying eigenstate with nonvanishing overlap, the asymptotic behavior of $C(\tau)$ reveals the mass.

The measurement of $C(\tau)$ on the lattice has some noise, characterized by the variance $\langle(\mathcal O(\tau) \mathcal O(0))^2\rangle - \langle\mathcal O(\tau) \mathcal O(0)\rangle^2$. The difficulty of obtaining an accurate measurement of $C(\tau)$ is measured by the signal-to-noise ratio; that is, the ratio the expectation value to the standard deviation of the estimator. An argument due to Parisi and Lepage~\cite{Parisi:1983ae,Lepage:1989hd} shows that, for the proton, the signal-to-noise ratio falls off exponentially with $\tau$. The correlator that yields the proton mass is $\langle\bar q\bar q \bar q(\tau) q q q (0)\rangle$, and asymptotically decays with $e^{-\tau m_p}$, where $m_p$ is the proton mass. The varianceis given by $\langle\bar q \bar q \bar q(\tau) q q q (\tau) \bar q \bar q \bar q(0) q q q (0)\rangle$. The operator in that correlator at $\tau=0$ has overlap with a three-pion state, and so the asymptotic behavior of the variance is $e^{-3 \tau m_\pi}$. The noise thus decays less quickly than the signal. As a result, the signal $C(\tau)$ is exponentially difficult to measure at large separations $\tau$.

This signal-to-noise problem is not as severe in practice as the sign problems associated to finite fermion density and real-time correlators. In particular, it has not prevented the accurate measurements of hadronic masses on the lattice~\cite{majumder2010hadron}. Although it does not outright prevent these calculations, it does make them more expensive. The Parisi-Lepage signal-to-noise problem can be reduced by complexification~\cite{Detmold:2020ncp} and, as we will see in Section~\ref{sec:measuring-mass}, can be evaded entirely on a quantum computer.

\chapter{Complexification}\label{ch:complexification}

Motivated by the previous chapters, we would like to compute via Monte Carlo sampling with reweighting, the expectation value of a function $\mathcal O(A)$, defined as
\begin{equation}
\left<\mathcal O\right>
=
\frac
{\int \mathcal D A\; e^{-S(A)}\; \mathcal O(A)}
{\int \mathcal D A\; e^{-S(A)}}
\text,
\end{equation}
where $S(A)$ denotes the action, $\mathcal O(A)$ comes from some Hermitian observable, and the integral is taken over all Euclidean lattice field configurations. For the moment, we will abstract the problem somewhat, allowing $S$ and $\mathcal O$ to be arbitrary functions, with the sole restriction that they be holomorphic\footnote{We will see in Section \ref{sec:cryptoholomorphicity} that even for fermionic theories, functions $\mathcal O$ coming from arbitrary correlators are in fact holomorphic.}.

To evaluate the expectation value via reweighting more efficiently, we will explore a method for alleviating the sign problem present in the denominator. A sign problem is present whenever $S$ fails to be real, and is characterized by the average sign
\begin{equation}
\left<\sigma\right> = \frac
{\int \mathcal D A\; e^{-S(A)}}
{\int \mathcal D A\; e^{-\Re S(A)}}
\equiv
\frac Z {Z_Q}
\text;
\end{equation}
smaller $\left<\sigma\right>$ correspond to worse sign problems. Note that despite the notation $\left<\sigma\right>$ is not an expectaton value with respect to $e^{-S}$. In fact, it is an expectation value with respect to the \emph{phase-quenched} action $\Re S$, and we have introduced the quenched partition function $Z_Q$, defined as the integral of the quenched Boltzmann factor $|e^{-S}|$.

In this chapter we will consider the integral over fields $A$ as a contour integral in the sense of complex analysis. As written, the integral is taken over $\mathbb R^N \subset \mathbb C^N$, but we will deform this contour to a different $N$-manifold $\mathcal M \subset \mathbb C^N$, and integrate over $\mathcal M$ instead. This procedure is motivated by two observations: first, that the action and observables are holomorphic functions of $A$, and therefore the expectation value $\left<\mathcal O\right>$ will have no dependence on $\mathcal M$; second, that the quenched partition function is the integral of a non-holomorphic function, and therefore the average sign will generically depend upon the choice of $\mathcal M$.

This chapter proceeds as follows. After a one-dimensional motivating example, the general procedure is rigorous described in the $N$-dimensional case, with a proof of the theorem that prevents expectation values from depending on the choice of manifold. Next we discuss two methods for selecting a manifold of integration, and apply each to the previously-discussed Thirring model. Finally, we discuss one case in which these methods completely remove the sign problem, and one case in which these methods provably have no effect.

\section{A One-Dimensional Example}

Our motivating example is the sign problem that comes from considering an action of one variable, $S(x) = x^2 + 2 i \alpha x$. In this case, the partition function and quenched partition function can both be evaluated exactly, and the average sign is
\begin{equation}
\left<\sigma\right>(\alpha)
=
\frac
{\int e^{-x^2 - 2 i \alpha x}}
{\int e^{-x^2}}
= e^{-\alpha^2}
\text.
\end{equation}

As an aside, note that one can increase the `volume' of the system by adding more dimensions to the integral. The partition function of the volume-$V$ system is then $Z_{V} = Z_1^V$, where $Z_1$ is the partition function given by the single-dimensional integral. The resulting sign problem is
\begin{equation}
\left<\sigma\right>(\alpha, V) =
\left<\sigma\right>(\alpha, 1)^V
= 
e^{-\alpha^2 V}
\text,
\end{equation}
a simple demonstration of the general fact that sign problems scale exponentially in the volume. (By contrast, the parameter $\alpha$ doesn't correspond to any parameter in a physical system, and so the scaling with $\alpha$ shouldn't be taken seriously.)

\begin{figure}
\centering
\includegraphics[width=2.8in]{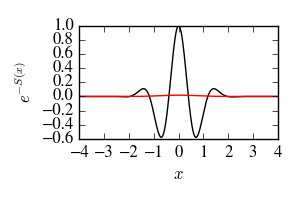}
\hspace{0.2in}
\includegraphics[width=2.8in]{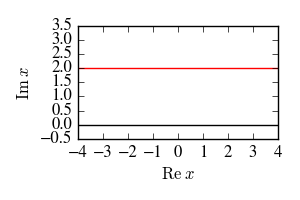}
\caption{\label{fig:complexify-gaussian}The sign problem, and its removal, of a gaussian integral. On the left is the Boltzmann factor $e^{-S}$ without any contour deformation performed (black) and with the optimal contour (red). The original manifold of integration (black) and the deformed contour (red) are on the right.}
\end{figure}

This sign problem can be removed entirely by deforming the contour of integration from the real plane to the line defined by $\Im x = - \alpha$. The situation is shown in Figure~\ref{fig:complexify-gaussian}: after the contour deformation, the partition function is written
\begin{equation}
Z
=
\int_{\infty}^\infty  \d x\; e^{-(x+i\alpha)^2}
=
\int_{\infty-i\alpha}^{\infty-i\alpha}  \d x\; e^{-x^2}
=
\int_{\infty}^{\infty}  \d x\; e^{-x^2}
\end{equation}
which is of course a sign-problem free integral. The first step is merely a change-of-variable; in the second step, Cauchy's theorem must be invoked to show that the two contour integrals are the same.

This motivates the broad strategy of `complexification' for attacking sign problems. In general, the partition function $Z$ (from which physical quantities are obtained) is invariant under contour deformations, as it is an integral of the holomorphic function $e^{-S(A)}$. The quenched partition function $Z_Q$, by contrast, is the integral of the non-holomorphic $e^{-\Re S(A)}$, so the value of $Z_Q$, and therefore the average sign, are not invariant under contour deformations.  Thus one may improve the sign problem by deforming the contour of integration, without changing the observables measured. The precise requirements on the contour deformation are imposed by Cauchy's theorem, which we discuss next.

\section{Cauchy's Integral Theorem}

Cauchy's integral theorem is the tool that allows us to deform a contour integral without changing the value of the integral. A function $f : \mathbb C \rightarrow \mathbb C$ is termed \emph{holomorphic} where it obeys the Cauchy-Riemann equations:
\begin{equation}
\frac{\partial u}{\partial x}
=
\frac{\partial v}{\partial y}
\hspace{1em}\text{ and }\hspace{1em}
\frac{\partial v}{\partial x}
=
-\frac{\partial u}{\partial y}
\text,
\end{equation}
where we have decomposed $f(x+iy) = u(x,y) + i v(x,y)$ into its real and imaginary parts. By introducing the holomorphic (Wirtinger~\cite{wirtinger1927formalen}) and antiholomorphic derivatives
\begin{equation}
\partial = \frac{\partial}{\partial z} \equiv
\frac 1 2 \left(\frac\partial{\partial x} - i \frac{\partial}{\partial y}\right)
\;\text{ and }\;
\bar\partial =
\frac{\partial}{\partial \bar z} \equiv
\frac 1 2 \left(\frac\partial{\partial x} + i \frac{\partial}{\partial y}\right)
\text,
\end{equation}
the Cauchy-Riemann equations can be rewritten as $\bar\partial f = 0$. Note that the holomorphic and antiholomorphic derivatives, just like the ordinary derivatives $\partial_x$ and $\partial_y$, are derivatives taken along orthogonal vectors.

Cauchy's integral theorem states that for holomorphic functions, integrals around closed contours vanish.

\begin{theorem}[Cauchy's Integral Theorem]
For a closed region $\Omega \subset \mathbb C$ and a function $f$ holomorphic on $\Omega$, the integral of $f$ around the boundary of $\Omega$ must vanish:
\begin{equation}
\oint_{\partial\Omega} f(z)\;\d z = 0
\end{equation}
\end{theorem}
\begin{proof}
By Stokes' theorem,
\begin{equation}
\oint_{\partial\Omega} f(z)\;\d z
=
\int_{\Omega} \d (f(z) \d z)
=
\int_{\Omega} \d f \wedge \d z
\text.
\end{equation}
The differential of $f$ is given by $\d f = \partial f \d z + \bar\partial f \d \bar z$. 
As the anti-holomorphic derivative of $f$ vanishes, the differential of $f$ is simply $\d f = \partial f \d z$. However, the wedge product $\d z \wedge \d z$ vanishes, and so must the integral.
\end{proof}

The usual form of Cauchy's theorem --- the one just given --- applies to functions of one complex variable. The theorem has a natural generalization to functions of many complex variables. A $f : \mathbb C^N \rightarrow \mathbb C$ of $N$ complex variables is termed holomorphic where it obeys the Cauchy-Riemann equations in each complex dimension independently; that is, where
\begin{equation}
\frac{\partial}{\partial \bar z_i} f(z_1,\ldots,z_i,\ldots,z_N) = 0
\end{equation}
for all $i$. A multidimensional generalization of Cauchy's theorem states that the integral around the boundary of any $(N+1)$-real-dimensional region, in which $f$ is holomorphic, must vanish.
\begin{theorem}[Multidimensional Cauchy's Theorem]
For a closed region $\Omega \subset \mathbb C^N$ of real dimension $N+1$, and a function $f$ holomorphic on $\Omega$, the integral of $f$ around the boundary of $\Omega$ vanishes:
\begin{equation}
\oint_{\partial\Omega} f(z)\;\bigwedge_{i=1}^N \d z_i = 0
\end{equation}
\end{theorem}
\begin{proof}
As before, Stokes' theorem yields
\begin{equation}
\oint_{\partial\Omega} f(z)\wedge \bigwedge_{i=1}^N \d z_i
=
\int_{\Omega} \d f \wedge \bigwedge_{i=1}^N \d z
\text.
\end{equation}
The differential $\d f$ now has $2N$ terms, of which $N$ vanish by the Cauchy-Riemann equations $\bar\partial_i f = 0$. Each of the remaining terms has the form $\partial_j f \d z_j$ for some $j$, and therefore is annihilated when the wedge product is taken with $\bigwedge_i \d z_i$.
\end{proof}

\begin{figure}
\centering
\includegraphics[width=3in]{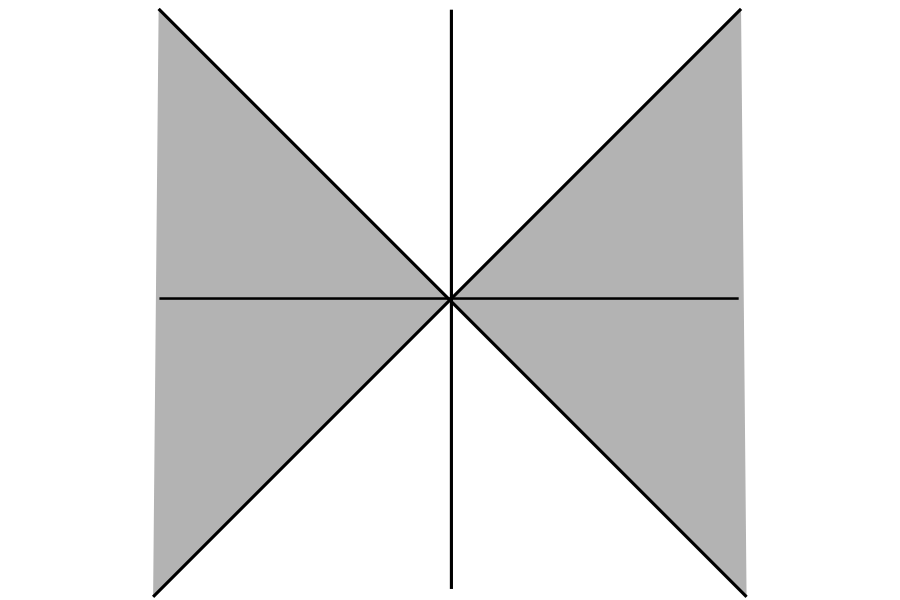}
\caption{The asymptotically safe regions for the Gaussian integral.\label{fig:asymptotically-safe}}
\end{figure}

A compication arises when deforming integration contours which extend to infinity. Application of Cauchy's integral theorem alone does not allow the asymptotic behavior of such a contour to be changed. However, in cases where the function being integrated decays rapidly\footnote{An exponential decay is always sufficient.} at infinity, the contribution of the part of the contour at infinity vanishes, and the asymptotic behavior can be changed without affecting the result. Figure~\ref{fig:asymptotically-safe} shows the asymptotically safe regions for the Gaussian model. The shaded regions mark the regions that decay exponentially at infinity, and as long as a contour's asymptotics remain in a shaded region, the integral will be unchanged by any deformation.

Tracking the asymptotically safe regions and ensuring the manifold deformation never leaves them can be technically challenging. It is often possible to arrange the physical model so that the domain of integration is compact\footnote{The standard formulation of lattice gauge theories accomplishes this.}. This ensures that the original integration manifold does not go near any asymptotic region, and any finite deformation will be permissible. This is the approach we will take when applying the method to fermionic models.

\subsection{Holomorphic Boltzmann Factors and Observables}\label{sec:cryptoholomorphicity}

The utility of a Cauchy's-theorem-based procedure comes from the observation that the integrands of interest are holomorphic functions of the integration variables. Here we discuss two major cases in which the holomorphicity of the integrands is not obvious: a theory of complex scalar fields, and a theory of fermions in which the fermions have been integrated out.

\subsubsection{Complex Scalar Fields}

The (Euclidean) lagrangian of a complex scalar field is given by
\begin{equation}
L = \left(\partial_\mu \phi^\dagger\right) \left(\partial_\mu \phi\right) + V(|\phi|)\text,
\end{equation}
and it is immediately obvious that the lagrangian is not, and therefore the lattice action is not, a holomorphic function of the field variable $\phi$. The resolution to this dilemma comes upon writing out the full lattice path integral:
\begin{equation}
Z =
\int_{\mathbb C^V} \d \phi\;
\exp\left\{
- \sum_{\left<ij\right>} \frac {\left|\phi_i - \phi_j\right|^2}{2}
- \sum_i V(\phi_i)
\right\}
\end{equation}
Here we see that the domain of integration (for a theory with $V$ sites) has $V$ \emph{complex} dimensions, rather than $V$ real dimensions. Expressing the complex scalar field at site $i$ as a sum $\phi_i = u_i + i v_i$ of two real scalar fields at the same site, we note that both $\phi$ and $\phi^\dagger$ are in fact analytic functions of the new field variables $u$ and $v$. The partition function may now be written
\begin{equation}
Z =
\int \d^V u \;  \d^V v \;
\exp\left\{
- \sum_{\left<ij\right>}
\frac {\left(u_i - u_j\right)^2}{2}
- \sum_{\left<ij\right>}
\frac {\left(v_i - v_j\right)^2}{2}
- \sum_i V(u, v)
\right\}
\text,
\end{equation}
in which the integrand is manifestly a holomorphic function of $u$ and $v$. The methods of the previous sections may now be applied, and the contour integral deformed into the space of imaginary $u$ and $v$.

\subsubsection{Fermionic Determinant and Correlators}

Thus far we have been concerned with ensuring the lattice action is a holomorphic function. This guarantees that the integrand of the partition function is holomorphic, and that the partition function is unchanged by the deformation. Holomorphicity of the action, however, is in general neither necessary nor sufficient. In fact it is required is that $e^{-S[\phi]}$ and $\mathcal O[\phi] e^{-S[\phi]}$ both be holomorphic. We will now see that this does not imply that either $S$ or $\mathcal O$ must themselves be holomorphic~\cite{Alexandru:2018ngw}.

Consider a theory of fermions $\psi$, with interactions mediated by a gauge field $A_\mu$. After the fermions have been integrated out, the lattice partition function is
\begin{equation}\label{eq:fermion-partition}
Z = \int \d^{dV} A \; e^{-S_{\mathrm{gauge}}[A]} \det D[A]
\text,
\end{equation}
where the lattice has $d$ dimensions (making $A_\mu$ a $d$-vector), $S_{\mathrm{gauge}}$ gives the terms in the lattice action involving only the gauge field, and $D[A]$ is the fermion inverse propagator in the presence of a fixed background field $A$.

A typical observable of interest is a meson propagator
\begin{equation}
\left<\bar\psi_i \bar\psi_j \psi_j \psi_i\right> = \frac 1 Z \int \mathcal D\phi\; e^{-S[\phi]} \left[D^{-1}_{ij} D^{-1}_{ji} - D^{-1}_{ii} D^{-1}_{jj}\right]
\text,
\end{equation}
so we require this integrand to be holomorphic as well.

Note that the effective action on the gauge fields --- i.e. the logarithm of the integrand in the partition function --- contains a term $\log \det D$. This term has logarithmic singularities where $\det D = 0$. Also at these points, $D^{-1}$ is not well-defined, so the meson propagator given above (along with many others) involves a singular $\mathcal O[A]$.

Despite this, the integrands
$e^{-S[\phi]}$ and $\mathcal O[\phi] e^{-S[\phi]}$ are always holomorphic (with lattice regularized actions). The holomorphicity of $e^{-S}$ is easiest to understand. The gauge action is of course holomorphic in the fields $A$, and for typical gauge-fermion interactions each element $D[A]_{ij}$ of the fermion matrix is also a holomorphic function of $A$. As the determinant is merely a polynomial of the elements of the matrix, it follows that $\det D$ itself is holomorphic. This establishes that the partition function (\ref{eq:fermion-partition}) will remain unchanged under appropriate deformations.

We now discuss fermionic observables. In the case of a fermion propagator $\langle \bar\psi_i \psi_j\rangle$, there is only a single $D^{-1}$ in the integrand, and it is easy to see that the singularity of this factor is cancelled by the zero of the fermion determinant coming from the action.
To see that integrands involving fermionic observables are holomorphic in general, we write an expectation value in terms of the original, fermionic path integral.
\[
\left<\bar\psi_a\psi_b\right> = \frac 1 Z \int \mathcal D A \;e^{-S_\mathrm{gauge}[A]} \int \mathcal D\bar\psi \;\mathcal D \psi\; \bar\psi_a \psi_b e^{-\bar\psi_i D_{ij}[A] \psi_j}
\]
With $V$ sites, the fermionic exponential $e^{-\bar\psi D\psi}$ may be expanded
in $2^{2V}$ terms, identified by what subset of the $2V$ Grassmann variables is included in each term. The $\mathbb C$-number part of each term is a product of finitely many
components of $D[A]$, and therefore is a holomorphic function of the gauge field 
$A$. Multiplying by any combination of $\bar\psi\psi$ and integrating over
$\mathcal D\bar\psi\mathcal D\psi$ has the effect of selecting one of these
coefficients. Therefore, the integral over fermionic fields yields a
holomorphic function of $A$.

The story remains the same no matter how many fermionic fields are inserted in the expectation value, as long as no Grassmanns are repeated. If a Grassmann is repeated (if the expectation value $\langle \bar\psi_i \psi_i\rangle$ is requested), then the expectation value is simply zero.

\section{Lefschetz Thimbles}

At this point we have motivated the use of deformed contour integrals, and shown that physical observables will remain unchanged, but we have no general principles for selecting integration manifolds on which the average sign is likely to be improved. The Lefschtz thimbles~\cite{witten2011analytic} provide an attractive choice of manifold for improving the sign problem~\cite{Cristoforetti:2012su,cristoforetti2013monte}. Each thimble is an $N$-dimensional manifold extending from a critical point $z_c$ of the action obeying $\left.\partial S\right|_{z=z_c} = 0$. The thimble extends from the critical point along the paths of steepest descent of the real part of the action $\Re S$. The thimble terminates either at infinite, or at a point where the real part of the action diverges (zeros of the fermion determinant have this effect).

The union of all thimbles is not necessarily obtainable as a smooth deformation of the real plane; however, some linear combination of the thimbles always is~\cite{witten2011analytic}. In other words, there is some linear combination of the thimbles that, when integrated, gives a result equal to the integral along the real plane (for any holomorphic integrand). Determining what linear combination is needed may be computationally difficult, as indeed may be the task of enumerating all critical points of the action. Section~\ref{sec:flow} below provides a closely related algorithm which circumvents this problem.

The usefulness of the thimbles is related to the fact that $\Im S$ is constant on each thimble. Note first that the path of steepest descent, defined by
\begin{equation}
\frac{\d \Re A}{\d t} = \frac{\partial S}{\partial \Re A}
\;\text{ and }\;
\frac{\d \Im A}{\d t} = \frac{\partial S}{\partial \Im A}
\text,
\end{equation}
can also be written in the form $\dot A = \bar\partial S$. It follows that the change of $S$ with flow time is given by
\begin{equation}
\frac{\d S}{\d t} = \frac{\d A}{\d t} \frac{\partial S}{\partial A}
=
\left|\frac{\partial S}{\partial A}\right|^2
\text.
\end{equation}
As the change in the action is real, the imaginary part is constant along each path of steepest descent, and therefore all over the thimble. The fact that $\Im S$ is constant along a thimble means that, within one thimble and neglecting the Jacobian, there can be no phase cancellations in the integral.

It is not the case, however, that Lefschetz thimbles completely remove the sign problem. Lattice theories often have multiple thimbles contributing to the integral~\cite{tanizaki2016lefschetz}, with different phases, creating a sign problem. Additionally, although $\Im S$ is constant, the Jacobian (i.e.\ the curvature of the thimble) introduces its own contribution to the phase~\cite{Lawrence:2018mve}.

\section{Holomorphic Gradient Flow}\label{sec:flow}

A thimble may be defined as the union of all solutions $A(t)$ to the holomorphic gradient flow equation
\begin{equation}\label{eq:flow}
\frac{\d A}{\d t} = \frac{\partial S}{\partial \bar A}
\end{equation}
that approach a certain critical point $A_c$ in the early-time limit: $\lim_{t\rightarrow -\infty} A(t) = A_c$. This definition makes apparently the fact that any thimble is left invariant under the action of the gradient flow. In fact, thimbles are not only fixed points of the flow, but attractive fixed points: any nearby manifold will evolve to become closer to the thimble, under the flow.

It follows that the holomorphic gradient flow can be used to construct an algorithm for integrating along the thimbles~\cite{Alexandru:2015xva}. Define a function $\tilde A_T(A)$ to be the result of evolving the point $A$ under (\ref{eq:flow}) for time $T$. Under mild conditions on the action (holomorphicity of $e^{-S}$ is sufficient), this is a continuous function of $A$ and $T$, and therefore defines a contour of integration homotopic to the real plane. Moreover, in the limit of large $T$, this integration contour approaches some linear combination of the Lefschetz thimbles, and is therefore expected to have an improved sign problem. The Monte Carlo integration is performed by sampling $A$ according to a modified action
\begin{equation}
S_{\mathrm{eff}}(A)
=
S[\tilde A_T(A)]
-
\log \det J
\text,
\end{equation}
where $J$ is the Jacobian of $\tilde A$, i.e.\ the matrix of complex first derivatives.

Although the thimbles are only obtained in the long-time limit of the gradient flow, any manifold created by flowing for a finite amount of time can be used as an integration contour. In practice, it is found that flowing for a short amount of time can dramatically improve the average sign at a relatively small computational cost~\cite{Alexandru:2015sua}. This choice of integration manifold, defined by flowing the real plane for some fixed time $T$, is both useful in practice and provides a guide for the search for other manifolds.

\subsection{Algorithmic Costs}
Methods based directly on the holomorphic flow have a substantial drawback: the computation of $\det J$ is expensive. The evolution of the Jacobian induced by (\ref{eq:flow}) is given by
\begin{equation}
\frac{\d J}{\d t} = \bar H \bar J
\text,
\end{equation}
where $H$, the Hessian, is the matrix of holomorphic second derivatives of the action. The matrix multiplication is unavoidable when computing $\det J$, and requires about $O(n^3)$ steps in practice for an $n\times n$ matrix. The asymptotic time complexity of one step of the flow, then, is approximately cubic in the volume of the lattice. Much of the technical effort around flow is motivated by the desire to avoid this cost, including by computing an approximation to the determinant and reweighting~\cite{Alexandru:2016lsn}, or modifying the Monte Carlo sampling to automatically include the Jacobian~\cite{Alexandru:2017lqr}.

Although flow-based methods can improve a wide variety of sign problems without much need for model-specific tweaks, the expense of the procedure restricts the method to small lattices in practice. The parameterization of the integration manifold by the real plane is also not particularly convenient: in the limit of long flow times, an entire thimble is mapped to by a single point, creating large potential barriers (in parameterization space) between different thimbles. Finally, we will see in Section~\ref{sec:suboptimal} that the Lefschetz thimbles, the ``ultimate goal'' of flow-based methods, are not the best possible manifold, and that for large lattices it may be possible to dramatically improve the sign problem with a different integration contour. This motivates us to continue the search for other manifolds.

\section{Machine Learning Manifolds}

Instead of directly integrating on the flowed manifold $\tilde A_T(A)$, we can use machine learning methods to create a computationally efficient approximation, and perform the integration on the approximated manifold instead~\cite{Alexandru:2017czx}.

\subsection{Feed-Forward Networks}

\begin{figure}
\centering
\includegraphics[width=3.5in]{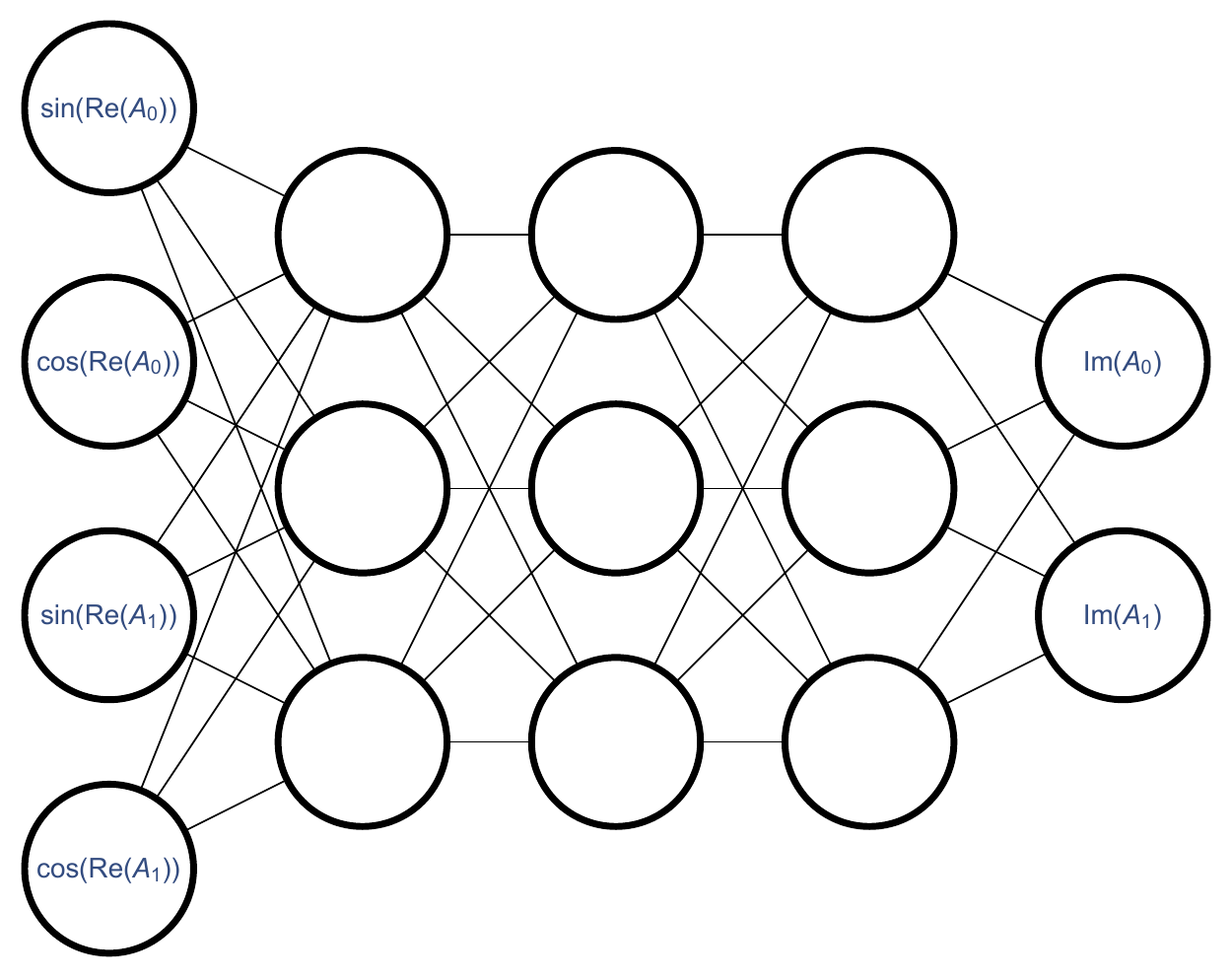}
\caption{Figure from~\cite{Alexandru:2017czx} depicting a feed-forward network, used to define a manifold in the complexified field space.\label{fig:feed-forward}}
\end{figure}

Feed-forward networks are a particular class of nonlinear functions $f(\vec x)$ of many variables which are particularly fast to compute. A feed forward network, depicted in Figure~\ref{fig:feed-forward}, consists of several layers, each containing a set of nodes. The initial layer (shown on the left) is the input layer, and to each node is associated one input variable $x_i$. Values are propagated through the network from left to right. The values $y_j$ in the second layer are determined from the values of the first layer by first performing some linear transformation on the $x_i$ (usually thought of as a set of weights associated to the edges between the first and second layers), and then performing a nonlinear transformation on each node in the second layer separately. Letting $w$ be the matrix of weights, $b_j$ be a linear bias associated to node $j$ of the second layer, and $\sigma(\cdot)$ be the nonlinear transformation, we have $y_j = \sigma(b_j + w_{ji} x_i)$. This process is then repeated for every subsequent layer. The output $f(\vec x)$ is extracted from the output node, or output nodes in the case of a vector-valued function.

A feed-forward network $f(\vec x)$ can be made to define a manifold by taking it to yield the imaginary part of a point on the manifold when given the real part as input:
\begin{equation}
\tilde A(A) = A + i f(A)
\text.
\end{equation}
Here, because $A$ has many components, there may in general be many functions $f(A)$, which must be trained separately. The number of separate functions required can be reduced by imposing translational invariance and other symmetries found in the target action. This construction of a manifold is not completely general, because it requires that there be only one point on the manifold with any given set of real coordinates. Nevertheless, this suffices to describe any manifold obtained by a sufficiently small amount of flow, and no example of a flowed manifold that folds back on itself (that is, for which the imaginary coordinate isn't a single-valued function of the real coordinate) has been found in practice.

\subsection{Training}

In order to find a suitable function $f(\vec x)$ for our manifold, a cost function must be defined, which we will seek to minimize. We may use the holomorphic gradient flow as a guide. After flowing some set of points from the real manifold, we obtain a training set of $\mathcal N$ points $(x,y)$ located on the flowed manifold. The cost function
\begin{equation}
C(w,b) = \frac 1 {\mathcal N}
\sum_{h= 1}^{\mathcal N}
\left|
\vec f_{w,b}(\vec x^(h)) - \vec y^(h)
\right|
\end{equation}
attempts to compute the imaginary part of each of these points by looking at the real part, and takes the average error.

It remains only to perform the minimization of the cost function; this minimization is usually done with some form of gradient descent algorithm. The space of biases and weights is of quite large dimension, and the cost function has many local minima, so some experimentation with different algorithms is advisable. An extensive review of gradient descent algorithms is given in~\cite{2016arXiv160904747R}; the Adaptive Moment Estimate algorithm~\cite{2014arXiv1412.6980K} (dubbed \textsc{Adam}) was used in~\cite{Alexandru:2017czx} for the purposes of training the manifold.

\section{Manifold Optimization}\label{sec:manifold-optimization}

The manifold learning method used a cost function defined by taking the distance between an ansatz manifold (e.g.\ a feed-forward network) and a set of training data constructed using the holomorphic flow. We can, however, use a more directly relevant cost function, constructed from the sign problem itself~\cite{Alexandru:2018fqp,mori2018application}. An immediate obstacle is that, if we chose the cost function to be the average sign evaluated on the manifold $\mathcal M$, we find that the task of evaluating the cost function on a given manifold is as difficult as measuring any other observable on that manifold. The average sign is a noisy observable, and it is difficult to measure precisely when there is a bad sign problem.

It turns out, though, that an inability to evaluate the cost function is no obstacle to its optimization. We select as our cost function the log of the average sign:
\begin{equation}
C(\mathcal M) = - \log \langle\sigma\rangle_{\mathcal M}
\text.
\end{equation}
Here we have written the cost function directly as a function of the manifold $\mathcal M$, and the average sign on that manifold is denotes $\langle\sigma\rangle_{\mathcal M}$. Where a family of manifolds parameterized by some $\lambda$ is used as an ansatz, this induces a cost function of the space of $\lambda$.

This cost function is no easier to evaluate. However, the derivative with respect to some manifold parameter $\lambda$ is quite simple, as a result of the fact that the physical partition function $Z$ cannot depend on the choice of manifold.
\begin{equation}
\frac{\partial}{\partial\lambda} C(\mathcal M_\lambda)
= - \frac{\partial}{\partial\lambda} \log Z_Q(\mathcal M_\lambda)
\end{equation}
We see that $\partial_\lambda C$ is a derivative of the log of the \emph{quenched} partition function. This is an expectation value of the quenched system, which can be evaluated without encountering a sign problem.At this point, we may again apply any minimization method to optimize the cost function, just as was done for the manifold learning procedure above.

This method has two substantial advantages over manifold learning. First, it does not require the potentially expensive step of preparing a library of flowed points to use as training data. Second, while the manifold learning procedure can at best be expected to perform (as measured by $\langle\sigma\rangle$) as well as the holomorphic flow, manifold optimization makes no reference to the flow and can in principle find manifolds with milder sign problems than any reached by flowing. We will see later that this is in fact the case for the Thirring model, even with a relatively simple ansatz.

\subsection{Another View of the Flow}

The flow was originally motivated by the observation that, in the limit of long flow times, the manifold would approach the Lefschetz thimbles, which generally have a substantially milder sign problem than the real plane. However, the flow has been found in practice to greatly improve the sign problem even for quite small flow times, without coming particularly close to the thimbles. This should be surprising: why does the flow perform so well, away from the regime where it is a well-motivated procedure?

The picture of manifold optimization above provides us with an answer. Take as an ansatz the family of manifolds $\tilde z(x)$ defined by interpolating from a fine mesh. The real plane itself is in this ansatz: the value of $y$ associated each $x$ in the mesh is $y_i(x) = 0$. Taking this as our starting point, we perform gradient descent on the cost function $C[y] = -\log \langle\sigma\rangle$. The gradient is
\begin{equation}
- \frac{\partial}{\partial y(x)}\log Z_Q
=
\frac 1 {Z_Q} \int e^{-\Re S} \frac{\partial \Re S}{\partial y(x)}
\text.
\end{equation}
We see that, starting from the real plane, the holomorphic gradient flow is in fact moving (in manifold space) in the direction which most quickly improves the average sign. Unfortunately, after the first infinitesimal step of flow has been performed, there is no longer a simple expression for the behavior of the manifold optimizing flow.

\section{Application to the Thirring Model}\label{sec:thirring}

The Thirring model~\cite{thirring1958soluble} is a common target for methods designed to alleviate or remove a fermionic sign problem. In $1+1$ dimensions, it is defined in the continuum by the Euclidean action
\begin{equation}\label{eq:thirring-action}
S=\int \d^{2}x\ \left[\bar\psi^\alpha (\slashed{\partial}+\mu \gamma_0 +m)\psi^\alpha 
+ \frac{g^2}{2N_F}\bar\psi^\alpha\gamma_\mu\psi^\alpha \bar\psi^\beta\gamma_\mu\psi^\beta\right],
\end{equation}
where the flavor indices take values $\alpha,\beta=1,\ldots,N_F$, $\mu$ is the chemical potential, and $\psi$ is a two-component spinor.

The four-fermi interaction is removed by introducing an auxilliary field $A_\mu$, which we take to be periodic with period $2\pi$. The resulting lattice action is 
\begin{equation}\label{eq:thirring-lattice-action}
S = \frac {N_F}{g^2} \sum_{x,\nu} (1 - \cos A_\nu(x)) 
+ 
\sum_{x,y}
\bar\psi^a(x) D_{xy}(A) \psi^a(y)
\end{equation}
where the spin index $a$ is implicitly summed over $a=1,2$. For Kogut-Susskind staggered fermions~\cite{kogut1975hamiltonian}, the matrix $D$ is defined by
\begin{equation}\label{eq:thirring-fermion-matrix}
D_{xy} = m\delta_{xy} + \frac{1}{2}\sum_{\nu=0}^2  
\Big[ 
 \eta_{\nu}(x) e^{i A_\nu(x)+\mu \delta_{\nu 0}} \delta_{x+\hat\nu, y}
 -\eta^\dag_{\nu}(y)e^{-i A_\nu(y)-\mu \delta_{\nu 0}}  \delta_{x, y+\hat\nu}
\Big]\text.
\end{equation}
This is of course not the only discretization possible. Another, with Wilson fermions~\cite{Wilson:1974sk}, yields the fermion matrix
\begin{equation}\label{eq:wilson-fermion-matrix}
D^W_{xy} = \delta_{xy} - \kappa \sum_{\nu=0,1}  
\Big[ 
 (1-\gamma_\nu) e^{i A_\nu(x)+\mu \delta_{\nu 0}} \delta_{x+\nu, y}
 + (1+\gamma_\nu) e^{-i A_\nu(x)-\mu \delta_{\nu 0}}  \delta_{x, y+\nu}
\Big]\text.
\end{equation}
Except where otherwise noted, statements in this chapter are applicable to both discretizations. As usual, because the lattice action is quadratic in the fermion fields, they can be integrated out of (\ref{eq:thirring-lattice-action}), yielding
\begin{equation}\label{eq:thirring-effective-action}
S = \frac{N_F}{g^2} \sum_{x,\nu} ( 1 - \cos A_\nu(x))
- \frac {N_F}{2} \log \det D(A)
\text.
\end{equation}
We will work in the case $N_F=2$.

The Thirring model has no sign problem at vanishing chemical potential. At finite chemical potential there is, as usual, a sign problem exponentially bad in the volume. The sign problem is made worse at larger couplings and larger chemical potentials. The sign problem of the Thirring model has been extensively investigated with flow-based methods, in both $0+1$ dimensions~\cite{Alexandru:2015xva,Alexandru:2015sua,Alexandru:2017oyw} and $1+1$ dimensions~\cite{Alexandru:2016ejd}. Attempts have also been made to approximate the Thirring model by integrating on a single thimble in isolation~\cite{Fujii:2015vha,Fujii:2015bua,DiRenzo:2020vou}.

\subsection{Field Complexification}
The path integral for the lattice Thirring model defined in this way is an integral over the manifold $S^{d V}$, that is, one copy of the unit circle for each variable $A$, of which there are $dV$, where $V$ is the volume and $d$ the dimension of the lattice. In order to apply the methods of complexification, we need to construct a space with complex structure which includes $S^{dV}$.

The complexification of $S^1$ is a cylinder: the set of points $(x,y)$ such that $x \in [0,2\pi)$ and $y$ is an unbounded real number. Under the exponential map, the original domain of $A_\nu$ maps to the unit circle in the complex plane. The full complexified space maps to the complex plane with one point removed, $\mathbb C\setminus\{0\}$. As the integration space is just the product of many copies of $S^1$, the complexification is the product of many cylinders. Topologically the space is $\left(S^1 \times \mathbb R\right)^{d V}$.

The real plane, which we will refer to as $\mathcal M_0$, may be defomed to another manifold $\mathcal M_1$ without changing the value of the path integral as long as the two manifolds together form the boundary of a closed region in $(S^1\times\mathbb R)^{dV}$. A sufficient condition is that there exists a homotopy between $\mathcal M_0$ and $\mathcal M_1$, that is, a continuous family of functions $g_t : \mathcal M_0 \rightarrow (S^1 \times \mathbb R)^{dV}$ from the real plane to the complex plane, such that $g$ is continuous in both $t$ and its argument, $g_0$ is the identity function, and the range of $g_1$ is $\mathcal M_1$. This sort of construction is also convenient from a computational point of view, as the function $g_1$ already provides a parameterization of the manifold to be integrated on.

For the purposes of the Thirring model, we can be even less general. Any manifold of the form
\begin{equation}
\tilde A_\nu(x) = A_\nu(x) + i f_\nu^{(x)}(\vec A)
\end{equation}
such that $f_\nu^{(x)}$ is a continuous function in its $dV$ arguments, is homotopically connected to the real plane. The homotopy is constructed by scaling the function $f_\nu^{(x)}$ by $t$.

\subsection{An Ansatz}
This is the manifold ansatz we will consider~\cite{Alexandru:2018fqp}:
\begin{equation}
\tilde A_0(x) = A_0(x) + i f(A_0(x)) \;\text{ and }\; \tilde A_i(x) = A_i(x)
\text.
\end{equation}
This is an enormously constrained ansatz: we have left all dimensions other than $\nu=0$ undeformed, and the deformation of the integral over $A_0(x)$ does not depend on the value of $A$ at any other link. Additionally, we have used the fact that the action is translationally invariant to infer that the ansatz should be as well. Nevertheless, the ansatz still has an infinite number of parameters. The requirement that $f$ be a continuous function suggests that we expand it in a fourier series:
\begin{equation}\label{eq:thirring-ansatz}
f(z) = a_0 + a_1 \cos(x) + a_2 \cos(2x) + \cdots + b_1 \sin(x) + \cdots
\text.
\end{equation}
The action is symmetric under $A_0 \rightarrow -A_0$, which suggests that the chosen manifold ought to be as well, so we can set $b_i = 0$. Finally, to have a finite number of manifold parameters, we truncate the fourier series to the first $3$ even terms, with coefficients $a_0$, $a_1$, and $a_2$.

\subsection{Phase Diagram}

Now that a plausible ansatz is constructed, it can be optimized with the methods of Section~\ref{sec:manifold-optimization}. The manifold parameters depend on the model parameters, so the optimization must be performed separately for every set of model parameters.

We simulate the $2+1$-dimensional Thirring model~\cite{Alexandru:2018ddf} with bare lattice parameters $g = 1.08$ and $m = 0.01$. All results are quoted in lattice units; physical quantities may be recovered by multiplication with the appropriate power of the lattice spacing. This choice of $g$ and $m$ puts the lattice model in the strong coupling regime: in a box of size $10^2$, we measure a fermion mass $m_f = 0.46(1)$ and a boson mass $m_b = 0.21(1)$. 

We focus on the chiral condensate, defined by the expectation value $\left<\bar\psi\psi\right>$. At low temperatures and low chemical potentials, the chiral condensate has a non-zero value, indicating the breaking of chiral symmetry. At either high temperature or larger chemical potential, chiral symmetry is nearly (because $m > 0$) restored, and the chiral condensate drops to near zero.

\begin{figure}
	\centering
	\includegraphics[width=0.48\textwidth]{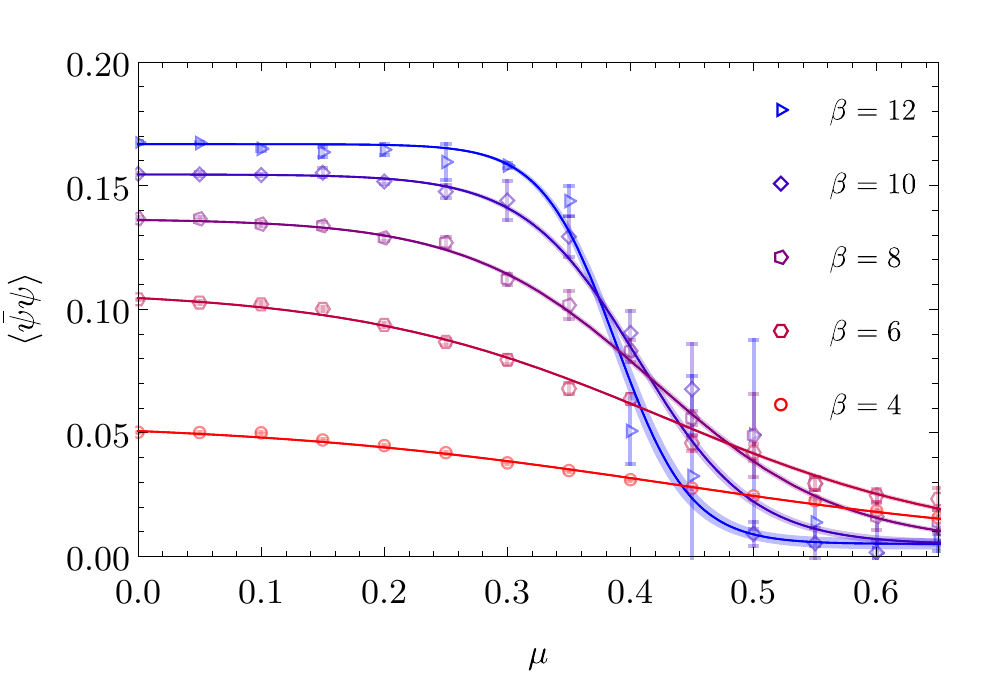}
	\includegraphics[width=0.48\textwidth]{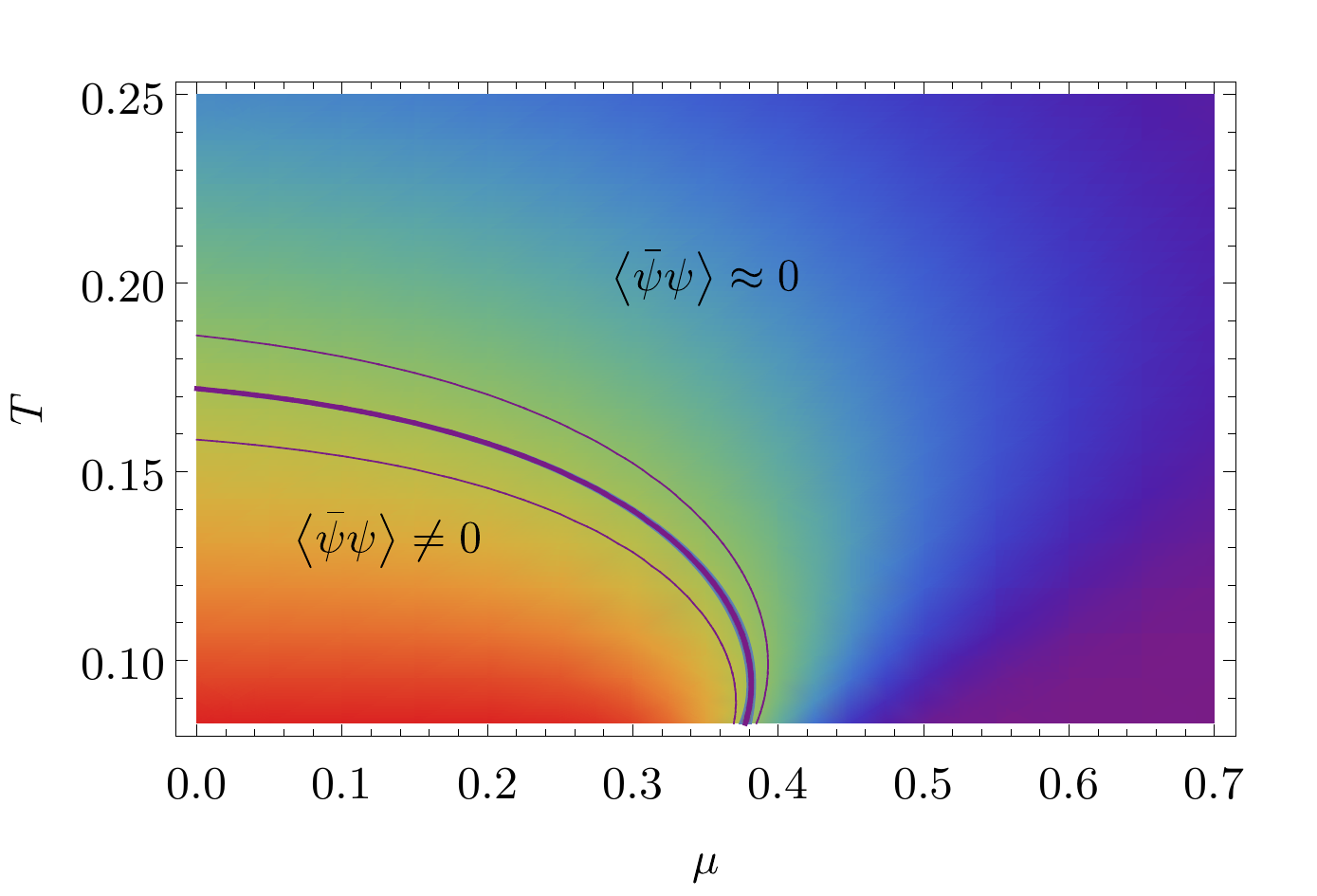}
	\caption{Figure from~\cite{Alexandru:2018ddf} of the chiral symmetry breaking phase transition of the Thirring model. On the left, the condensate $\left<\bar\psi\psi\right>$, as a function of chemical potential $\mu$ on $\beta\times 6^2$ lattices. On the right, the full $T$-$\mu$ plane for the same spatial volume. The central band indicates the location of $\left<\bar\psi\psi\right>_{\mu,T} = 0.5 \left<\bar\psi\psi\right>_{0}$; the thinner lines on either side indicate the width of the crossover\label{fig:transition}}
\end{figure}

Figure~\ref{fig:transition} shows measurements of the chiral condensate on a $6^2$ spatial lattice; the size of the time dimension is temperature-dependent. At low temperatures, a relatively sharp transition between the broken and unbroken phases is seen near $\mu\sim 0.4$. The crossover broadens at higher temperatures, and moves to lower chemical potentials.

\section{Optimal Manifolds}

The general method of complexification may fail for two different reasons. For any given model, the complexification method may fail because no manifold that removes the sign problem exists, or it may fail because the manifold is just too computationally expensive to integrate on (for instance, it may be difficult to find in the first place!). In practice it is difficult to distinguish these two failure modes unless we can prove no satisfactory manifold exists. In this section we consider general questions about the ``best possible'' manifolds, that is, those that minimize the quenched partition function, and look in particular at one case where a ``perfect'' manifold can be found, and at another where we can prove none exists.

\subsection{Lefschetz Thimbles Are Not Optimal}\label{sec:suboptimal}

\begin{figure}
\centering
\includegraphics[width=3in]{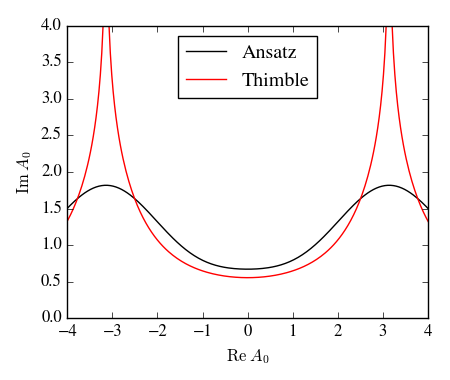}
\caption{Thimbles and an optimized manifold from the ansatz \ref{eq:thirring-ansatz} for the heavy-dense limit of the Thirring model, or equivalently the one-site Thirring model.\label{fig:thirring-thimble-picture}}
\end{figure}

The lattice Thirring model (\ref{eq:thirring-lattice-action}) becomes trivial in the heavy-dense limit of large chemical potential $\mu$. Physically, the large chemical potential pushes the Fermi momentum up to the lattice cutoff, so that every site of the lattice is filled by a fermion. To see this effect algebraically, we can expand the fermion determinant in powers of $e^{-\mu}$. The leading-order term is the one in which only time-like links are included. The physical `saturation' effect of the lattice manifests in the partition function factorizing to leading order $e^{-\mu}$, so that each link is now independent and uncoupled from all other links:
\begin{equation}
	Z = \left[\int \d A_0 \;e^{\frac {N_F} {g^2} \cos A_0 + \mu + i A_0}\right]^{\beta L} \left[\int \d A_1\; e^{\frac {N_F} {g^2} \cos A_1}\right]^{\beta L}\text.
\end{equation}
With exponents appropriately modified, this factorization holds independent of the number of spacetime dimensions. This observation provides a convenient \emph{post-hoc} rationalization for the ansatz of Section~\ref{sec:thirring}: that ansatz defines the most general manifold which maintains all the symmetries of the heavy-dense limit of the Thirring model\footnote{Strictly speaking, we have also imposed that $\Im A_0$ be a single-valued function of $\Re A_0$.}.

This trivial limit also allows us to study how the Lefschetz thimbles compare to a ``best-possible'' manifold~\cite{Lawrence:2018mve}. Because the partition function factorizes, the average sign does as well:
\begin{equation}
\langle\sigma\rangle
=
\left(\frac{Z_1}{Z_{1,Q}}\right)^{\beta L}
\end{equation}
where $Z_1$ and $Z_{1,Q}$ denote the partition function and quenched partition function, respectively, of the one-site model. This allows us to accurately compute the average sign at large volumes, where a direct calculation of $\left<\sigma\right>$ would be impractical. The thimbles and the ansatz of that model are shown in Figure~\ref{fig:thirring-thimble-picture}.

Numerical evidence indicates that the ansatz (with all Fourier coefficients maintained) achieves an average phase of exactly $1$, regardless of coupling. The average phase on the thimbles, meanwhile, is less than $1$ for any nonvanishing coupling. In the case of $g = 1.08$, the average phase obtained on the thimbles is $0.985$ for the one-link model, and therefore $(0.985)^{\beta L}$ for the full lattice.

This heavy-dense limit of the Thirring model is one of the few physically-inspired models in which both the Lefschetz thimbles and a provably optimal manifold can be understood. In this case, a manifold that completely solves the sign problem does exist, and the thimbles do not.

The fact that the sign problem of this model can be solved exactly may be very special --- in fact, in the next section, we will see a (less physical) model in which no manifold can solve the sign problem. However, the fact that the Lefschetz thimbles are non-optimal is probably less special.  Looking at Figure~\ref{fig:thirring-thimble-picture}, note that the thimbles contain a sharp `cusp', where the contour folds back along itself. At stronger couplings, the cusp becomes sharper, and the thimbles come closer to each other. The action on one side of the cusp does not differ much from the action on the other, but the sign of the integration element $\d z$ flips. Therefore, the contributions to the integral from the two sides of the cusp almost exactly cancel. The ansatz manifold cuts off the cusp, removing the considerable residual sign problem in this region.

\subsection{A Complexification-Immune Sign Problem}

Not every sign problem can be removed with complexification. A simple example serves to prove the point:
\begin{equation}
Z(\epsilon) = \int_0^{2\pi} \d \theta\; \left(\epsilon + \cos \theta\right)
\text.
\end{equation}
This partition function should be thought of as a `lattice' with a single degree of freedom $\theta$ on a single site, and an action of $S_\epsilon(\theta) = - \log\left(\epsilon + \cos\theta\right)$. We are interested in the regime of small $\epsilon$, where we will be able to establish upper bounds on the best possible average phase $\langle\sigma\rangle$.

We begin by taking $\epsilon = 0$, where the partition function vanishes. Here there are two thimbles, constituting the two halves of the real line, and the antisymmetry of the Boltzmann factor under $\theta \rightarrow \theta+\pi$ causes them to exactly cancel. The partition function will vanish no matter what manifold is chosen. The quenched partition function depends strongly on the manifold, but can be rigorously bounded from below. For any $x \in [0,2\pi)$, the chosen manifold must have a point $\theta$ with $\Re \theta = x$. Because $\cos\theta$ is minimized, for any fixed $\Re\theta$, by $\Im\theta = 0$, the minimum possible value of $|e^{-S(x)}|$ is achieved at $\theta = x$. Therefore we can do no better in minimizing $Z_Q$ than to integrate over the real line, resulting in the bound
$Z_Q \ge 4$.
Because the partition function itself vanishes, the average sign will always be zero. This is a pathological example.

The pathology is lifted by introducing $\epsilon$. At small $\epsilon>0$, the thimbles remain on the real line, but the cancellation is no longer exact, and the partition function no longer vanishes, but is instead given by $Z = 2\pi \epsilon + O(\epsilon^2)$. The average sign, then, is forced to be of order $\epsilon$ as well:
\begin{equation}
\langle\sigma\rangle \le \epsilon \frac{\pi}{2} + O(\epsilon^2)
\text.
\end{equation}
We conclude that, at small but nonvanishing $\epsilon$, there is no manifold that can improve the sign problem beyond what is achieved on the real line. Moreover, even with only one degree of freedom, the sign problem on the real line can be made arbitrarily bad.

A key feature of this example is the presence of multiple cancelling thimbles. As mentioned earlier, the other way in which the Lefschetz thimbles fail to completely remove the sign problem is via the residual phase introduced by the Jacobian. Whether this residual phase can always be counteracted by deforming away from the thimbles (at in the heavy-dense limit of the Thirring model) remains an open question.

\chapter{Quantum Simulations}\label{ch:quantum}

In this chapter we discuss the use of a quantum computer in studying the time-evolution of physical quantum systems. Sections \ref{sec:quantum-computer} and \ref{sec:quantum-simulation} provide introductions to quantum computing and quantum simulations, respectively; however, a cursory overview of quantum computers suffices to show that they are a powerful tool for studying quantum systems.

After abstracting away implementation details\footnote{Crucially, this also requires abstracting away the fact that, as of this writing, no quantum computer exists at the scale necessary to perform any field theory simulation discussed in this chapter.}, a quantum computer consists of a set of qubits, and the ability to apply arbitrary unitary operations on pairs of qubits. The state of a single qubit is described by a two-dimensional Hilbert space $\mathcal H_2 = \mathrm{span} \{\ket{0},\ket{1}\}$. In a computer with $N$ qubits, the full Hilbert space is given by the tensor product of $N$ copies of the single-qubit Hilbert space, $\mathcal H_{QC} = \mathcal H_1^{\otimes N}$. This Hilbert space describes the set of possible states of the quantum computer; in addition, there is a set of unitary operations (termed `gates') on this Hilbert space, which may be performed in any order in order to manipulate the qubits.

We are interested in simulating some physical system, described by a Hilbert space $\mathcal H$, and a time-evolution operator $U(t) = e^{-i H t}$. Here it becomes clear how a quantum computer might be useful. Two Hilbert spaces of equal dimension are necessarily isomorphic, so it is possible to establish a mapping between the physical Hilbert space $\mathcal H$ and (some linear subspace of) the quantum computer's $\mathcal H_{QC}$. If, after this mapping is established, the time-evolution operator $U(t)$ can be efficiently implemented in terms of the available quantum gates, then it will be possible to simulate time-evolution of the physical system with the quantum computer. In Section \ref{sec:quantum-simulation} we will see that, as shown in \cite{Lloyd1073}, this is true for a large class of physically relevant systems.

\section{Digital Quantum Computers}\label{sec:quantum-computer}

In this section we give an expedited overview of quantum computation, tailored to those aspects which will be important in designing quantum simulations.

\subsection{A Single Qubit}
For physical intuition, one may think of a single qubit as being implemented by a quantum spin-$1/2$ system, although any two-state system will suffice and many are used in practice. The state of a single qubit is a vector in the Hilbert space $\mathcal H_1 = \mathrm{span}\{\ket{0},\ket{1}\} \approx \mathbb C^2$. There are two types of manipulations we perform on a qubit: quantum gates acting on 1 or 2 qubits, which correspond to unitary $2\times 2$  or $4\times 4$ matrices, and measurements, which yield classical information while collapsing the state of the qubit.

The set of unitary operators on $\mathbb C^2$ is denoted $U(2)$. An overall phase on a quantum state cannot be measured and is treated as physically irrelevant. For this reason, the set of physically distinct quantum operations on one qubit is actually $U(2) / U(1) \approx SU(2)/\mathbb Z_2$.

Implementing the uncountable set of operations directly is often inconvenient, particularly when constructing an error-correct quantum commputer. Instead, one implements a small discrete subset of these operations (fundamental gates), such that any unitary operator can be arbitrarily well approximated by a sequence of fundamental gates. A common set of fundamental gates are
\begin{equation}
H = \frac{1}{\sqrt{2}}\left(\begin{matrix}
1 & 1\\
1 & -1
\end{matrix}
\right)
\hspace{2em}
\text{and}
\hspace{2em}
T = \left(
\begin{matrix}
e^{i \pi / 8} & 0\\
0 & e^{-i \pi / 8}
\end{matrix}
\right)
\text,
\end{equation}
and as shown by Solovay and Kitaev, any operation in $U(2) / U(1)$ can be approximated to within $\epsilon$ with $O(1/\epsilon)$ gates chosen from this set \cite{1997RuMaS..52.1191K,solovay1999lie,dawson2005solovaykitaev}.

The gate $T$, often referred to as the $\frac \pi 8$-gate, is the exponential of the Pauli matrix $\sigma_z$. Often $T^\dagger$ is included in the set of fundamental gates, but it can of course be obtained as $T^\dagger = T^3$ (note that this equality is true in $U(2)/U(1)$, and not in $U(2)$). The Hadamard gate $H$, which is its own inverse, corresponds to a change-of-basis between the $x$ and $z$ bases: $H \sigma_z H = \sigma_x$.

In addition to gates, we may also perform measurements. For our purposes, we will consider all measurements to be performed in the $z$-basis of $\ket{0}$ and $\ket{1}$. If a qubit is in state $\alpha \ket{0} + \beta \ket{1}$, measurement changes the state to $\ket{0}$ (resp.\ $\ket{1}$) with probability $|\alpha|^2$ (resp.\ $|\beta|^2$), and yields the classical bit $0$ (resp.\ $1$).

\subsection{Coupling Many Qubits}

A quantum computer with only one qubit is of no use --- after all, it can be efficiently simulated by a classical computer by multiplying $SU(2)$ matrices.  Quantum computers become interesting when we add the ability to operate on multiple qubits simultaneously. The resulting operations, with $N$ qubits, are unitaries on the Hilbert space $\mathcal H_1^{\otimes N}$, modulo an irrelevant overall phase. This is the group $U(2^N)/U(1)$.

It is sufficient to add to our set of primitive gates only a single extra gate\footnote{In fact, universal quantum computation can be achieved with almost any multi-qubit gate~\cite{deutsch1995universality}.}, which couples two qubits~\cite{nielsen2001quantum}. A common choice is the controlled-not gate, defined by
\begin{equation}
CX = \left(
\begin{matrix}
1 & 0 & 0 & 0\\
0 & 1 & 0 & 0\\
0 & 0 & 0 & 1\\
0 & 0 & 1 & 0
\end{matrix}
\right)
\text.
\end{equation}
This gate may be thought of as acting on a `control' and a `target' qubit: the target qubit is flipped when the control qubit is $1$. The matrix above is written in the basis $\{\left|00\right>, \left|01\right>, \left|10\right>, \left|11\right>\}$, so that the first (``high-order'') qubit is the control.

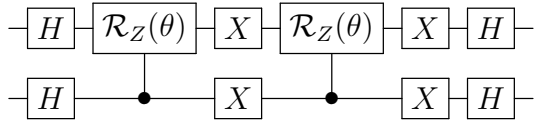
\begin{figure}
\centerline{
\Qcircuit @C=0.6em @R=0.8em {
& \gate{H} & \gate{\mathcal R_Z(\theta)} & \gate{X} & \gate{\mathcal R_Z(\theta)} & \gate{X} & \gate{H} & \qw \\
& \gate{H} & \ctrl{-1} & \gate{X} & \ctrl{-1} & \gate{X} & \gate{H} & \qw \\
}}
\caption{A two-qubit circuit, implementing time-evolution under the Hamiltonian $H = \sigma_x \otimes \sigma_x$.\label{fig:ex-circuit}}
\end{figure}
A quantum circuit is a composition of primitive gates. A simple example is shown in Figure~\ref{fig:ex-circuit}, in which a change of basis is achieved by Hadamard gates, followed by a phase rotation of the $\left|11\right>$ state, followed by a \texttt{NOT} gate on each qubit and another phase rotation. The result is equivalent to time-evolution under a Hamiltonian $H = \sigma_x \otimes \sigma_x$.

\subsection{Some Simple Algorithms}

A small number of quantum algorithms will be directly relevant to the task of creating a simulation, and are introduced here.

The first algorithm to discuss is really a meta-algorithm --- that is, a procedure for producing a quantum algorithm: any reversible classical circuit may be re-interpreted as a quantum circuit. If a reversible quantum circuit starts with some bitstring $x$ and yields $x + f(x)$ (that is, the concatenation of $x$ and $f(x)$), then the corresponding quantum circuit yields the transformation $\left|x\right>\left|0\right> \rightarrow \left|x\right>\left|f(x)\right>$. (As $\left|x\right>$ and $\left|f(x)\right>$ are both computational basis states, this is a unitary  transformation.)

As any classical circuit can be easily transformed to be reversible~\cite{Toffoli:1980}, this implies that any classical algorithm yields a quantum algorithm. This procedure requires polynomially many (in the memory size of the original, non-reversible classical circuit) ancillary qubits.

\subsubsection{Controlled Nots}

The controlled-not operation is usually considered to be a primitive gate, implemented directly by the quantum computing hardware. It is usually convenient, however, to make use of many-controlled nots, such as the Toffoli gate (here called \texttt{CCX}) defined by
\begin{equation}
U_{\texttt{CCX}} \left|110\right> = \left|111\right>
\>\text{ and }\>
U_{\texttt{CCX}} \left|111\right> = \left|110\right>
\text,
\end{equation}
and acting as the identity on all other basis states. These can be implemented from the single-controlled not \texttt{CX} and general one-qubit gates. Furthermore, arbitrarily-controlled nots \texttt{C$^n$X} may be efficiently constructed (with $poly(n)$ gates) by introducing ancillary qubits. In fact this can be improved to remove the need for ancilla~\cite{saeedi2013linear,barenco1995elementary}, but these constructions will not be discussed here.

\begin{figure}
\[
\Qcircuit @C=0.6em @R=0.8em {
& \ctrl{1} & \qw & & & &
& \qw & \qw & \qw & \ctrl{2} & \qw & \qw & \qw & \ctrl{2} & \qw & \ctrl{1} & \gate{T} & \ctrl{1} & \qw \\
& \ctrl{1} & \qw & & = & &
& \qw & \ctrl{1} & \qw & \qw & \qw & \ctrl{1} & \qw & \qw & \gate{T} & \gate{X} & \gate{T^\dagger} & \gate{X} & \qw\\
& \gate{X} & \qw & & & &
& \gate{H} & \gate{X} & \gate{T^\dagger} & \gate{X} & \gate{T} & \gate{X} & \gate{T^\dagger} & \gate{X} & \gate{T} & \gate{H} & \qw & \qw & \qw \\
}
\]
\caption{Implementation of the Toffoli gate from 1- and 2-qubit gates.\label{fig:toffoli-circuit}}
\end{figure}
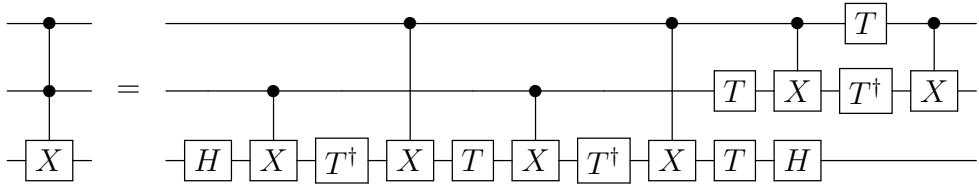

A common construction~\cite{nielsen2001quantum} of the twice-controlled not gate is shown in Figure~\ref{fig:toffoli-circuit}. This is in fact the construction that minimizes the number of \texttt{CX} gates required~\cite{shende2008cnot}. From this building block, $n$-controlled not gates may be constructed with the introduction of $n-2$ ancillary qubits, as demonstrated in Figure~\ref{fig:cnnot-circuit}. Intuitively, each Toffoli may be thought of as adding two bits (base 2), and storing the result in the target bit. Eventually, the sum of all bits is accumulated.

\begin{figure}
\centerline{
\Qcircuit @C=0.6em @R=0.8em{
\left|0\right> & & \gate{X} & \qw & \ctrl{6} & \qw & \gate{X} & \qw & \left|0\right>\\
\left|0\right> & & \qw & \gate{X} & \ctrl{5} & \gate{X} & \qw & \qw & \left|0\right>\\
& & \ctrl{-2} & \qw & \qw & \qw & \ctrl{-2} & \qw &\\
& & \ctrl{-3} & \qw & \qw & \qw & \ctrl{-3} & \qw &\\
& & \qw & \ctrl{-3} & \qw & \ctrl{-3} & \qw & \qw &\\
& & \qw & \ctrl{-4} & \qw & \ctrl{-4} & \qw & \qw &\\
& & \qw & \qw & \gate{X} & \qw & \qw & \qw &\\
}
}
\caption{Construction of a $4$-controlled not gate from the Toffoli gate, with the aid of ancilla. The ancillary qubits are on the top; the target is the bottom qubit.\label{fig:cnnot-circuit}}
\end{figure}
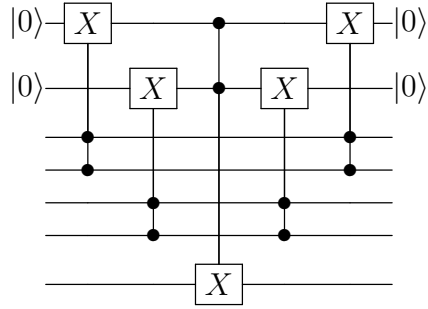

\subsubsection{Circuits From Circuits}\label{sec:cfromc}

Given a quantum circuit implementing a unitary $U$, certain related unitaries can be easily obtained by modification of the quantum circuit. For concreteness, we will assume that the circuit is implemented from the Hadamard gate, the $\frac \pi 8$ gate, and the controlled-not. This constraint can be substantially relaxed without materially changing the methods discussed in this section.

Our first task, given a circuit implementing $U$, is to obtain a circuit implementing the inverse $U^\dagger$. Note first that, for every gate in our gateset, the inverse gate is already known: the Hadamard and \texttt{CX} gates are their own inverses, and $T^\dagger = T^3$. This allows us to construct an inverse circuit simply by inverting each gate and reversing the order of application.
\begin{equation}
U^\dagger
=
\left[
\prod_{i=1}^K V_i
\right]^\dagger
= 
\prod_{i=K}^1 V_i^\dagger
\end{equation}

Our second task, again given a quantum circuit implementing $U$, is to implement the controlled-$U$ operation $U_C$ defined by
\begin{equation}
U_C\left( \left|0\right>\otimes\left|\Psi\right>\right) = \left|0\right>\otimes\left|\Psi\right>
\>\text{ and }\>
U_C \left(\left|1\right>\otimes\left|\Psi\right>\right) = \left|0\right>\otimes U\left|\Psi\right>
\end{equation}
General constructions of controlled circuits are given in~\cite{barenco1995elementary}; for our purposes, we will specialize to the case where the available gates are $H$, $T$, $T^\dagger$, and $CX$. The technique stems from the observation that if $U = V V'$ can be decomposed as a product of (potentially simpler) unitaries, then the controlled unitary is given by the product $U_C = V_C V'_C$ of the controlled versions of the simpler unitaries. By assumption we have $U$ expressed as a product of the fundamental gates, and so we need only construct controlled versions of those four gates.

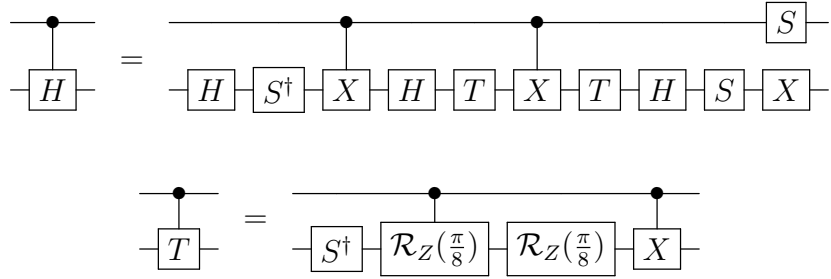
\begin{figure}
\[
\Qcircuit @C=0.6em @R=0.8em {
& \ctrl{1} & \qw & & & & & \qw & \qw & \ctrl{1} & \qw & \qw & \ctrl{1} & \qw & \qw & \qw & \gate{S} & \qw & \\
& \gate{H} & \qw & & \raisebox{1.8em}{=} & & & \gate{H} & \gate{S^\dagger} & \gate{X} & \gate{H} & \gate{T} & \gate{X} & \gate{T} & \gate{H} & \gate{S} & \gate{X} & \qw & \\
}
\]
\vspace{1em}
\[
\Qcircuit @C=0.6em @R=0.8em {
& \ctrl{1} & \qw & & & & & \qw & \ctrl{1} & \qw & \ctrl{1} & \qw & \\
& \gate{T} & \qw & & \raisebox{1.8em}{=} & & & \gate{S^\dagger} & \gate{\mathcal R_Z(\frac\pi{8})} & \gate{\mathcal R_Z(\frac\pi{8})} & \gate{X} & \qw & \\
}
\]
\caption{Construction of the controlled Hadamard (top) and controlled $\pi/8$ (bottom) gates, from the primitive gateset of $H$, $\mathcal R_Z(\theta)$, and $CX$. For brevity, an additional gate is defined as $S = T^2$, and we have used $X = HT^4H$.\label{fig:ch-cphase}}
\end{figure}

The Toffoli gate $CCX$ has been constructed above. Constructions of the controlled Hadamard and controlled $\frac \pi 8$ gates are shown in Figure~\ref{fig:ch-cphase}.

\subsubsection{Quantum Fourier Transform}\label{sec:quantum-fourier-transform}

The Fourier transform, suitably generalized to nonabelian groups is a change-of-basis operation between the regular representation of a group and the Fourier basis. This operation is unitary, and therefore can be implemented by quantum circuits. An efficient (polylogarithmic in the size of the group) implementation for arbitrary groups is not known; nevertheless, circuits are known for large classes of groups~\cite{puschel1999fast}, including abelian groups~\cite{coppersmith2002approximate}.

The quantum Fourier transform is central to major quantum algorithms, most notably Shor's factoring algorithm~\cite{shor1994algorithms} and Grover's search algorithm~\cite{grover1996fast}. For quantum simulations, the relevance of the quantum Fourier transform comes from the fact that it diagonalizes the quantum-mechanical kinetic term $\nabla^2$ of the Hamiltonian. We will see that the asymptotic scaling of the quantum Fourier transform with group size is not very relevant to the efficiency of lattice simulations.

\section{Quantum Simulations in General}\label{sec:quantum-simulation}

\subsection{Mapping Hilbert Spaces}

As discussed at the beginning of this chapter, the first step to setting up a simulation of a physical system on a quantum computer is to establish a mapping between the Hilbert space of the physical system and that of the quantum computer. The Hilbert space of the quantum computer is necessarily finite-dimensional, having $2^Q$ dimensions for a quantum computer with $Q$ qubits; therefore, we are constrained to consider a physical system with a similarly finite Hilbert space. In particular, we will be working at finite volume and lattice spacing.

Any two Hilbert spaces of equal finite dimension are isomorphic. However, it is helpful to have the mapping between the states of the quantum computer and the physical states be a natural one. In particular, in simulating a field theory, it is often helpful to preserve the notion of locality. Many lattice field theories have a Hilbert space which is naturally expressed as a tensor product of many local Hilbert spaces, each associated to some lattice site (or link, for gauge theories). By grouping qubits, and allowing each group to correspond to a single site, we can write the state space of the quantum computer similarly as a tensor product of simpler, local Hilbert spaces.

The nature of the mapping of local Hilbert spaces depends strongly on the system being simulated. The simplest case is when the local Hilbert space is two-dimensional, therefore mapping cleanly to a single qubit. This is true for a spin chain or the $\mathbb Z_2$ gauge theory, and we will use these systems as illustrative examples below.

\subsection{Suzuki-Trotter Decomposition} 

The basic gates discussed above can be viewed as time-evolution under simple Hamiltonians affecting one or two qubits at a time. The one-qubit gates, applied to site $i$, yield evolution under $\vec n \cdot \vec\sigma_i$, while the two qubit gates applied to sites $i$ and $j$ yield evolution under $\sigma^\mu_i \otimes \sigma^\nu_j$ for any $\mu,\nu\in\{x,y,z\}$. A general Hamiltonian of interest, however, is not so simple. The Heisenberg spin chain, for example, has the Hamiltonian
\begin{equation}
H = 
- J_x \sum_{\langle i j\rangle} \sigma^x_i \otimes \sigma^x_j
- J_y \sum_{\langle i j\rangle} \sigma^y_i \otimes \sigma^y_j
- J_z \sum_{\langle i j\rangle} \sigma^z_i \otimes \sigma^z_j
\text.
\end{equation}
Note that because the terms of the Hamiltonian do not commute, the unitary time-evolution operator does not factorize. More complicated systems of physical interest will also have Hamiltonians which can be expressed as sums of few-qubit Hermitian operators which fail to mutually commute.

The Suzuki-Trotter decomposition~\cite{suzuki1976generalized,trotter1959product} provides an approximate factorization of the time-evolution operator in the case where the terms of the Hamiltonian do not commute. With two terms in the Hamiltonian $H = A + B$, we have
\begin{equation}
e^{-i \delta(A+B)} \approx e^{-i\delta A} e^{-i\delta B} + O(\delta^2)
\text.
\end{equation}
If the decomposition is such that $A$ and $B$ have well-understood diagonal bases, the operators $e^{- i A t}$ and $e^{-i B t}$ are readily implemented, and therefore we have an easy implementation of approximate time-evolution, becoming exact in the limit $\delta \rightarrow 0$.

In general, this procedure can be generalized to any sparse efficiently computable sparse Hamiltonian~\cite{Lloyd1073}. In the case of field theories, the Hamiltonian can typically be split into two terms, one diagonal in field space and the other diagonal in conjugate momentum space; thus the general theorem is not needed.

\subsection{Aside: Disordered Potentials}

Here we describe the simulation~\cite{Alexandru:2019dmv} of a particular quantum-mechanical system: the Anderson tight-binding model~\cite{Anderson:1958vr} of a single particle living on $V$ lattice sites. This is a model of a particle in a random potential. The Hamiltonian is
\begin{equation}
\label{eq:site-disorder}
H = - \sum_{<ij>} \left(c^\dagger_i c_j + c^\dagger_j c_i\right)
+
W \sum_i u_i c^\dagger_i c_i
\text.
\end{equation}
Here the $u_i$ are random variables, taken to be independently and identically distributed on the interval $[0,1]$. The first sum is taken over all pairs of neighboring sites, and the parameter $W$ gives the strength of the disorder.

Depending on $W$ and the dimension of the lattice, this model may exhibit \emph{Anderson localization}. When $W = 0$, this is a model of a free particle in a box. A wavefunction that begins concentrated at one site, spreads out throughout the lattice over time. For sufficiently strong disorder, however, an initially concentrated wavefunction will remain concentated at all later times~\cite{Anderson:1958vr,doi:10.1142/7663,RevModPhys.80.1355}. This localization effect is not due to any potential well trapping, but rather interference effects between the different paths a particle could take to propagate.

In one or two dimensions, any amount of disorder yields Anderson localization. In three dimensions, Anderson localization only sets in above the critical disorder of $W_c \approx 16.5$~\cite{brandes2003anderson}, with a second-order transition at that point.

\subsubsection{Simulation}

The Hilbert space of this model is of dimension $V$. At large volumes, where the second-order phase transition is most visible, it becomes numerically difficult to simulate: naive algorithms run in time at least $O(V^2)$. When performing a simulation on a quantum computer, we expect to require only $O(\log V)$ qubits to represent the Hilbert space and similarly $O(\log V)$ operations per time step, indicating that far larger volumes can be obtained at relatively little cost.

We first discuss the simulation of a particle in a one-dimensional random potential~\cite{Alexandru:2019dmv}. Each of the $V$ sites is labelled by an integer $0\ldots V-1$, and the state $\left|i\right>$ is the position eigenstate of an electron located at site $i$. This will be the computational basis. The mapping of the computation basis of the quantum computer is achieved by representing the integer $i$ in binary; thus the physical state $\left|6\right>$ is mapped to the state $\left|110\right>$ on the computer.

The time-evolution is simulated according to the Suzuki-Trotter decomposition, splitting the Hamiltonian into three pieces $H = H_{K,e} + H_{K,o} + H_V$, where:
\begin{align}
H_{K,e} &= \sum_{i = 0,2,\ldots} \left(c^\dagger_i c_j + c^\dagger_j c_i\right)\\
H_{K,o} &=\sum_{i = 1,3,\ldots} \left(c^\dagger_i c_j + c^\dagger_j c_i\right)\\
H_{V} &= W \sum_i u_i c^\dagger _i c_i
\text.
\end{align}
Spliting the kinetic term into even and odd links in this fashion allows it to be simulated without a change of basis to momentum-space, and allows the algorithm to generalize to the case where the disorder in the potential lives on the links instead of the sites. When simulating in $d$ dimensions, this Trotterization scheme requires $2d+1$ steps.

The even links couple states that differ only in the last qubit
\begin{equation}
H_{K,e} = \left|0\right>\left<1\right| + \left|2\right>\left<3\right| + \cdots + \mathrm{h.c.} = I^{\otimes (V-1)} \otimes \sigma_x
\text,
\end{equation}
and $e^{-i H_{K,e} \Delta t}$ is therefore obtained by a rotation about the $X$ axis of the least-significant qubit. The evolution of the odd links can be put into a similar form by first shifting the whole lattice by $1$. This is a change of basis that maps $H_{K,o}$ to $H_{K,e}$. The shift corresponds to the addition of $1$ modulo $V$, for which a classical circuit (and therefore quantum circuit) is readily constructed.

The evolution under the disordered potential is, on its face, more difficult. This evolution requires the phase of the state to be changed by the same random number each time the electron finds itself as a particular site. In a classical simulation, this effect is accomplished by generating a list of $V$ random numbers at the beginning of the computation. This step already exponentially exceeds our $O(\log V)$ budget.

To avoid this, we note that the $u_i$ are typically not truly random variables, but instead are defined to be the output of a pseudo-random number generator (PRNG) with a seed chosen in advance.  A PRNG is a circuit sufficiently complicated that the $u_i$ look random to any practical statistical test. Given a PRNG $f(i)$ returning a $Q$-bit number, we may take $u_i = 2^{-Q} f(i)$. Critically, there exist \emph{seekable} PRNGs, from which the $i$th element $f(i)$ can be obtained in fixed time for any $i$. Two appropriate constructions of PRNGs are discuseed below.

Given a classical circuit for a suitable PRNG, we can construct a quantum circuit $U_f$ defined by $U_f \left|i\right>\left|0\right> = \left|i\right>\left|f(i)\right>$. The evolution under $H_V$ is then implemented by applying $U_f$ to compute the PRNG, and a diagonal phase rotation by the value specified in the anciliary register.

The resulting circuits corresponding to kinetic and potential evolution are shown in Figure~\ref{fig:anderson-circuits}.

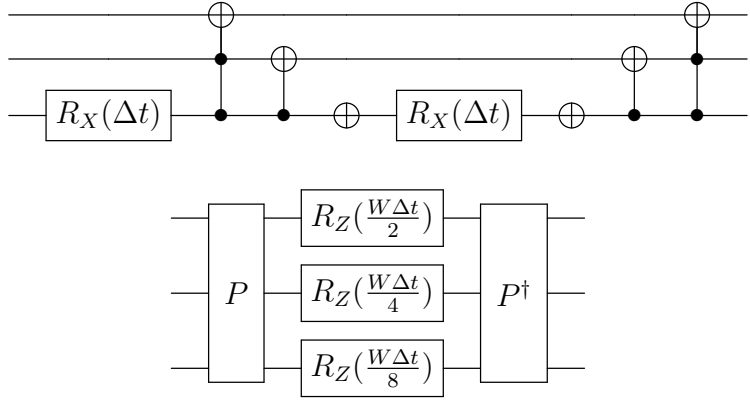
\begin{figure}
\centerline{
\Qcircuit @C=1.2em @R=0.6em {
& \qw & \targ{2} & \qw & \qw & \qw & \qw & \qw & \targ & \qw\\
& \qw & \ctrl{-1} & \targ & \qw & \qw & \qw & \targ & \ctrl{-1} & \qw\\
& \gate{R_X(\Delta t)} & \ctrl{-2} & \ctrl{-1} & \targ & \gate{R_X(\Delta t)} & \targ & \ctrl{-1} & \ctrl{-2} & \qw\\
}
}
\vspace{1.5em}
\centerline{
\Qcircuit @C=1.2em @R=0.6em {
& \multigate{2}{P} & \gate{R_Z(\frac{W \Delta t}{2})} & \multigate{2}{P^\dagger} & \qw \\
& \ghost{P} & \gate{R_Z(\frac{W \Delta t}{4})} & \ghost{P^\dagger} & \qw \\
& \ghost{P} & \gate{R_Z(\frac{W \Delta t}{8})} & \ghost{P^\dagger} & \qw \\
}
}
\caption{Quantum circuits for the simulation of the Anderson model. The time-evolution is trotterized: on top is the kinetic piece, and on bottom is the evolution under a random Hamiltonian. For brevity, $R_X(\theta)$ and $R_Z(\theta)$ denote $e^{i \theta \sigma_x}$ and $e^{i \theta \sigma_z}$, respectively. The construction of a pseudo-random permutation operator $P$ is discussed in the text.\label{fig:anderson-circuits}}
\end{figure}

Whether a system is localized or not can be detected by placing a particle on the lattice (at the origin, say), allowing it to diffuse for a long time, and then observing how close to the origin it remains on average. In the non-localized phase, the particle will diffuse out to infinity; in the localized phase, it will remain within some finite distance. When simulating the Anderson transition at finite volume on a periodic lattice, the degree of localization of a wavefunction can be measured by
\begin{equation}\label{eq:coconut}
D = \frac{L}{\pi\sqrt{2}}\sqrt{1 - \left<\cos \frac{2\pi \hat x}{L}\right>}
\text,
\end{equation}
which, in the large volume limit, yields the average distance. This serves as an order parameter for the Anderson transition. On a three-dimensional lattice, this order parameter as a function of disorder $W$ is seen in Figure~\ref{fig:anderson3d}, with the expected transition being visible even at such small volumes, near $W \sim 15$.

\subsubsection{A Pseudo-Random Number Generator}\label{sec:prng}

We now discuss the construction of a suitable (seekable) pseudo-random number generator. A particularly straightforward construction of a seekable PRNG uses a cryptographic hash function like {\tt SHA256}~\cite{sha256ref}. A sequence of $K$ PRNGs indexed by seeds $k$ are constructed via $f_k(i) = {\tt SHA256}(i*K+k)$. This construction is validated and known to perform well in practice~\cite{salmon:2011}. Unfortunately, near-term quantum computers do not have enough qubits available to compute modern cryptographic hash functions, which operate on fixed-size registers of hundreds of bits.

An alternative approach is to construct the function $f(i)$ from a random reversible classical circuit (which therefore implements a random permutation matrix).

A random permutation matrix $P$ defines a seekable PRNG via $P\left|i\right> = \left|f(i)\right>$. Although the reversibility of the operation implies weak correlations between the different values of $f(i)$, these correlations are unmeasurable in the large volume limit and can be neglected.

We construct a random circuit by appending a fixed-length sequence of not Toffoli gates, acting on random argumnts. With the aid of $Q-3$ ancillary qbits, any permutation matrix can be obtained in this way~\cite{Toffoli:1980}. In the limit of a large number of gates, this samples uniformly from the distribution of permutation matrices.

This construction of a PRNG was validated against the \texttt{dieharder}~\cite{diehard,dieharder} battery of statistical tests in~\cite{Alexandru:2019dmv}. For circuits acting on $30$ bits, a sequence of $600$ random gates was sufficient to consistently pass all statistical tests in the battery, and the number of gates required was found to scale polynomially with the number of bits (and therefore polylogarithmically with the volume).

This construction is further validated in Figure~\ref{fig:anderson3d}, where it is shown that the same physical results are obtained in a simulation using this construction as in a simulation using a conventional (sequential) PRNG.

\begin{figure}
\centering
\includegraphics{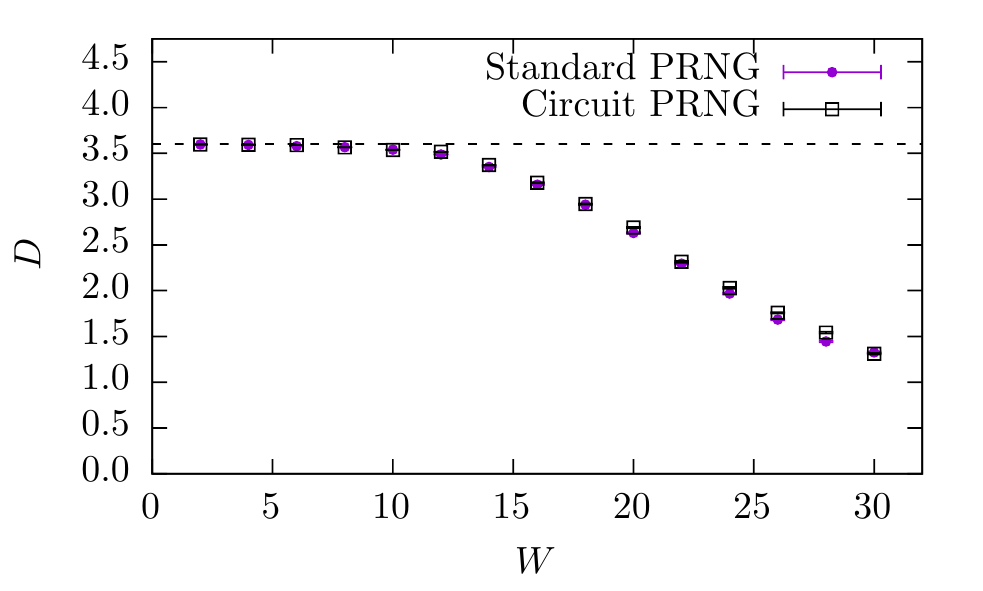}
\caption{Figure from~\cite{Alexandru:2019dmv}, showing average distance as defined by (\ref{eq:coconut}), as a function of the disorder parameter $W$, for a $16^3$ lattice, in the limit of long time evolution. The blue triangles show results obtained with a conventional PRNG, and the red squares give results obtained with the circuit-based PRNG described in Section~\ref{sec:prng}. The dashed line shows the delocalized limit for $D$.\label{fig:anderson3d}}
\end{figure}

\subsubsection{Demonstration}

The method described above has the nice property that even very small numbers of qubits can be used in a sensible simulation, albeit on a small lattice. The field theory simulations we will look at later don't have this property, as representing even one link can require many qubits ($11$, in the main example). Furthermore, because the size of the lattice is exponential in the number of qubits (in contrast to the scaling for simulating a field theory), even near-term quantum computers can simulate lattice sizes at which localization can be seen.

Figure~\ref{fig:anderson-demo} demonstrates the algorithm on two physical quantum processors, one provided by IBM (programmed with \texttt{qiskit}~\cite{aleksandrowicz2019qiskit}) and the other by Rigetti (programmed with \texttt{quil}~\cite{smith2016practical}). The simulation is done on two qubits, and therefore involves four lattice sites. There is no sensible notion of a random potential on four sites; we fix the potential to be $V(0,1,2,3) = 0,1,3,2$, which is computed with a single $CX$ gate. On each processor, for each point, we perform $300$ quantum measurements to estimate $D(t)$.

\begin{figure}
\centering
\includegraphics[width=4.2in]{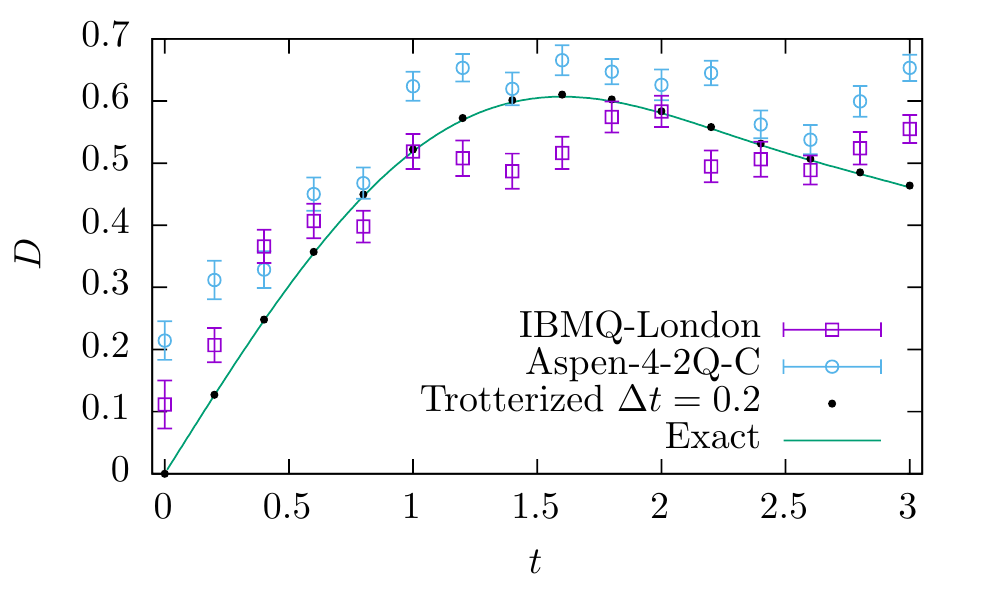}
\caption{Figure from~\cite{Alexandru:2019dmv}, showing average distance as measured by (\ref{eq:coconut}), as a function of evolution time for a $4$-site lattice with disorder parameter $W=5$ and Trotterization step size $\Delta t = 0.2$.\label{fig:anderson-demo}}
\end{figure}

\section{Simulating a Field Theory}

Now we turn to the problem of simulating a field theory on a quantum compute. For concreteness, let us focus on the Heisenberg spin chain, with Hamiltonian
\begin{equation}\label{eq:heisenberg-hamiltonian}
H =
-J
\sum_{\left<ij\right>} \sigma_z(i) \sigma_z(j)
- \mu
\sum_i \sigma_x(i)
\end{equation}
where the first sum runs over all pairs of adjacent sites on a one-dimensional lattice, J is the ferromagnetic coupling, and $\mu$ is a magnetic field. Because the local degrees of freedom posess a two-dimensional Hilbert space, this model is particularly amenable to qubit-based quantum simulation. Nevertheless, the features of the simulation of this model are essentially the same as those of other field theories; moreover, many models can be rewritten (at least approximately) as spin chains, such as $\sigma$ models~\cite{Alexandru:2019ozf}. The generalizations to scalar $\phi^4$ field theory and fermionic fields are discussed in Sections~\ref{sec:scalar-simulation} and \ref{sec:fermion-simulation}, respectively.

The first step of preparing a quantum simulation is mapping the physical Hilbert space to that of the quantum computer. For the spin chain, the mapping is trivial: each site of the spin chain corresponds to a single qubit.

The next step is to select a Suzuki-Trotter decomposition of the Hamiltonian. A natural choice for (\ref{eq:heisenberg-hamiltonian}) is
\begin{equation}
e^{-i H } \approx
\left(
e^{i (\Delta t) \mu \sum \sigma_x}
e^{i (\Delta t) J \sum \sigma_x \sigma_x}
\right)^{t / \Delta t}
\text.
\end{equation}
A more general Heisenberg Hamiltonian is possible, with arbitrary couplings and magnetic fields along all three axes:
\begin{equation}
H = - \sum_{\left<ij\right>}
\left[
J_x \sigma_x(i) \sigma_x(j) +
J_y \sigma_y(i) \sigma_y(j) +
J_z \sigma_z(i) \sigma_z(j)
\right]
-
\mu \cdot \sum_i \sigma(i)
\text.
\end{equation}
In this case a natural Trotterization has three factors, diagonal in the $x$, $y$, and $z$-bases individually.

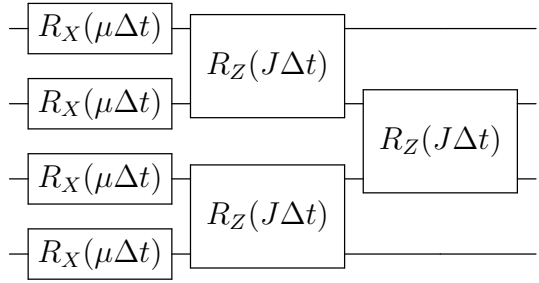
\begin{figure}
\[
\Qcircuit @C=0.6em @R=0.8em {
& \gate{R_X(\mu \Delta t)} & \multigate{1}{R_Z(J \Delta t)} & \qw & \qw\\
& \gate{R_X(\mu \Delta t)} & \ghost{R_Z(J\Delta t)} & \multigate{1}{R_Z(J \Delta t)} & \qw\\
& \gate{R_X(\mu \Delta t)} & \multigate{1}{R_Z(J\Delta t)} & \ghost{R_Z(J \Delta t)} & \qw\\
& \gate{R_X(\mu \Delta t)} & \ghost{R_Z(J \Delta t)} & \qw & \qw\\
}
\]
\caption{Quantum circuit for a single Trotter step of time-evolution for the Heisenberg spin chain with Hamiltonian (\ref{eq:heisenberg-hamiltonian})\label{fig:heisenberg-circuit}, on four lattice sites.}
\end{figure}

Because the mapping between the physical Hilbert space and that of the quantum processor is so clean, the resulting time-evolution circuits are particularly simple. Evolution for a single Trotterization step is shown in Figure~\ref{fig:heisenberg-circuit}. Note that the order in which the coupling terms are applied is irrelevant, as they all commute.

This completes the implementation of time-evolution on the quantum computer. Two important steps remain: the preparation of a physically relevant initial state, and measurement at the end of some time evolution. Typically the initial state of greatest interest is the ground state of the Hamiltonian, or the lowest-lying state constrained to have some quantum numbers. Preparation of such states is a major area of study, and some strategies are discussed in detail in Sections~\ref{sec:preparation} and \ref{sec:pim}.

Assuming that the ground state $\left|\Omega\right>$ has been prepared, there are a wide selection of physically interesting observables readily accessible. By allowing the magnetic field to be time-dependent (and even space-dependent, at no cost to circuit complexity), the response of the system to time-varying magnetic fields can be measured.

\subsection{Linear Response, Two Ways}\label{sec:response}
A frequently relevant subset of these observables are those obtained by considering the limit of a weak perturbation. In the case of a delta-function perturbation described by the Hamiltonian
\begin{equation}
H_\epsilon(t)
=
H_0 + \epsilon \delta(t) H'
\text,
\end{equation}
the expectation value of an operator $\mathcal O$ at time $T$ after the perturbation hits is given, to leading order, by
\begin{equation}
\left<\mathcal O(t)\right>
=
\left<\mathcal O(0)\right>
+
\epsilon \left<\left[H', \mathcal O(t)\right]\right>
+
O(\epsilon^2)
\text.
\end{equation}
Thus we see that \emph{linear} response is governed by time-dependent correlation functions evaluated in the original Hamiltonian.

This discussion can be re-cast as a method for evaluating time-dependent correlation functions. Evolution under a time-dependent Hamiltonian is easily achieved on a quantum computer. We can evaluate, therefore, the expectation value $\left<\mathcal O(t)\right>$ under no perturbation and under a small perturbation, and finite-differencing yields an approximation to the expectation value $\left<[H', \mathcal O(t)]\right>$.

An alternative procedure~\cite{PhysRevLett.113.020505}, not limited to the commuator of Hermitian operators, is available with the use of an ancillary qubit. To begin with, we describe a procedure for measuring the expectation value of an arbitrary unitary operator~\cite{Ortiz:2000gc}. Let $U$ be a unitary operator, for which we would like to measure the expectation value $\left<\Psi\right|U\left|\Psi\right>$ in some state $\left|\Psi\right>$. Define $U_C$ to be the controlled-$U$ unitary acting on the combination of the original system with the one ancillary qubit. Thus $U_C$ is defined as \begin{equation}
U_C\left|\Psi\right>\left|0\right>
= \left|\Psi\right>\left|0\right>\;\text{ and }\;
U_C\left|\Psi\right>\left|1\right>
= \left(U \left|\Psi\right>\right)\left|1\right>
\text.
\end{equation}
A circuit implementing this unitary can be obtained using the technique described in Section~\ref{sec:cfromc}. We now begin by applying a hadamard gate to the ancillary qubit, then apply $U_C$, and finally measure the controlled qubit. Depending on the basis chosen to measure the controlled qubit, obtain either the real or imaginary part of the desired expectation value.

An obvious application is to the Heisenberg spin chain. The time-separated correlator $\sigma_x(t) \sigma_x(0) = e^{i H t} \sigma_x e^{-i H t} \sigma_x$ is unitary, and therefore can be measured directly. However, the method is slightly more general: the ability to measure an arbitrary unitary allows us also to measure any operator which can be decomposed as a sum of unitaries.

\subsection{Measuring Masses}\label{sec:measuring-mass}
Measuring energies on a quantum computer may be accomplished via the algorithm of quantum phase estimation~\cite{cleve1998quantum} (QPE). Given a unitary operator $U$ (implemented via quantum circuits) and a prepared eigenstate $|\Psi\rangle$, QPE is a procedure for estimating the phase $\theta$ of the eigenvalue $e^{i\theta}$ of the prepared state. When the unitary operator is time-evolution, this phase is of course the energy of the state. The simplest form of quantum phase estimation proceeds by introducing an ancillary qubit in the state $|0\rangle + |1\rangle$, and performing controlled evolution under $U$. After this evolution, the state of the system is $(|0\rangle + e^{i\theta}|1\rangle)|\Psi\rangle$, and $\theta$ may be estimated modulo $pi$ via repeated measurements in the $Z$-basis. With a little more sophistication, the binary representation of $\theta$ can be determined with a single measurement, to precision $1/\epsilon$, with $\log (1/\epsilon)$ ancilla.

We would like to measure the mass of a hadron --- that is, the difference in energies between the vacuum and the lowest-lying state with quantum numbers of that hadron.Assume we have the ability to prepare both the ground state of the lattice theory and the ground state of the sector with quantum numbers of some hadron. The most straightforward procedure to obtain the mass of that hadron is to first prepare the ground state $|\Omega\rangle$, and then measure via QPE the energy $E_\Omega$ of this state. On a lattice of the same parameters, we may prepare the ground state $|P\rangle$ of the hadron, and similarly measure the energy $E_P$ of that state. The mass is then given by $E_P - E_\Omega$. This method is simple, but suffers from a significant flaw, related to the fact that $E_\Omega$ and $E_P$ are not sensible physical quantities. The vacuum energy is divergent in both the continuum and infinite-volume limits. Therefore, as these limits are approached, both energies must be measured with increasing precision to resolve the cancellation, before any information about the mass is obtained. This is another signal-to-noise problem.

This signal-to-noise problem can be done away with by preparing two lattices at once on the same quantum processor. These lattices are uncoupled, and we prepare in the first the ground state and in the second the hadron state, so that the quantum processor is in the state $|\Omega\rangle \otimes |P\rangle$. We now consider the unitary operator $U(t) = e^{-i H t} \otimes e^{i H t}$. The prepared state is an eigenstate of this operator; moreover, the divergent part of the energes cancel. Thus, QPE applied to $U$ directly yields the hadron mass, with no need to resolve fine cancellations.

\subsection{Scalar Fields}\label{sec:scalar-simulation}

Lattice scalar field theory was examined as a target for quantum simulation in~\cite{Jordan:2011ci}. For a lattice scalar field theory described by the Hamiltonian (\ref{eq:h-scalar-interacting}), each lattice site is associated to an anharmonic oscillator, and the oscillators are coupled by the term $(\phi_x - \phi_y)^2$. This system presents a new difficulty for quantum simulations, which is in fact characteristic of most bosonic field thoeries: the local Hilbert space is of infinite dimension, and there is therefore no isomorphic Hilbert space that can be created with a finite number of qubits.

The only solution is to truncate the physical Hilbert space. To maintain the locality of the theory, it is convenient to truncate the local Hilbert spaces independently, so that the full Hilbert space remains a tensor product, but now of finite-dimensional systems.

A reasonable truncation for this system becomes apparent once the Hamiltonian is rewritten in terms of creation and annihilation operators:
\begin{equation}
H = \sum_{\left<xy\right>}
\left(
a_x^\dagger a_y + \mathrm{h.c.}
\right)
+
\sum_x 
\left[
a_x^\dagger a_x
+ \lambda \left(a_x^\dagger + a_x\right)^4
\right]
\text.
\end{equation}
Truncating each local Hilbert space to the lowest $K$ eigenstates of the harmonic oscillator Hamiltonian $a^\dagger a$, the number of qubits required for the full simulation scales with $V \log K$. When performing computations, there are now three extrapolations that must be performed: removing the truncation ($K\rightarrow\infty$), removing the lattice cutoff ($a\rightarrow 0$), and the infinite volume limit ($V \rightarrow \infty$), in that order.

\subsection{Fermions}\label{sec:fermion-simulation}

Fermionic lattice theories are limited to a finite Hilbert space by the Pauli exclusion principle. Each fermionic degree of freedom (there may be many per lattice site, due to spin, flavor, and internal symmetries) is associated to a two-dimensional Hilbert space, preparing a convenient mapping between the physical space and that of the processor.

A difficulty arises, however, when attempting to map the operators used to define the physical Hamiltonian, to operators defined on the qubits of the quantum computer. The fundamental operators of a fermionic theory are the raising and lowering operators $a^\dagger$ and $a$, defined to anticommute:
\begin{equation}
\{a_i, a_j\} = 0\;\text{ and }\;
\{a^\dagger_i, a_j\} = \delta_{ij}
\text.
\end{equation}
The natural raising and lowering operators defined on qubits, given by $\sigma^{\pm}(i) = \sigma_x(i) \pm i \sigma_y(i)$, commute at different qubits. Anticommuting operators must be constructed on the quantum computer in order to map the physical Hamiltonian over to the qubits.

A suitable appropriate mapping of operators is provided by the Jordan-Wigner transormation~\cite{Ortiz:2000gc,Jordan:1928wi}, in which $a_i$ maps to
\begin{equation}
\left[\bigotimes_{k=1}^{i-1} \sigma_z\right]
\otimes
\sigma^-
\otimes
I
\text.
\end{equation}
Two such operators are readily seen to anticommute, and so a fermionic Hamiltonian can be rewritten in terms of them.

The Jordan-Wigner transformation does violence to locality\footnote{With the exception of one-dimensional lattices without periodic boundary conditions, where an Jordan-Wigner transformation can be constructed such that local fermion bilinears map to local spin operators.}. The string of $\sigma_z$ operators imposes an extra cost, typically polynomial in the volume ($V^{2/3}$ for a three-dimensional lattice). The transformation can be improved to alleviate or remove this asymptotic cost. The first such improvement was due to Bravyi and Kitaev, and reduced this to a logarithm of the volume of the lattice~\cite{bravyi2002fermionic}. More recent improvements result in a constant overhead~\cite{whitfield2016local,verstraete2005mapping}, removing all asymptotic penalty.

\section{Simulating a Gauge Theory}

The simplest gauge theory to simulate is the $\mathbb Z_2$ gauge theory of Section~\ref{sec:hamiltonian-gauge}, with Hamiltonian (\ref{eq:z2-hamiltonian}). This is only a small modification from the spin system considered above. Each degree of freedom is now associated to a link on the lattice, instead of a site, and the coupling term diagonal in the $Z$ basis couples four degrees of freedom (one plaquette) rather than just two.

\begin{figure}
\[
\Qcircuit @C=0.6em @R=0.8em {
& \gate{\mathcal R_X(\Delta t)} & \ctrl{1} & \qw & \qw & \qw & \ctrl{1} & \qw\\
& \gate{\mathcal R_X(\Delta t)} & \gate{X} & \ctrl{2} & \qw & \ctrl{2} & \gate{X} & \qw \\
& \gate{\mathcal R_X(\Delta t)} & \ctrl{1} & \qw & \qw & \qw & \ctrl{1} & \qw & \\
& \gate{\mathcal R_X(\Delta t)} & \gate{X} & \gate{X} & \gate{\mathcal R_Z(\Delta t)} & \gate{X} & \gate{X} & \qw & \\
}
\]
\caption{Quantum circuit for a single Trotter step of the time-evolution of $\mathbb Z_2$ gauge theory with one plaquette.\label{fig:z2-gauge-circuit}}
\end{figure}
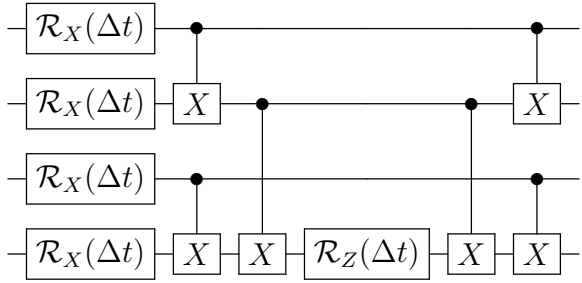

The procedure for simulating time-evolution of such a theory is consequentially simple. Figure~\ref{fig:z2-gauge-circuit} shows a single step of time evolution of a one-plaquette model. This major new feature of this procedure is the use of an ancillary qubit to compute the value of the plaquette. The qubit begins in state $\left|0\right>$, and each of the four links in the plaquette are combined via controlled not operations (corresponding to the $\mathbb Z_2$ group operation) until the ancillary qubit contains the value of the plaquette.

Generalizing to an arbitrary gauge group $G$, we introduce the concept of a $G$-register~\cite{Lamm:2019bik}: a collection of qubits whose Hilbert space is mapped to $\mathbb C G$. For continuous gauge groups, where the space $\mathbb C G$ is infinite-dimensional, this is necessarily some approximation, however we will ignore this difficulty for the time being. The Hilbert space of a $G$-register is spanned by states $\left|g\right>$, which we will take to be the computational basis. On this space a set of primitive operations are needed:
\begin{itemize}
\item An inversion gate, which takes a $G$-register and transforms, in the computational basis, by taking the inverse of the group element. The gate is defined by $\mathfrak U_{-1} \left|U \right> = \left|U^{-1}\right>$. This gate is self-adjoint.
\item A multiplication gate, which acts on two $G$-registers and transforms the second. This gate is defined by $\mathfrak U_\times \left|g\right>\left|h\right> = \left|g\right> \left|gh\right>$.
\item A trace gate, which gives each computational basis state a phase proportional to the trace of the stored group element. This gate is defined by $\mathfrak U_{\Tr}(\theta)\left|U\right> = e^{i\theta \Re \Tr U} \left|U\right>$. Note that the definition of this gate depends on the representation, and indeed, some lattice Hamiltonians may involve traces taken in multiple representations.
\item The (nonabelian) Fourier transform $\mathfrak U_F$, which transforms a $G$-register into Fourier space (a $\hat G$-register). It is defined by
\begin{equation}
\mathfrak U_F
\sum_{U \in G} f(U) \left|U\right>
=
\sum_{\rho \in \hat G}
\hat f(\rho)_{ij} \left|\rho,i,j\right>
\end{equation}
where the second sum is taken over all representations $\rho$ of $G$, and $\hat f$ denotes the Fourier transform of $f$. This gate is the operation described in Section~\ref{sec:quantum-fourier-transform}, and diagonalizes the kinetic part of the Hamiltonian.
\item The Laplace-Beltrami gate $\mathfrak U_{LB}$, which acts on a $\hat G$-register and gives each state a diagonal phase, which is a function of the representation alone (not the indices $i,j$).
\end{itemize}
A few note about the generality of these operations are in order. First, although the multiplication gate $\mathfrak U_\times$ is defined here to perform left multiplication, a gate for right multiplication is obtained from the combination of $\mathfrak U_\times$ and $\mathfrak U_{-1}$ via the identity $\mathfrak U_{\times,R}(1,2) = \mathfrak U_{-1}(1) \mathfrak U_{-1}(2) \mathfrak U_\times(2,1) \mathfrak U_{-1}(2) \mathfrak U_\times(1,2)$.

From these operations we can construct time-evolution for a pure gauge theory. In the presence of matter fields, a few additional ones will be needed, discussed in Section~\ref{sec:gauge-fermions} below.

\subsection{Time Evolution}

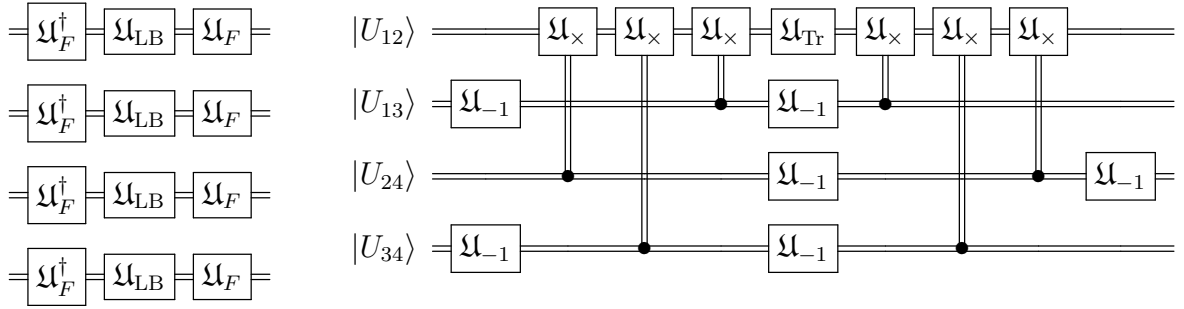
\begin{figure}
\newcommand{\gw}[1][-1]{\ar @{=} [0,#1]}
\newcommand{\gwx}[1][-1]{\ar @{=} [#1,0]}
\newcommand{\ggate}[1]{*+<.6em>{#1} \POS ="i","i"+UR;"i"+UL **\dir{-};"i"+DL **\dir{-};"i"+DR **\dir{-};"i"+UR **\dir{-},"i" \gw}
\newcommand{\gctrl}[1]{\control \gwx[#1] \gw}

\centerline{
\Qcircuit @C=0.6em @R=0.8em {
& \ggate{\mathfrak U_F^\dagger} & \ggate{\mathfrak U_\text{LB}} & \ggate{\mathfrak U_F} & \gw & \\
& \ggate{\mathfrak U_F^\dagger} & \ggate{\mathfrak U_\text{LB}} & \ggate{\mathfrak U_F} & \gw & \\
& \ggate{\mathfrak U_F^\dagger} & \ggate{\mathfrak U_\text{LB}} & \ggate{\mathfrak U_F} & \gw & \\
& \ggate{\mathfrak U_F^\dagger} & \ggate{\mathfrak U_\text{LB}} & \ggate{\mathfrak U_F} & \gw & \\
}
\hspace{0.7in}
\Qcircuit @C=0.6em @R=0.8em {
\lstick{\left|U_{12}\right>} & \gw & \ggate{\mathfrak U_\times} & \ggate{\mathfrak U_\times} & \ggate{\mathfrak U_\times} & \ggate{\mathfrak U_{\Tr}} & \ggate{\mathfrak U_\times} & \ggate{\mathfrak U_\times} & \ggate{\mathfrak U_\times} & \gw & \gw &\\
\lstick{\left|U_{13}\right>} & \ggate{\mathfrak U_{-1}} & \gw & \gw & \gctrl{-1} & \ggate{\mathfrak U_{-1}} & \gctrl{-1} & \gw & \gw & \gw & \gw & \\
\lstick{\left|U_{24}\right>} & \gw & \gctrl{-2} & \gw & \gw & \ggate{\mathfrak U_{-1}} & \gw & \gw & \gctrl{-2} & \ggate{\mathfrak U_{-1}} & \gw & \\
\lstick{\left|U_{34}\right>} & \ggate{\mathfrak U_{-1}} & \gw & \gctrl{-3} & \gw & \ggate{\mathfrak U_{-1}} & \gw & \gctrl{-3} & \gw & \gw & \gw &\\
}
}
\caption{Circuits implementing the time-evolution of a pure-gauge lattice field theory. The first circuit implements the quantum-mechanical kinetic erm, and the second the quantum-mechanical potential term associated to a single plaquette $\Re \Tr U^\dagger_{13}U^\dagger_{34} U_{24}U_{12}$. Note that in these circuits, the primitive object is a $G$-register, denoted with a doubled line.\label{fig:nonabelian-evolution}}
\end{figure}

Time-evolution of the pure gauge theory for a general group $G$ is performed, as usual, with a Trotter-Suzuki decomposition. The kinetic piece of the evolution is diagonalized by the Fourier transform. The potential piece of the evolution requires, as for the $\mathbb Z_2$ gauge theory, the accumulation of a plaquette in a single $G$-register. This is accomplished by picking one link in the plaquette and repeatedly multiplying with all other links in the plaquette.

The resulting circuits for the propagation of a single-plaquette, nonabelian gauge theory are shown in Figure~\ref{fig:nonabelian-evolution}.

\subsection{Gauge Invariance}

The time-evolution circuits presented above can only be said to represent a gauge theory when the initial state lies in the physical subspace; i.e., is gauge-invariant. It is critical, therefore, that we are able to prepare gauge-invariant states in general.

Two gauge-invariant states are particularly easy to prepare: the strong-coupling ground state and the weak-coupling ground state. In the strong-coupling limit, the kinetic term of the Hamiltonian dominates. This term does not couple links, and the resulting ground state is a product state, particularly easy to prepare:
\begin{equation}
\left|\Omega_{\mathrm{strong}}\right>
=
\bigotimes_L
\left(
\sum_{U \in G} \left|U\right>
\right)
\text.
\end{equation}
Here each link is in an equal superposition of all group elements. The state described here is already gauge-invariant, so no further symmetrization is needed.

In the weak-coupling limit, each plaquette is forced to be the identity in the ground state\footnote{Assuming that such a configuration is permitted by boundary conditions --- otherwise, we would have a frustrated system, and no general efficient algorithm for preparing the ground state.}. Naively, then, the ground state is the product state where each link is set to $\left|I\right>$; however, this state is not gauge-invariant.
We can prepare a gauge-symmetric version of this state with the aid of one ancillary $G$-register per lattice site. These registers represent a gauge transformation, and we will denote them $V$. Begin by initializing each link to $\left|I\right>$, and each ancillary register to an equal superposition of all group elements. We may now perform the gauge transformation: the link $U_{ji}$ from site $i$ to site $j$ is multiplied on the right by $V^\dagger_i$ and on the left by $V_j$. This yields the state
\begin{equation}
\sum_V \left|(V_2 V_1^\dagger)\cdots\right>\left|V_1 V_2 \cdots\right>
\text.
\end{equation}
At this point, the physical link registers are engtangled with the (unphysical) gauge registers. The entanglement is removed by again repeatedly applying $\mathfrak U_\times$, but now with the link registers as the control and the gauge registers as the target. A privileged site $i$ is selected (the choice will not affect the final state), and for each other site $j$, a particular path from $j$ to $i$ is selected. After one multiplication for each link in this path, the ancillary register associated to site $j$ is transformed to the state $|V_i\rangle$. Repeating for all ancilla, we obtain the state
\begin{equation}
\sum_V|(V_2 V_1^\dagger)\cdots\rangle |V_i V_i \cdots\rangle\text.
\end{equation}
That this state is in fact a product state may be seen by noting that a global gauge transformation by $V_i^\dagger$ leaves the physical registers invariant, while rotating each ancilla into $|I\rangle$. The disentangled ancilla may now be discarded, and the resulting state is the gauge-projected $P\left|I\right>$.

This suffices to show that the gauge-invariant sector is efficiently accessible. Further discussion of state preparation, in particular the preparation of physical ground states and thermal states, is in Sections~\ref{sec:preparation} and \ref{sec:pim} below.

\subsection{Adding Fermions}\label{sec:gauge-fermions}

The most interesting gauge theories have matter fields coupled to the gauge degrees of freedom. We will consider here the case of fermionic matter fields, as in QCD. In order to introduce matter fields, one must first pick a representation of the group $G$ under which the matter fields are to transform\footnote{The gauge fields themselves have no choice but to transform in the adjoint representation.}. For matrix groups (e.g.\ $SU(N)$), it is most common to select the fundamental representation.

Labeling the dimension of the selected representation by $N$, and ignoring for simplicity spin and flavor, we will have $N$ independent fermionic degrees of freedom at each lattice site, creating a local Hilbert space of dimension $2^N$. Denote the creation and annihilation operators $a^\dagger_i$ and $a_i$, where $i\in 1\ldots N$ is the `color' index of the chosen representation $\rho$. These operators transform into each other under the action of the group $G$:
\begin{equation}
U^\dagger a_i U
= \sum_j \rho(U)_{ij} a_j
\text.
\end{equation}
Gauge transformations, and the projection operator $P$, now affect the fermionic modes as well.

Crucially, the $VN$ fermionic degrees of freedom can be created via the Jordan-Wigner transformation (or any other method) before making reference to the fact that they transform into each other under various symmetries. In other words, the presence of gauge-invariance does not complicate the task of creating anticommuting operators from the fundamental commuting operators of a quantum computer.

\subsection{Demonstration: $D_4$ Gauge Theory}
\begin{figure}
\centering
\begin{tikzpicture}
\draw[->] (0,0) -- node[below] {$U_0$} (2,0);
\draw[dashed,->] (0,2) -- node[below] {$U_0$} (2,2);
\draw[->] (2,0) -- node[below] {$U_1$} (4,0);
\draw[dashed,->] (2,2) -- node[below] {$U_1$} (4,2);
\draw[->] (0,0) -- node[left] {$U_2$} (0,2);
\draw[->] (2,0) -- node[left] {$U_3$} (2,2);
\draw[dashed,->] (4,0) -- node[left] {$U_2$} (4,2);
\end{tikzpicture}
\caption{The lattice geometry of the $D_4$ gauge theory simulated here. Dashed lines indicate repeated links due to the periodic boundary conditions.\label{fig:d4-geometry}}
\end{figure}
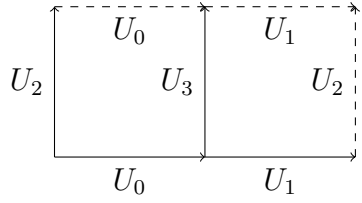

The methods of the previous sections can be demonstrated on a classically simulated quantum computer, with a small finite gauge group. The smallest two nonabelian groups are $D_3$ and $D_4$, defined as the group of isometries of the triangle and the square, respectively. Because $D_3$ has $6$ elements and $D_4$ has $8$, they each require three qubits per link to simulate, and therefore we may as well simulate $D_4$.

The first order of business is to construct a $G$-register; that is, to chose a particular isomorphism between the space of complex-valued functions on $D_4$, and the Hilbert space of three qubits. The group $D_4$ can be defined as the subgroup of $U(2)$ generated by the matrices
\begin{equation}
\left(\begin{matrix}
i & 0\\
0 & -i\\
\end{matrix}\right)
\text{ and }
\left(\begin{matrix}
0 & 1\\
1 & 0\\
\end{matrix}\right)
\text.
\end{equation}
The state $|abc\rangle$ is defined to correspond to the matrix
\begin{equation}
\left[\left(\begin{matrix}
0 & 1\\
1 & 0\\
\end{matrix}\right)\right]^a
\left[
\left(\begin{matrix}
i & 0\\
0 & -i\\
\end{matrix}\right)\right]^{2b + c}
\text.
\end{equation}
We next construct the inversion, multiplication, trace, and Fourier transform circuits. The inversion and multiplication circuits are classical circuits, easily constructed. As the only element of $D_4$ with a nonvanishing trace in this representation is is the identity, the trace circuit is a three-qubit-controlled phase gate. The circuits for these three operations and the Fourier transform are given in the appendix of~\cite{Lamm:2019bik}.

Each link requires three qubits, so a classical simulation can simulate a lattice containing four links without much trouble. (Substantially larger lattices --- up to $\sim 30$ qubits --- could be obtained with more sophisticated algorithms and reasonable computer time; however, this is pointless for the purposes of demonstration.) A two-plaquette geometry based on a four links is shown in Figure~\ref{fig:d4-geometry}. On this geometry, we initialize the gauge-invariant ground state in the weak-coupling limit, and measure the expectation value of a plaquette as a function of time. The result is shown in Figure~\ref{fig:d4-results}.

\begin{figure}
\centering
\includegraphics[width=4in]{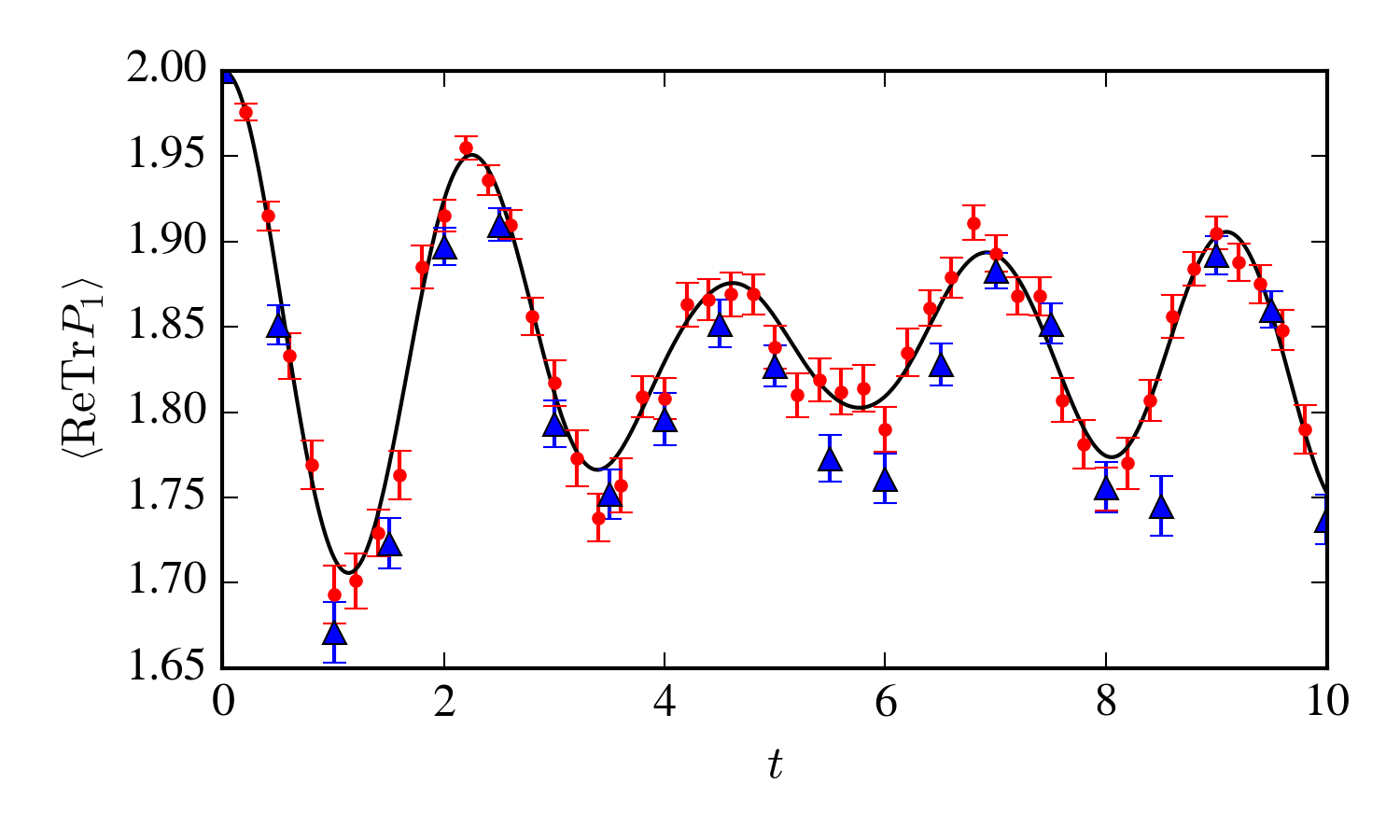}
\caption{Figure from~\cite{Lamm:2019bik}: simulation of the two-plaquette $D_4$ gauge theory. The expectation value of one plaquette as a function of time. The exact result is shown in black, with sampled data shown in red (blue) for a Trotterization time-tep of $\Delta t = 0.2$ ($\Delta t = 0.5$).\label{fig:d4-results}}
\end{figure}

\section{Simulations of QCD}

We now apply the general method of~\cite{Lamm:2019bik} to the simulation of QCD with a quantum computer. The largest difficulty encountered is the one mentioned previously: the dimension of the Hilbert space of $SU(3)$ gauge theory, even on a finite lattice, is infinite, and therefore any quantum system consisting of qubits must be only an approximation. A particular subgroup of $SU(3)$, the Valentiner group, is shown to provide an adequate approximation, and from there we can discuss how to prepare interesting states and extract partonic physics.

\subsection{The Valentiner Group}

\newcommand{\V}{\mathcal V}

The Valentiner group~\cite{valentiner1889endelige} $\V$, also referred to as $\tilde S(1080)$~\cite{flyvbjerg1985character}, is a finite subgroup of $SU(3)$ with $1080$ elements. It is not the largest finite subgroup: $\mathbb Z_n$, for instance, is a subgroup for any $n$, due to the presence of $U(1) < SU(3)$. It is, however, the largest \emph{exceptional} subgroup, that is, the largest subgroup that does not fall into one of a small number of infinite families.

The Valentiner group is particularly suitable as an approximation to $SU(3)$ because it tiles the surface of $SU(3)$ evenly. In particular, the Voronoi diagram of $\V \subset SU(3)$ has two special properties:
\begin{itemize}
\item Each region is isomorphic to every other region. This follows from the fact that $\V$, being a subgroup of $SU(3)$, is also a symmetry group of the Voronoi diagram itself.
\item Each face --- the boundary between two regions --- is isomorphic to every other face, and in fact every face is the same distance from the center of the neighboring regions.
\end{itemize}
The group $\V$ thus serves as a good approximation to $SU(3)$ in much the same sense as a dodecahedron might serve as an approximation to the sphere\footnote{In fact the platonic solids correspond directly to the nice approximations of $SU(2)$.}.

\begin{figure}
\centering
\includegraphics[width=4.2in]{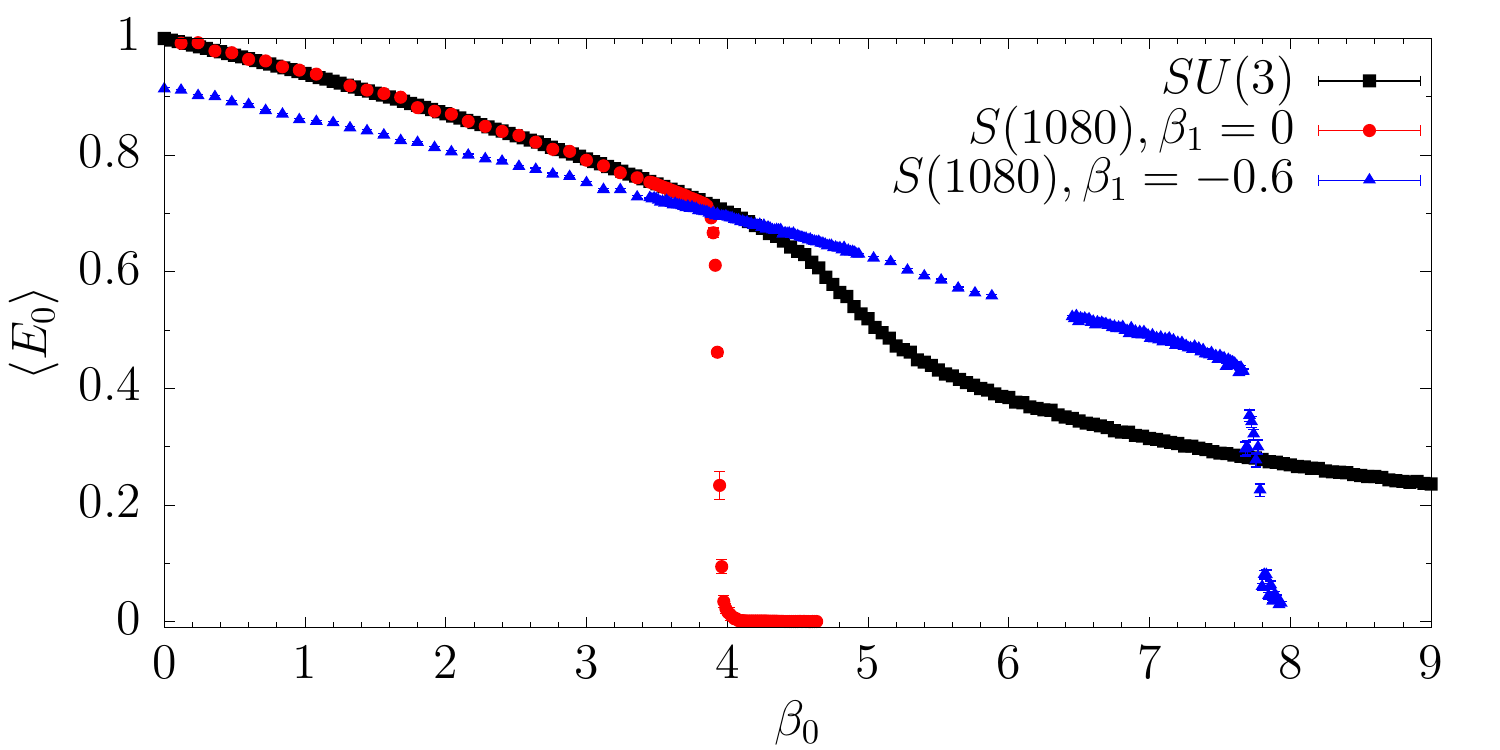}
\caption{Figure from~\cite{Alexandru:2019nsa} of the average energy per plaquette vs $\beta_0$ for $SU(3)$, $\V$ with $\beta_1 = 0$, and $\V$ with $\beta_1 = -0.6$.\label{fig:valentiner-energy}}
\end{figure}

The degree to which $\V$ is a good approximation for $SU(3)$ is a matter for empirical study, first investigated in~\cite{Bhanot:1981xp}. Critically, this question can be addressed by performing classical Monte Carlo calculations (in imaginary time), without the need for a quantum processor~\cite{Alexandru:2019nsa}. We will work with the lattice action
\begin{equation}
S = -\sum_P \left(\frac{\beta_0}{3} \Re \Tr P + \beta_1 \Re \Tr P^2\right)
\text,
\end{equation}
where the partition function is defined over all gauge configurations with links taking values in $\V$, rather than the $SU(3)$. With $\beta_1 = 0$, this is the standard Wilson action.

This lattice theory is compared to the Wilson gauge theory of $SU(3)$ in Figure~\ref{fig:valentiner-energy}, using the average energy per plaquette as a probe. The two are qualitatively different. The $SU(3)$ theory posesses a crossover from strong coupling (at small $\beta_0$) to weak coupling around $\beta_0 \sim 5$. For the $\beta_0 = 0$ version of the Valentiner gauge theory, a first-order transition is encountered before this crossover. Smaller average energy is a proxy for a closer approach to continuum physics; failing to reach the so-called scaling regime indicates that the Valentiner theory is far from the $SU(3)$ continuum theory.

The breakdown of the Valentiner gauge theory as we approach the continuum limit can be understood by noting that, in $SU(3)$ gauge theory, the fluctuations of the gauge fields about the identity become smaller as $\beta_0$ is increased. In the theory of $S(1080)$, there is limit past which the fluctuations cannot get any smaller; there are no group elements arbitrarily close to the identity. Beyond a critical value $\beta_0 \sim 4$, the fields become fixed.

The situation can be improved by noting that for the $SU(3)$ theory, the same continuum limit is expected to be reached as $\beta_0\rightarrow \infty$ regardless of the value of $\beta_1$. We see in Figure~\ref{fig:valentiner-energy} that, increasing $\beta_0$ along this trajectory, lower values of the average energy are achieved, suggesting a closer approach to the continuum physics of $SU(3)$.

So far this is just a heuristic argument that the gauge theory of $\V$ can approximate the $SU(3)$ theory, if an appropriate trajectory in $(\beta_0,\beta_1)$ is selected. We can make this more rigorous by performing, on each theory separately, a measurement at multiple points along the continuum trajectory. Then, we extrapolate (again in each theory separately) to the continuum, and compare. Note that we perform a continuum extrapolation on the $\V$ theory despite the fact that, due to the first-order transition, we know that the $\V$ theory cannot reach the continuum at all. The process of extrapolation, however, cannot see the (nonanalytic) first-order transition, and ends up extrapolating to the continuum $SU(3)$ result.

Both the $SU(3)$ and $S(1080)$ gauge theories posess a phase transition in the temperature (that is, the temporal extent of the Euclidean lattice). The chosen measurement is the dimensionless ratio $T_c\sqrt{t_0}$, where $T_c$ is the temperature of this phase transition, and $t_0$ is a scale set by the Wilson flow~\cite{Luscher:2010iy}. Figure~\ref{fig:valentiner-continuum} shows the continuum extrapolation of $T_c \sqrt{t_{0}}$ for the Valentiner theory with a trajectory defined by
\begin{equation}
\beta_1 = -0.1267\beta_0 + 0.253
\end{equation}
against the measurements of~\cite{Francis:2015lha,Kitazawa:2016dsl}. The Valentiner group serves as a good approximation to the $SU(3)$ theory for this low-energy observable.

\begin{figure}
\centering
\includegraphics[width=4.2in]{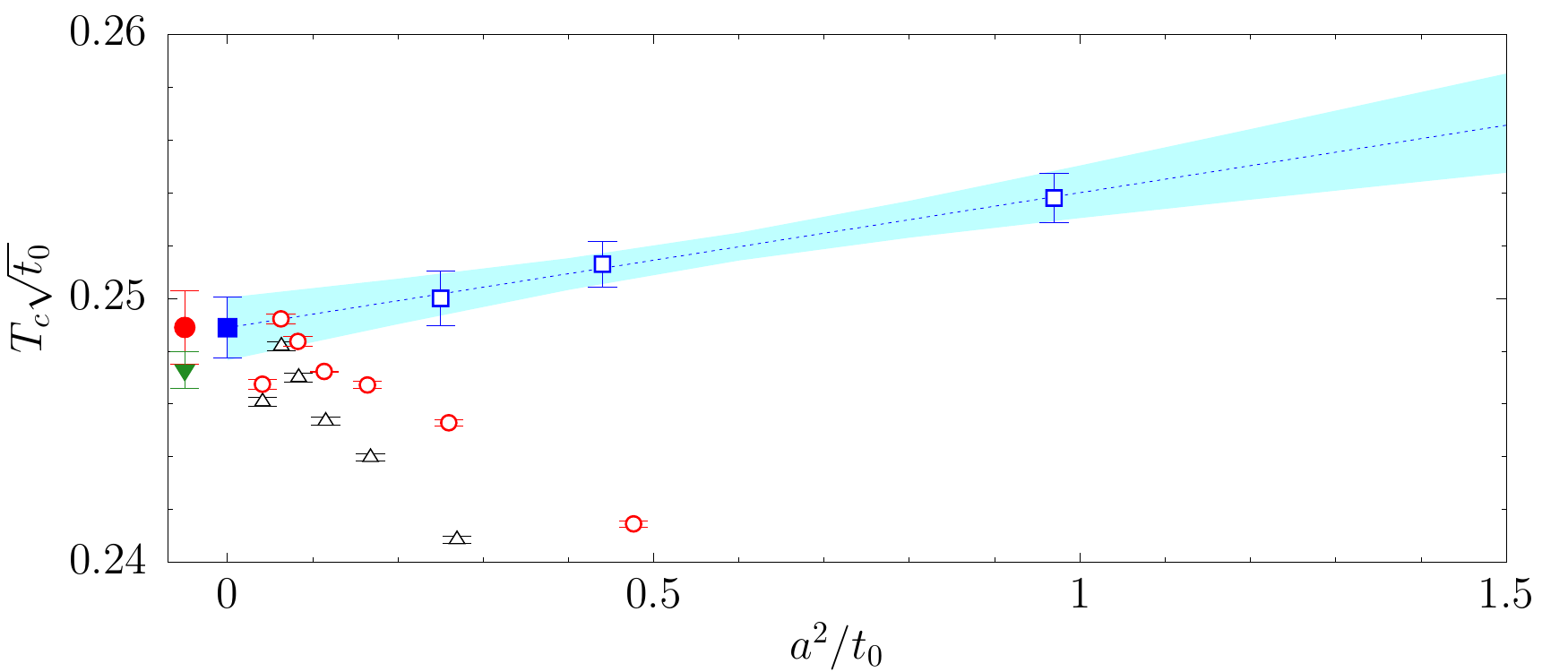}
\caption{Continuum extrapolations of $T_c \sqrt{t_0}$ for $SU(3)$ and $S(1080)$ gauge theories. The $S(1080)$ result is shown in blue, with the extrapolated value marked by the solid shape; the $SU(3)$ extrapolated results are also indicated by solid shapes.\label{fig:valentiner-continuum}}
\end{figure}

The task of comparing other low-energy observables, particularly spectroscopy, is not complete as this is written. It must be remembered that the use of $\V$ as an approximation to $SU(3)$ carries a serious drawback: because $\V$ is the largest `nice' subgroup, this is not a systematically improvable approximation. The only hope for improvement is to add more terms to the action, and seek a trajectory that provides a better approximation to continuum physics.

\subsection{State Preparation}\label{sec:preparation}

The easiest states to prepare on a quantum computer are the eigenstates of field operators. The infinite-temperature thermal state is also readily accessible. Unfortunately, these states are not of much physical interest. We would like to prepare the ground state of our Hamiltonian, perhaps restricted to some sector. For example, the ground state of the baryon-number-$1$ sector of QCD contains a single proton.

A large body of literature exists with many proposed methods for this task; see~\cite{Kaplan:2017ccd,Lamm:2018siq,Jordan:2011ci,Jordan:2014tma,jordan2018bqp} for a very incomplete sample. Quantum computers have not yet been constructed that are large enough to investigate the practical performance of these methods, and formal analysis of most preparation methods is not available. One exception is adiabatic state preparation~\cite{Jordan:2011ci,Jordan:2014tma,jordan2018bqp}. This method is backed by the adiabatic theorem~\cite{messiah1962quantum}, allowing us to make crude estimates of the costs of this method as applied to physical theories. Here we will discuss in detail the cost of preparing the ground state of the baryon-number-$1$ sector, i.e.\ the proton. This does not imply that adiabatic state preparation is the most efficient or most practical method, merely that it is the only one for which \emph{a priori} analysis is currently possible.

The adiabatic theorem~\cite{messiah1962quantum} constrains the behavior of a quantum system with a time-dependent Hamiltonian. Let $H(t)$ be our Hamiltonian, and let $\Omega$ be an eigenstate of $H(0)$, prepared at time $t=0$. When the time-dependence of the Hamiltonian is mild compared to the size $\Delta$ of the gap $\Delta(t)$ between $\left|\Omega\right>$ and the nearest eigenstate (when $\dot H / \Delta(t)^2 \ll 1$ for all $t$), the adiabatic theorem states that the system will never transition to a different eigenstate\footnote{This notion is well-defined as long as $\Delta > 0$, which is of course necessary for the theorem to make any statement at all.}. In particular, if $\left|\Omega(0)\right>$ is the ground state of $H(0)$, then $\Omega(t)$ will be the ground state of $H(t)$ for all $t$.

Adiabatic state preparation exploits the adiabatic theorem by beginning with a well-understood Hamiltonian $H(0)$, and slowly deforming along some trajectory to the desired Hamiltonian $H(T)$. As long as the evolution is performed slowly enough, the ground state is prepared with high probability. The computational cost of this method is dominated by the smallest gap along the trajectory: we must spend time proportional to $\Delta^{-2}$ in preparing the state.

It is useful to note that, although phrased as a method for finding the ground state of a Hamiltonian, the adiabatic state preparation method allows access to ground states of restricted sectors of the Hamiltonian as well. As an example, as long as all Hamiltonians along the adiabatic trajectory are translation-invariant, the total momentum will commute with each time-evolution operator $e^{-i H t}$. If the initially-prepared state has some non-vanishing momentum, then the adiabatically prepared state will preserve that momentum. The same holds for any conserved quantum numbers, as long as the conservation law holds at each point on the trajectory.

We now consider the application of adiabatic state preparation~\cite{Lamm:2019uyc} to the proton state $\left|P\right>$. We chose the adiabatic trajectory to begin with the zero-coupling theory, and the coupling is slowly increased until the desired lattice spacing is reached. The ground state of the baryon-number-$1$ sector of the free theory is thre zero-momentum fermions in a box, with the gauge fields in a Gaussian configuration (for the $SU(3)$ theory) or pegged to the identity (for the Valentiner gauge theory); these states are easily prepared. At the end of the trajectory, in the physical regime, the gap is equal to the pion mass $m_\pi$ and is relatively large. The other end is more problematic. At vanishing coupling, the `proton' fills the lattice, and the lowest-lying excited states are the back-to-back low-momentum states of any two fermions. (Massless glue excitations, which do not exist at all in the Valentiner theory, can be removed with appropriate boundary conditions). The gap in this regime is proportional to $L^{-1}$, implying that $L^2$ time steps are needed. Each evolution step requires $O(V)$ operations, so the effort required to prepare the proton is $O(L^5)$.

\subsection{Hadronic Tensor}

The hadronic tensor characterizes the response of a hadron to a perturbative probe. In the context of an electromagnetic probe, it is defined by
\begin{equation}
W^{\mu\nu}(q)
=
\Re
\int
\d^3 x
e^{i q x}
\left<P\right|
T\left\{
J^\mu(x) J^\nu(0)
\right\}
\left|P\right>
\text,
\end{equation}
where $J^\mu = \bar\psi \gamma^\mu \psi$ is the current associated to electromagnetic charge. As an example of the use of this object, the inclusive cross section of an electron scattering off of a hadron is given at leading order by~\cite{PhysRevD.98.030001}
\begin{equation}
\frac{\d^2\sigma}{\d x \d y} = \frac{\alpha^2 y}{Q^4}L_{\mu\nu}W^{\mu\nu}
\end{equation}
where the leptonic tensor $L_{\mu\nu}$, calculable in QED perturbation theory, is:
\begin{equation}
2 ( k_\mu k'_\nu + k_\nu k'_\mu - g_{\mu\nu} k \cdot k' )
\text.
\end{equation}
In both equations above, $Q^2 = -q^2$, $x = Q^2/2P\cdot q$, $y = P\cdot q / P\cdot k$, and $k' = k - q$.
Note that the hadronic tensor itself captures the nonperturbative physics of the proton; the perturbative expension in the QED coupling appears only because the interaction between the proton and electron is treated perturbatively.

The hadronic tensor (and its various limits, particularly the parton distribution function discussed below) has been the target of extensive Euclidean lattice calculations~\cite{Liu:1993cv,Aglietti:1998mz,Liu:1999ak,Ji:2013dva,Chambers:2017dov,Bali:2018spj,Radyushkin:2017cyf,Orginos:2017kos,Karpie:2018zaz,Karpie:2019eiq,Cichy:2019ebf,Alexandrou:2019lfo,Sufian:2019bol,Ma:2014jla,Ma:2014jga,Ma:2017pxb,Liang:2019frk}. Ultimately, all approaches suffer from difficulties stemming from the need to pass from a calculation performed in imaginary time to a quantity defined in real Minkowski time. However, the fact that the hadronic tensor is defined in Minkowski time makes it particularly amenable to computation with a quantum computer~\cite{Lamm:2019uyc}. Indeed, the method of Section~\ref{sec:response} may be directly applied to compute the Hadronic tensor, using the perturbed Hamiltonian
\begin{equation}
H_x(t)
=
H_0
+
\epsilon_x J^\mu(x)
\end{equation}
and measuring the observable $J_\mu(x)$ after some time evolution.

\subsubsection{Parton Distribution Function}

The response of a hadron to a probe with large momentum transfer is characterized by an apparently simpler object, the parton distribution function (PDF):
\begin{equation}
f(x) = \int \d y;
e^{i x P^+ y}
\left<P\right|\bar\psi(y) \gamma^+ W \psi(0) \left|P\right>
\end{equation}
where $u^+ = \frac{1}{\sqrt{2}} (u^0 + u^1)$ denotes the lightcone component of a vector $u$, and $W$ is a lightlike Wilson line between the origin and $y$. The PDF can be extracted from the Hadronic tensor, but one may consider calculating it directly. A procedure for directly computing the PDF is given in~\cite{Lamm:2019uyc}. Here we will only discuss why this is a bad idea.

Momentarily ignoring gauge fields, note that one of the tools of Section~\ref{sec:response} is inapplicable to the correlation function in the integrand. Because $\psi$ is not Hermitian, the desired correlation function cannot be expressed in terms of linear response. We must instead decompose the operator $\bar\psi(y) \gamma^+ \psi(0)$ as a linear combination of unitaries, each to be measured individually. This is a minor inconvenience and a major inelegance, but not fatal.

For a gauge theory, the situation is dire. The operator $\psi(0)$ is not gauge-invariant in isolation. The lightlike Wilson like involves operators at many points in space and many points in time. The resulting algorithm would require one order of finite differencing for every lattice link included in the Wilson line. This is not a practical routine even in the absence of statistical noise from the quantum computer. The hadronic tensor is protected from these issues precisely because it is a correlator of physical operators --- operators which can be coupled to external sources in the Hamiltonian.

\section{Avoiding State Preparation}\label{sec:pim}
In terms of number of quantum gates required (that is, the algorithmic time complexity), the study of QCD with a quantum computer is dominated by the process of preparing a suitable ground state. It may be practical, particularly for near-term quantum computers, to avoid doing so by coupling a classical Euclidean lattice calculation to a quantum computer~\cite{Lamm:2018siq,Harmalkar:2020mpd}.

We will consider a physical system that begins in the thermal state $\rho = e^{-\beta H_0}$ of an initial Hamiltonian $H_0$, and evolves for some time $t$ under a different Hamiltonian $H$, at which points we measure an observable $\mathcal O$. Note that this is not limited to linear response. The expectation value desired is given by
\begin{equation}
\left<\mathcal O(t)\right>
=\frac
{\sum_{i,j}
\rho_{ji}
\mathcal O(t)_{ij}
}{\sum_{i}\rho_{ii}}
=\left(
\frac
{\sum_{i,j}
\rho_{ji}
\mathcal O(t)_{ij}
}
{\sum_{i,j}\rho_{ij}}
\right)
\left(
\frac
{\sum_i \rho_{ii}}
{\sum_{i,j} \rho_{ij}}
\right)^{-1}
\equiv\frac{\langle \mathcal O(t)\rangle_{\rho}}{\langle \delta_{ij}\rangle_{\rho}}
\end{equation}
where $\rho_{ij}$ and $\mathcal O(t)_{ij}$ denote matrix elements of the density matrix and time-evolved operator, respectively. It is important that the basis states be cheap to prepare on the quantum processor. Eigenstates of field operators are a good choice. The notation $\langle\cdot\rangle_\rho$ denotes expectation values sampled from the distribution $\rho_{ij}$.

The normalization $\langle\delta_{ij}\rangle_\rho$ can be disregarded if we restrict ourselves to looking at ratios of expectation values. (It also happens that it can be efficiently computed; see the appendix of~\cite{Harmalkar:2020mpd}.) The distribution $\rho_{ij}$ may be efficiently sampled by the standard methods of Euclidean lattice field theory, with one important difference. Ordinarily, we are evaluating $\langle \mathcal O\rangle$ for some observable $\mathcal O$ which is diagonal in the fiducial basis used for the lattice calculation. This means that we can disregard all off-diagonal elements of the density matrix and sample only along the diagonal; this is how the periodic boundary conditions of the Euclidean lattice come into being. The operator $\mathcal O(t)$, however, does not vanish off the diagonal, and so we must sample the full density matrix.

By treating the operator $\mathcal O(t)$ as an indivisible entity, we have avoided introducing the sign problem associated with the lattice Schwinger-Keldysh method. However, the classical computer has no way of accessing the matrix elements of $\mathcal O(t)$. Fortunately, this is precisely the task for which a quantum computer is best suited. We have assumed that the basis states $\left|\Psi_i\right>$ are easily prepared on the quantum processor; this implies that we can also prepare the states
\begin{equation}
\left|+_{ij}\right> = \frac 1 {\sqrt{2}}\left(\left|\Psi_i\right>+\left|\Psi_j\right>\right)
\;\text{ and }\;
\left|-_{ij}\right> = \frac 1 {\sqrt{2}}\left(\left|\Psi_i\right>-\left|\Psi_j\right>\right)
\end{equation}
with the aid of an ancillary qubit. The expectation value of $\mathcal O(t)$ in each of these states is readily measured, and the desired matrix element is given by
\begin{equation}
\frac 1 2 \left[\mathcal O(t)_{ij} + \mathcal O(t)_{ji}\right]
=
\mathcal O(t)_{++} - \mathcal O(t)_{--}
\text.
\end{equation}
Because $\rho$ is Hermitian, this is the only linear combination needed.

This method encounters a signal-to-noise problem, which may be alleviated as detailed in~\cite{Harmalkar:2020mpd}. Whether these techniques are sufficient to make this method practical is a matter for further empirical study, which awaits the creation of intermediate-scale quantum computers on which they  can be tested.

\titleformat{\chapter}{\normalfont\large}{Appendix \thechapter:}{1em}{}
\appendix
\chapter{Grassmann Numbers}\label{ap:grassmann}
Grassmann numbers are anticommuting objects: two Grassmann numbers $\eta_1$ and $\eta_2$ obey $\eta_1 \eta_2 = \eta_2 \eta_1$. The square of a Grassmann number vanishes: $\eta_1^2 = 0$. The Grassmann algebra on $N$ Grassmann numbers is constructed by considering complex linear combinations of Grassmann numbers. An object in the Grassmann algebra of $N$ Grassmann numbers has $2^N$ complex coefficients, one for each possible combination of Grassmann numbers. Addition is performed as it would be for a vector in $\mathbb C^{2^N}$, and the multiplication rule is fixed by the fact that it distributes over addition, the anticommutativity of Grassmann numbers, and the fact that Grassmann numbers commute with complex numbers.

\section{Integration}
The Berezin integral~\cite{Berezin:1966nc} is defined by the rule
\begin{equation}
\int \d\eta (a + \eta b) = b
\text,
\end{equation}
where $a$ and $b$ are elements of the Grassmann algebra that do not contain $\eta$ itself.
Note that this integral is only a formal manipulation --- there is no sense in which it can be approximated by a limit of finite sums, as in Riemann integration. Integrations over Grassmann numbers, like the numbers themselves, anticommute.

The only form of the Berezin integral that will be relevant for physical applications discussed here is when all Grassmann variables are integrated over together. This is a linear map from the Grassmann algebra to the complex numbers, defined by taking the coefficient of the term containing all $N$ Grassmann numbers

\section{Coherent States}

The fermionic path integral is derived through the use of coherent states, defined as
\begin{equation}
\left|\psi\right> = e^{\psi^\dagger c^\dagger}\left|0\right>
\;\text{ and }\;
\left<\psi\right| = \left<0\right| e^{c \psi}
\text,
\end{equation}
where $c$ ($c^\dagger$) is the annihilation (creation) operator acting on a fermionic mode, and the state $\left|0\right>$ is the state in which that mode is unoccupied. This coherent state can be used to construct the idenitty operator on the Hilbert space of the theory:
\begin{equation}\label{eq:grassmann-identity}
I = \int \d \psi^\dagger \d \psi e^{\psi^\dagger \psi} \left|\psi\right>\left<\psi\right|
\text.
\end{equation}
The derivation of the fermionic path integral proceeds, at this point, in the usual way, with the insertion of many copies of this identity operator into the expression $\Tr e^{-\beta H}$ for the partition function.

Note that this so-called ``coherent state'' $\left|\psi\right>$ does not sit in the Hilbert space $\mathbb C^2$ of the theory at all: it is not equal to any complex linear combination of $\left|0\right>$ and $\left|1\right>$. Formally, we have allowed coefficients to be Grassmann-valued, thus extending the Hilbert space to a larger module. Crucially, the expression (\ref{eq:grassmann-identity}) is only the identity operator when acting on the original Hilbert space; it annihilates all objects in the module not part of that original space.

\renewcommand{\baselinestretch}{1}
\small\normalsize

\newpage
\bibliographystyle{unsrt}
\bibliography{References,Self,PDF}

\end{document}